\documentclass[a4paper,12pt]{article}

\usepackage[a4paper]{geometry}
\usepackage{color}
\usepackage{graphicx}
\usepackage{amsmath}
\usepackage{amsfonts}
\usepackage{amsthm}
\usepackage{amscd}
\usepackage[mathcal]{euscript}
\usepackage{epsfig}
\usepackage{amssymb}
\usepackage{slashed}
\usepackage{tabularx}
\usepackage{calligra}
\usepackage{enumerate}
\usepackage{hyperref}
\usepackage{float}
\usepackage{stackrel}
\usepackage[T1]{fontenc}
\usepackage{longtable}
\usepackage{amsthm}
\usepackage{mathrsfs}
\usepackage{verbatim} % for comments
\usepackage{fancyhdr}
\usepackage{tikz}
\usepackage{tikz-cd}
\usepackage{slashed}
\usetikzlibrary{matrix}
\usetikzlibrary{positioning}
\usepackage[all]{xy}
\usepackage{units}
\usepackage{textcomp}
\usepackage{turnstile}
\usepackage{vmargin}
\usepackage{anysize}
\usepackage{ stmaryrd }

\numberwithin{equation}{section}

\setmargins{2,5cm}       % left margin
{1,5cm}                        %  top margin
{16cm}                      %  text widht
{23,42cm}                    % text height
{10pt}                           % header height
{1cm}                           % space between text and header
{0pt}                             % footer height
{2cm}                           %  space between text and footer

\theoremstyle{plain}
\newtheorem{theorem}{Theorem}[section]
\newtheorem{proposition}[theorem]{Proposition}
\newtheorem{lemma}[theorem]{Lemma}
\newtheorem{corollary}[theorem]{Corollary}

\newtheorem{example}[theorem]{Example}
\theoremstyle{definition}
\newtheorem{definition}[theorem]{Definition}
\newtheorem{remark}[theorem]{Remark}

\newcommand{\llb}{\llbracket}
\newcommand{\rrb}{\rrbracket}
\newcommand{\la}{\langle}
\newcommand{\ra}{\rangle}
\renewcommand{\a}{\alpha}
\renewcommand{\b}{\beta}

\newcommand{\G}{\Gamma}

\newcommand{\e}{\epsilon}

\renewcommand{\d}{\delta}

\newcommand{\w}{\omega}

\newcommand{\mbb}{\mathbb}

\begin{document}

\title{Double field theory, twistors, and integrability in 4-manifolds}

\author{Bernardo Araneda\footnote{Email: \texttt{bernardo.araneda@aei.mpg.de}} \\
Max-Planck-Institut f\"ur Gravitationsphysik \\ 
(Albert Einstein Institut), Am M\"uhlenberg 1, \\
D-14476 Potsdam, Germany}

\date{June 3, 2021}

\maketitle

\begin{abstract}
The search for a geometrical understanding of dualities in string theory, in particular T-duality, has led to the development 
of modern T-duality covariant frameworks such as Double Field Theory, 
whose mathematical structure can be understood in terms of generalized geometry and, more recently, para-Hermitian geometry.
In this work we apply techniques associated to this doubled geometry to four-dimensional manifolds, 
and we show that they are particularly well-suited to the analysis of integrability in special spacetimes, 
especially in connection with Penrose's twistor theory and its applications to general relativity. 
This shows a close relationship between some of the geometrical structures in the para-Hermitian approach to double field theory and 
those in algebraically special solutions to the Einstein equations.
Particular results include the classification of four-dimensional, possibly complex-valued,
(para-)Hermitian structures in different signatures, 
the Lie and Courant algebroid structures of special spacetimes, and the analysis of deformations of (para-)complex structures.
We also discuss a notion of ``weighted algebroids'' in relation to a natural gauge freedom in the framework.
Finally, we analyse the connection with two- and three-dimensional (real and complex) twistor spaces, and 
how the former can be understood in terms of the latter, in particular in terms of twistor families.
\end{abstract}

\tableofcontents

\section{Introduction}

Generalized geometry \cite{Hitchin, Gualtieri} is a branch of differential geometry 
that unifies symplectic and complex geometry and
in which, in particular, vector fields and 1-forms on a manifold $M$ are treated on an equal footing, 
as sections of a ``doubled'' or generalized tangent bundle, $\mathbb{T}M=TM\oplus T^{*}M$. 
The generalization to $\mathbb{T}M$ of the differential structure encoded in the Lie bracket is captured by the notion of
Courant algebroids, which involve the so-called Dorfman bracket (or its antisymmetric version, the Courant bracket).
During the last years these structures have appeared naturally in theoretical high-energy physics,
as a central tool for the understanding of the geometry of string theory,
related in particular to the geometrical interpretation of T-duality (e.g. \cite{Cavalcanti}).  
A prominent line of developments in this respect is Double Field Theory (see below),
where it has recently been shown \cite{Vaisman1, Vaisman2, FRS2017, Svo18} 
that also para-Hermitian geometry is particularly well-suited to the understanding of its mathematical structure.
In this work we show that {\em doubled geometry} is also natural in general relativity,
where it provides a new perspective on `integrability issues' in special spacetimes 
and their connection with twistor theory.

\smallskip
As an initial motivation, let us give a rough argument to introduce para-Hermitian 
and generalized geometry in string theory, and their connections with T-duality.
When considering string toroidal compactifications, a distinctive feature of strings is that their extended 
nature allows them to wrap non-contractible cycles around the compact dimensions.
As a consequence, strings are not only characterized by momentum modes but also by 
{\em winding} modes, which describe how the strings wind around the tori.
Roughly speaking, T-duality is a symmetry of string theory that establishes that momentum and winding states should 
be considered on an equal footing.
As a simple example, suppose we have only one compact dimension, a circle $S^{1}$ with radius $R$. 
The momentum and winding modes are quantized; they are characterized by two quantum numbers $n,m$.
From the mass spectrum for closed strings one can see (as a standard reference cf. \cite[Section 8.3]{Polchinski})
that, when $R\to\infty$, winding modes become very heavy, while momentum 
modes become light and start to form a continuum, as corresponds to a non-compact dimension. 
Similarly, when $R\to 0$, momentum modes become infinitely massive 
but winding modes become light, and the spectrum again approaches a continuum.
Therefore, the $R\to\infty$ and $R\to 0$ limits appear to be physically equivalent, a remarkable feature that is exclusive 
to string theory since a field theory of point-like particles has no winding states (so no spacetime dimension 
`opens up' when $R\to 0$ in that case).
This equivalence comes from the fact that the string spectrum is invariant under the simultaneous change 
$R \leftrightarrow \alpha'/R$, $n \leftrightarrow m$ (where $\alpha'$ is the inverse of the string tension), a symmetry which 
extends to any observable in the theory.
This is the simplest manifestation of T-duality;
more generally, T-duality relates the physics of strings propagating on backgrounds with very different geometries.

\smallskip
Since momentum modes are  conjugate or `dual' to spacetime coordinates $x^{i}$, the consideration of momentum and winding 
modes at the same level suggests the introduction of a new set of coordinates $\tilde{x}_{i}$ `dual' to winding.
This doubling of coordinates allows to formulate a field theory that incorporates the degrees of freedom associated to 
winding without necessarily representing strings, and such that T-duality is a manifest symmetry of the theory. 
This scheme is called Double Field Theory (DFT) since the pioneering work \cite{HullZwiebach} 
(see also e.g. \cite{Aldazabal} for a review), and represents a T-duality covariant formulation of supergravity that incorporates
stringy aspects not present in the usual treatment of supergravity.
Although the doubling of space coordinates should only affect the compact dimensions, it is useful in practice to duplicate all 
space coordinates; in this way one ends up with a doubled space, which in the simplest case is a product $T\times\tilde{T}$ 
and has global coordinates $(x^{i},\tilde{x}_{i})$,
where $T$ and $\tilde{T}$ are $d$-dimensional flat tori that are said to be dual from each other.
This `extended spacetime' is intended to be the target space of the fundamental string, 
and the physical fields\footnote{In the simplest formulation of DFT, 
the physical fields correspond to the massless sector of closed strings, namely a metric tensor, 
a 2-form field (the Kalb-Ramond field) and a scalar field (the dilaton).} 
are organized in representations of the T-duality group ${\rm O}(d,d)$.
Furthermore, from the field content of the theory
one deduces a notion of generalized diffeomorphisms, that are infinitesimally generated 
by a generalized Lie derivative or `D-bracket'. 
Consistency conditions (i.e. closure of the algebra of generalized diffeomorphisms) impose a restriction on the fields.
A particular solution to this restriction is the so-called `section condition' or `strong constraint',
where the fields are forced to depend on only half of the doubled coordinates. 
This half is determined as a subspace that is maximally isotropic (or totally null) with respect to the  
${\rm O}(d,d)$ metric $\eta={\rm d}x^{i}\otimes{\rm d}\tilde{x}_{i}+{\rm d}\tilde{x}_{i}\otimes{\rm d}x^{i}$ 
\cite{HullZwiebach2009b}, i.e. a $d$-dimensional subspace $S$ 
such that\footnote{The metric $\eta$ involved in \eqref{Isotropic} does not 
represent the gravitational field; it is part of the `kinematics' of the theory. 
The gravitational field is unified with the 2-form field into a `generalized metric', which codifies the dynamics;
this will not play a role in the present work.} 
\begin{equation}\label{Isotropic}
 \eta|_{S}=0.
\end{equation}
T-duality acts by transforming any such isotropic subspace into another.

\smallskip
One might already anticipate some parallelism of the DFT framework with generalized geometry, 
in the sense that both formalisms are based on a doubled space.
However, an important difference is that in the former, what is doubled is the base manifold, while 
in the latter is the tangent bundle.
This difference is important if one wants to understand T-duality in the generalized geometry setting, 
since, as argued above, the very essence of T-duality is to exchange the manifolds $T$ and $\tilde{T}$.
Upon restriction to fields satisfying the strong constraint, the D-bracket reduces to the Dorfman bracket
and one can argue that the DFT setting reduces locally to generalized geometry (see e.g. \cite{Hohm}).
A possible formalization of all this (that is also intended to provide a global formulation of DFT)
originates in the work of Vaisman in \cite{Vaisman1, Vaisman2}, who 
shows, in the first place, that the usual extended spacetime of DFT is a flat para-K\"ahler manifold.
To see this, consider first an almost para-Hermitian structure on a $2d$-dimensional manifold $M$: 
a pair $(\eta,K)$ where $K\in{\rm End}(TM)$ and $\eta$ is a metric such that 
\begin{equation}\label{paraHermitian0}
 K^{2}=\mathbb{I}, \qquad \eta(KX,KY)=-\eta(X,Y)
\end{equation}
for any vector fields $X,Y$. The map $K$ is called almost para-complex structure, and
it produces a splitting of the tangent bundle as $TM=L\oplus\tilde{L}$, where $L, \tilde{L}$ are the $(\pm 1)$-eigenbundles of $K$. 
These eigenbundles are maximally isotropic with respect to $\eta$, i.e. they are $d$-dimensional and satisfy \eqref{Isotropic}, 
and this implies that $\tilde{L}\cong L^{*}$ (see Section \ref{Section:background} for details).
If $L$ and $\tilde{L}$ are involutive under the Lie bracket then $K$ is said to be integrable, 
and there exist local coordinates $(x^{i},\tilde{x}_{i})$, $i=1,...,d$, such that, 
denoting $\partial_{i}=\partial/\partial x^{i}$ and $\tilde{\partial}^{i}=\partial/\partial\tilde{x}_{i}$,
one has $L={\rm span}(\partial_{i})$ and $\tilde{L}={\rm span}(\tilde{\partial}^{i})$. 
If $\eta$ is flat, then it can be expressed in these coordinates as $\eta={\rm d}x^{i}\odot{\rm d}\tilde{x}_{i}$.
So there are two complementary foliations $F,\tilde{F}$ such that $L=TF$ and $\tilde{L}=T\tilde{F}$ 
\footnote{This notation means that if $\Sigma$ is a leaf in $F$, then $T\Sigma=L|_{\Sigma}$, etc.},
and $\tilde{L}\cong L^{*}$ implies that $T\tilde{F}\cong T^{*}F$, thus $TM$
has actually the structure of a generalized tangent bundle: $TM\cong TF\oplus T^{*}F$.
Furthermore, Vaisman shows \cite{Vaisman1} that there is a natural notion of D-bracket on $TM$, and that it reduces to the 
Dorfman bracket for fields with dependence only on $F$, so in this way one recovers generalized geometry.
Finally, different T-dual spacetimes correspond to different choices of para-complex structures.

\smallskip
The study of DFT from the perspective of para-Hermitian geometry, in settings more general than the flat para-K\"ahler case,
was further developed in \cite{Vaisman2} and \cite{FRS2017, Svo18, FRS19, MarottaSzabo}, and also in several subsequent works, 
being currently an active research area.
(Early considerations of the relationship between para-Hermitian structures and T-duality can be found in \cite{Hull}.)
A particularly important difference between Hermitian and para-Hermitian geometry is that the 
eigenbundles of an almost-complex structure are complex conjugates of each other, 
while those of an almost para-complex structure are independent. 
As a consequence, the integrability properties of a para-Hermitian structure can be 
split into {\em separate} questions about the involutivity of each eigenbundle, so one can talk about 
``half-integrability''.
We will see that this is especially important in general relativity and its connection with twistor structures.

\smallskip
A key condition in the analysis above is equation \eqref{paraHermitian0}, which implies \eqref{Isotropic} 
for the eigenbundles of $K$.
For real $K$, this means that the signature of $\eta$ must be $(d,d)$; so in four dimensions this is $(++--)$,
which is different from the usual Lorentz signature $(+---)$ of general relativity.
However, if we allow $K$ to be a map on the {\em complexified} tangent bundle, 
then the metric is allowed to have any signature (or to be complex).
Our interest in complex-valued maps arises from the power of using complex methods in relativity,
which date back to remarkable programmes such as Penrose's twistor theory, Newman's H-space, 
and Pleba\'nski's hyper-heavenly (HH) construction.
(One can argue that, in a sense, all these constructions can be understood as parts of the twistor programme.)
This is one of our main motivations in this work:
submanifolds where the condition \eqref{Isotropic} holds are actually the basic object of twistor theory, 
but they arise from a completely different perspective, namely the essential premise of the twistor programme that light-rays 
are more fundamental than spacetime points.
A related motivation comes from the fact that 
a ``complex-valued almost-complex structure'' has some significance in developments in mathematical relativity
and their relation to twistor structures, see \cite{Fla76, Flaherty2, Bailey, Ara18, Ara20}.
As in ordinary para-Hermitian geometry, the phenomenon of ``half-integrability'' is shared by such complex-valued maps.
The present work can be regarded as the application of some of the modern DFT-like techniques mentioned above 
to 4-manifolds, with the extra ingredient that $K$ can be complex-valued so that one can make contact with 
general relativity and twistor theory.

\subsection*{Main results}

The fact that we consider complex-valued ``para-Hermitian'' structures implies that, 
unlike the usual literature in DFT and related matters, we can analyse different signatures of the metric, 
by imposing different reality conditions.
Since we work in a 4-manifold $M$ (which we assume to be orientable), this means that we will deal with
Lorentz, Riemannian (or Euclidean) and split (or neutral) signature.
The Lorentzian case is relevant, of course, for general relativity (and in particular for 
the so-called hyper-heavenly spaces \cite{PlebanskiRobinson}); 
the Riemannian case is relevant because e.g. of the Atiyah-Hitchin-Singer approach to twistor theory \cite{Atiyah}; 
and the split case can be related, since the work of Ooguri and Vafa \cite{Ooguri},
to the geometry of strings with $N=2$ supersymmetry. (It is also related to the LeBrun-Mason twistor construction 
in split signature \cite{LebrunMason}.)

\smallskip
Since an arbitrary 4-manifold with a metric does not naturally come equipped with a (possibly complex-valued)
almost para-Hermitian structure, we first focus on classifying all possible such structures. 
We show that regardless of signature, an almost para-Hermitian structure 
is essentially equivalent to an (anti-)self-dual 2-form (in the sense of Hodge duality). 
Then we use this result to describe the space $P_{x}$ of para-Hermitian structures at any point $x\in M$:
we show that $P_{x}$ has two connected components and, imposing appropriate reality conditions whenever possible,
each component can be described as:
a complex projective line $\mathbb{CP}^{1}$ in Riemannian signature\footnote{This turns out to be analogous
to the result on almost-complex structures, which is already known in the literature \cite{Atiyah}.}, 
the space $(\mathbb{RP}^{1}\times\mathbb{RP}^{1})\backslash\mathbb{RP}^{1}$ in split signature,
and the space $(\mathbb{CP}^{1}\times\mathbb{CP}^{1})\backslash\mathbb{CP}^{1}$ in Lorentz signature 
(and in complex 4-manifolds).
We show that, topologically, these spaces are respectively 
a real 2-sphere $S^{2}$, a hyperboloid of one sheet $H^{1}$, and a complex 2-sphere $\mathbb{C}S^{2}$.
For {\em Hermitian} structures, the corresponding spaces are $\mathbb{CP}^{1}$, 
$\mathbb{CP}^{1}\backslash\mathbb{RP}^{1}$ (a hyperboloid of {\em two} sheets)
and $(\mathbb{CP}^{1}\times\mathbb{CP}^{1})\backslash\mathbb{CP}^{1}$, respectively.
We also find that there is a natural ``gauge freedom'' associated to the representation of a (para-)Hermitian structure, 
and we develop a correspondingly covariant formalism.

\smallskip
We then analyse the integrability conditions for almost (para-)Hermitian structures, and we show that involutivity of one of  
the eigenbundles of $K$ (or ``half-integrability'') is equivalent 
to the existence of certain special spinors, that in relativity language give origin to what are called 
shear-free null geodesic congruences.
This allows us to show that there are natural algebroid structures associated
to any four-dimensional manifold with a shear-free congruence, namely Lie and Courant algebroids.
This gives explicit examples of Courant algebroids in four dimensions, 
that include e.g. all algebraically special Einstein manifolds (such as the Kerr and Schwarzschild black hole solutions, 
but also all vacuum Petrov type II solutions), 
so it connects the structures studied in generalized geometry to spacetimes of interest in 
classical relativity.
Furthermore, the Lie algebroid structure has naturally associated a cochain complex that 
can be used to prove the existence of potentials in problems of interest in relativity.
We are naturally led to the question of generalizing the algebroid structure to fields with general transformation 
properties under the gauge freedom mentioned above, and we discuss the possible construction 
of such an object (``weighted algebroid'') and the associated differential complex.

\smallskip
Finally, we focus on connections with twistor theory, showing that:
any particular half-integrable para-Hermitian structure defines a 2-dimensional (2D) twistor space, 
and if all para-Hermitian structures are half-integrable, then there is 
a 3-dimensional (3D) twistor space, which is a one-parameter family of 2D twistor spaces 
parametrized by projective spinor fields. This 3D space is the total space of
a fibration of 2D twistor spaces over projective spinors if and only if the vacuum Einstein equations are satisfied.
We also discuss the relationships with other twistor constructions in the literature.
Finally we define deformations of para-complex structures in $M$ and analyse their integrability, and we show that 
small, half-integrable deformations exist if and only if the Weyl tensor is half-algebraically special, and in this 
case they define sections of a line bundle over a 2D twistor space.
As in the rest of the paper, our results here are valid for any metric signature.

\subsection*{Conventions and overview}

Our conventions for different kinds of indices, curvature, etc. follow Penrose and Rindler \cite{PR1,PR2}.
The organization of this work is as follows.
We start in Section \ref{Section:background} with some background material: we review basic notions on
para-Hermitian geometry and related structures (keeping always in mind that we allow complex-valued maps), 
generalized geometry and Lie and Courant algebroids, and conformal geometry and the Lee form.
Then in Section \ref{Sec:4D} we study almost para-Hermitian structures in 4-manifolds equipped with a metric 
of different signatures, and we describe the spaces of such structures in each case.
In Section \ref{Sec:Integrability} we study integrability issues: half-integrability of para-Hermitian structures,
the Lie and Courant algebroids associated to special spacetimes,
and a possible notion of ``algebroids'' for the treatment of weighted fields.
In Section \ref{Sec:Twistors} we study some connections with twistor theory: 
two- and three-dimensional twistor spaces, twistor families, and
integrability of small deformations of para-complex structures.
In Section \ref{Sec:Conclusions} we make a summary of this work and comment on some possible future directions.
We include three appendices: Appendix \ref{Appendix:spinors} with a brief review of spinors in 4 dimensions 
(in any signature), Appendix \ref{Appendix:CovariantFormalism} with additional details about the 
``gauge freedom'' and the associated covariant formalism, and Appendix \ref{Appendix:curvature}
with some properties of the curvature of the natural connection in the above formalism.

\section{Background}\label{Section:background}

\subsection{Para-Hermitian and related structures}\label{Sec:paraHermitian}

In the following, $M$ denotes a $d$-dimensional real, smooth manifold. 
The definitions still apply if $M$ is complex.

\begin{definition}\label{Def:paraHermitianStructure}
Let $E\to M$ be a vector bundle with even rank. 
Let $g$ be  a non-degenerate symmetric bilinear form on the fibers.
Given a map $K\in{\rm End}(E)$, we say that $(g,K)$ is a {\em para-Hermitian structure} on $E$ if it holds:
\begin{enumerate}
 \item $K^2=+\mbb{I}$
 \item The $\pm1$ eigenbundles of $K$ have the same rank.
 \item $g(KX,KY)=-g(X,Y)$ for all $X,Y\in E$. 
\end{enumerate}
If we require only the first two conditions, $K$ is said to be a {\em para-complex structure}.
\end{definition}

Consider a para-Hermitian structure $(g,K)$.
The map $K$ gives a decomposition $E=L\oplus\tilde{L}$, where 
$L$ is the $(+1)$-eigenbundle of $K$ and $\tilde{L}$ the $(-1)$-eigenbundle.
Any element $X\in E$ can be written as 
\begin{equation}\label{splitvector0}
 X=x+\tilde{x}
\end{equation}
where $x\in L$ and $\tilde{x}\in\tilde{L}$. In the following, we will extensively use this notational convention.
Note that if $x,y\in L$, i.e. $Kx=x$ and $Ky=y$, then $g(x,y)=g(Kx,Ky)=-g(x,y)$ so $g(x,y)\equiv 0$
for all $x,y \in L$, and similarly $g(\tilde{x},\tilde{y})\equiv0$ for all $\tilde{x},\tilde{y}\in\tilde{L}$. 
We then say that $L$ and $\tilde{L}$ are isotropic with respect to $g$, i.e. $g|_{L}=0=g|_{\tilde{L}}$.
The following is a standard result:

\begin{proposition}
We have the isomorphisms $\tilde{L}\cong L^{*}$ and $L\cong \tilde{L}^{*}$.
\end{proposition}

\begin{proof}
If $X\in E = L\oplus\tilde{L}$, in general we have $g(X,\cdot)\in L^{*}\oplus\tilde{L}^{*}$.
But if $\tilde{x}\in \tilde{L}$, then $g(\tilde{x},\tilde{y})=0$ for all $\tilde{y}\in \tilde{L}$ since $\tilde{L}$ is 
isotropic, therefore $g(\tilde{x},\cdot)\in L^{*}$. 
Thus we have a map $\tilde{L}\to L^{*}$ given by 
\begin{equation}\label{isom0}
 \tilde{L}\ni\tilde{x} \mapsto \tilde{x}_{\flat}:=g(\tilde{x},\cdot) \in L^{*},
\end{equation}
and since $g$ is non-degenerate, the map is an isomorphism. The proof of $L\cong \tilde{L}^{*}$ is analogous.
\end{proof}

\noindent
Since here we are not using the abstract index notation, we use the musical isomorphisms $\flat$ and $\sharp$.
From \eqref{isom0} it follows that the map
\begin{equation}
 \varphi:E=L\oplus\tilde{L}\to L\oplus L^{*}, 
 \qquad X=x+\tilde{x}\mapsto \varphi(X)=(x,\tilde{x}_{\flat}) \label{isom}
\end{equation}
is also an isomorphism. 

\begin{definition}\label{Def:almostParaHermitian}
We say that $(g,K)$ is an {\em almost para-Hermitian structure} on a manifold $M$
if $(g,K)$ is a para-Hermitian structure on $TM\otimes\mathbb{C}$.
\end{definition}

\begin{remark}
Notice that the map $K$ in Def. \ref{Def:almostParaHermitian} satisfies the conditions of 
Def. \ref{Def:paraHermitianStructure} and is allowed to be {\em complex-valued}.
From this perspective, we can equally well refer to the map $J=iK$ as a ``Hermitian structure'' in the sense 
that it satisfies $J^{2}=-\mathbb{I}$, its eigenbundles have equal rank, and $g(J\cdot,J\cdot)=g(\cdot,\cdot)$.
We choose the terminology `para-Hermitian' because it suggests that the 
eigenbundles are not related to each other
(unlike in the usual real-valued Hermitian structures), see Remark \ref{Remark:integrability} below.
See also Remark \ref{Remark:reality}.
\end{remark}

\begin{remark}\label{Remark:Conformal}
If $(g,K)$ is an almost para-Hermitian structure on $M$, and $\Omega$ is a nowhere vanishing scalar field,
then $(\Omega^{2}g,K)$ is also an almost para-Hermitian structure on $M$, since
$(\Omega^{2}g)(K\cdot,K\cdot)=-(\Omega^{2}g)(\cdot,\cdot)$.
This means that we can interpret $K$ as an object associated not to a particular metric $g$ but to the 
conformal class $[g]=\{\Omega^{2}g \;|\; \Omega\in C^{\infty}(M), \Omega > 0 \}$, in the sense that 
the conditions of definitions \ref{Def:almostParaHermitian} and \ref{Def:paraHermitianStructure} are satisfied 
for any $g\in[g]$.
Therefore, conformal invariance will play an important role in our study. 
We will refer to $([g],K)$ as an {\em almost para-Hermitian conformal structure} on $M$.
\end{remark}

A para-Hermitian structure $(g,K)$ in $M$ can also be thought of as an ``almost-symplectic structure'', 
in the sense that it automatically defines a non-degenerate 2-form $\w(X,Y):=g(KX,Y)$;
this is usually called the fundamental 2-form.  
The relationship of this with conformal geometry is particularly interesting, see Section \ref{Sec:Conformal} below.
The eigenbundles of $K$ are isotropic with respect to $\w$, that is $\w|_{L}=0$ and $\w|_{\tilde{L}}=0$.

\begin{definition}\label{Def:almostparaKahler}
An {\em almost para-K\"ahler structure} on a manifold $M$ is an almost para-Hermitian structure 
such that the fundamental 2-form is closed, ${\rm d}\w=0$.
\end{definition}

Although our primary interest is in the case where there is only one para-Hermitian structure,
in later sections it will appear naturally the case in which one has more than one para-complex or complex structure.
Because of this, it is useful to recall now the algebras of quaternions $\mathbb{H}$ and of 
para-quaternions (or split-quaternions) $\mathbb{H}'$:
\begin{align}
 \mathbb{H} ={}& \{ q = a+bi+cj+dk \;|\; i^2=j^2=k^2=-1, \; k=ij=-ji \}, \label{quaternions} \\
 \mathbb{H'} ={}& \{ q' = a+bi+cs+dt \;|\; i^2=-1, \; s^2=t^2=1, \; t=is=-si \} \label{splitquaternions}
\end{align}
where $a,b,c,d \in \mathbb{R}$. The quaternion algebra $\mathbb{H}$ is associated to a positive definite 
quadratic form\footnote{The conjugate of a quaternion $q$ is $\bar{q}=a-bi-cj-dk$, and the conjugate of 
a para-quaternion $q'$ is $\bar{q}'=a-bi-cs-dt$.} 
$q\bar{q}=a^2+b^2+c^2+d^2$, while $\mathbb{H}'$ is associated to a quadratic form with 
split signature $q'\bar{q}'=a^2+b^2-c^2-d^2$.
Based on these algebras, we define:
\begin{definition}\label{Def:hypercomplex}
Let $M$ be a real manifold, and let $I_{1},I_{2},I_{3}$ be three anticommuting endomorphisms of $TM$. Then:
\begin{enumerate}
\item $(I_{1},I_{2},I_{3})$ is an {\em almost-hypercomplex} (or almost-quaternionic) structure 
if $I_{1},I_{2},I_{3}$ satisfy the quaternion algebra \eqref{quaternions}.
If $M$ is equipped with a metric $g$, then $(g,I_{1},I_{2},I_{3})$ is an {\em almost-hyperhermitian} structure on $M$ if
it is almost-hypercomplex and $g(I_{i}X,I_{i}Y)=g(X,Y)$ for all $i=1,2,3$ and for all vectors $X,Y$.
\item $(I_{1},I_{2},I_{3})$ is an {\em almost-para-hypercomplex} (or almost-para-quaternionic)
structure if $I_{1},I_{2},I_{3}$ satisfy the para-quaternion algebra \eqref{splitquaternions}.
If $M$ is equipped with a metric $g$, then $(g,I_{1},I_{2},I_{3})$ is an {\em almost-para-hyperhermitian} structure on $M$ if 
it is almost-para-hypercomplex and $g(I_{1}X,I_{1}Y) = g(X,Y)$ and $g(I_{i}X,I_{i}Y)=-g(X,Y)$ for $i=2,3$.
\end{enumerate}
\end{definition}

The word `almost' in all the definitions above can be removed by introducing a notion of integrability.
To do this, it is instructive to first briefly recall the situation in complex geometry.
An almost-complex structure in a real manifold $M$ is a real map $J$ in $TM$ such that 
$J^{2}=-\mathbb{I}$, and whose $\pm i$ eigenbundles, denoted $T^{\pm}M$, have the same rank.
One has the decomposition $TM\otimes\mathbb{C}=T^{+}M\oplus T^{-}M$ 
(note that the elements in $T^{\pm}M$ are necessarily complex).
An almost-complex structure is said to be {\em integrable} if $T^{+}M$ is an involutive distribution in $TM\otimes\mathbb{C}$, 
that is $[\Gamma(T^{+}M),\Gamma(T^{+}M)]\subset \Gamma(T^{+}M)$, where $[\cdot,\cdot]$ is the Lie bracket of vector fields, 
and $\Gamma(E)$ denotes the space of sections of a vector bundle $E$.
A real manifold equipped with an integrable almost-complex structure is a complex manifold.
The integrability condition for $J$ is equivalent to the vanishing of its associated Nijenhuis tensor, where,
if $A:\G(TM)\to\G(TM)$ is a linear map, the Nijenhuis tensor associated to $A$ is the map $N_A:\G(TM)\times\G(TM)\to\G(TM)$ given by
\begin{equation}
 N_A(X,Y)=\frac{1}{4}\left( A^2[X,Y]+[AX,AY]-A([AX,Y]+[X,AY]) \right), \label{Nij}
\end{equation}
for all vector fields $X,Y$.

Mimicking the above definitions for complex structures, we define
\begin{definition}
A {\em (para-)Hermitian structure} on a manifold $M$ is an almost (para-)Hermitian structure $(g,K)$ such that 
the Nijenhuis tensor \eqref{Nij} associated to $K$ vanishes:
\begin{equation}
 N_{K} \equiv 0.
\end{equation}
Similarly, a {\em (para-)hyperhermitian structure} on $M$ is an almost-(para-)hyperhermitian structure $(g,I_1,I_2,I_3)$
such that $N_{I_{i}}\equiv 0$ for all $i=1,2,3$.
\end{definition}

A crucial difference between integrability of ordinary Hermitian and para-Hermitian structures is the following:
\begin{remark}\label{Remark:integrability}
The fact that an ordinary almost-complex structure $J$ is a real tensor implies that its eigenbundles 
are complex conjugates of each other, so $J$ is integrable if and only if {\em both} distributions $T^{+}M$ 
and $T^{-}M$ are involutive.
In the para-Hermitian case, the two eigenbundles of a real-valued $K$ are {\em not} complex conjugates of each 
other, so one of the distributions may be involutive while the other one is not. 
This also applies if $K$ is complex-valued as in Def. \ref{Def:almostParaHermitian}.
\end{remark}

Recalling the splitting $L\oplus\tilde{L}$ induced by $K$, 
one can then separate integrability properties associated to the two eigenbundles $L$ and $\tilde{L}$.
This gives origin to a notion of ``half-integrability'':

\begin{definition}\label{Def:LparaHermitian}
We say that an almost para-Hermitian structure $(g,K)$ on $M$ is {\em half-integrable} if one of the 
eigenbundles of $K$ is involutive. 
If both eigenbundles are involutive, then $K$ is integrable and $(g,K)$ is a para-Hermitian structure on $M$.
\end{definition}

\begin{remark}\label{Remark:LparaHermitian}
In \cite{FRS2017}, half-integrable almost-para Hermitian structures on a manifold $M$ (with real-valued $K$)
are called {\em $L$-para Hermitian} or {\em $\tilde{L}$-para Hermitian} manifolds, 
depending on whether the $(+1)$- or the $(-1)$-eigenbundle is integrable, respectively.
\end{remark}

It is useful to introduce the projectors to $L$ and $\tilde{L}$; respectively:
\begin{equation}
 P:=\tfrac{1}{2}(\mbb{I}+K), \qquad \tilde{P}:=\tfrac{1}{2}(\mbb{I}-K). \label{projectors}
\end{equation}
Using $K^2=\mathbb{I}$, one can easily check that
these operators satisfy $P^2=P$, $\tilde{P}^2=\tilde{P}$ and $P\tilde{P}=0=\tilde{P}P$.
The Nijenhuis tensor \eqref{Nij} for $A\equiv K$ can be rewritten in terms of these projectors as 
\cite[Eq. (3.13)]{FRS2017}
\begin{equation}
 N_K(X,Y)\equiv N_{P}(X,Y)+N_{\tilde{P}}(X,Y),
\end{equation}
where 
\begin{equation}
 N_{P}(X,Y)=\tilde{P} [PX,PY], \qquad N_{\tilde{P}}(X,Y)=P[\tilde{P}X,\tilde{P}Y]
\end{equation}
Note that if $L$ is integrable, then $[PX,PY]\in \G(L)$ and therefore $\tilde{P}[PX,PY]=0$, 
i.e. $N_{P}\equiv 0$.
Hence $N_{P}$ and $N_{\tilde{P}}$ govern the involutivity of $L$ and $\tilde{L}$ respectively.

\begin{remark}\label{Remark:Foliation}
By Frobenius theorem (see e.g. \cite[Theorem 19.21]{Lee-manifolds}), if $M$ is a real manifold and
one has an involutive distribution $L\subset TM$ then the collection of all integral 
manifolds\footnote{An integral manifold of a distribution $L\subset TM$ is an immersed submanifold $N\subset M$ such 
that $T_{p}N=L|_{p}$ for all $p\in N$. A generic distribution does not admit integral manifolds.}
of $L$ forms a foliation of $M$. 
But if the involutive distribution is complex, i.e. $L\subset TM\otimes\mathbb{C}$, then the 
integral manifolds are {\em complex submanifolds} living in the {\em complexification} of $M$ 
\footnote{Note that, given a real manifold $M$, the complex manifold resulting from an integrable almost-complex structure, 
and the complexification of $M$, are two different concepts.}. 
Since complexification requires real-analyticity of $M$ (which is a restrictive property from 
the point of view of relativity), in general we will not assume the existence of a foliation. 
The involutivity of the distribution, however, is well-defined, and is what we mostly need in this work.
\end{remark}

\begin{remark}\label{Remark:Lagrangian}
Since the eigenbundles $L$ and $\tilde{L}$ of $K$ are isotropic with respect to the almost-symplectic form 
$\omega$, we can refer to them as {\em Lagrangian subbundles}. 
If, say, $L$ is integrable and gives origin to a foliation $F$, we can refer to $F$ as a {\em Lagrangian foliation}.
\end{remark}

The integrability of almost-para-Hermitian and related structures is a conformally invariant property. 
Other properties that are not conformally invariant are also of interest, in particular the K\"ahler and 
related conditions:

\begin{definition}\label{Def:paraKahler}
Let $M$ be a real manifold.
\begin{enumerate}
\item An almost-para-Hermitian structure $(g,K)$ on $M$ is {\em para-K\"ahler} if it is para-Hermitian and almost-para-K\"ahler, 
that is, $K$ is integrable and ${\rm d}\omega=0$.
\item An almost-para-hyperhermitian structure $(g,I_1,I_2,I_3)$ on $M$ is {\em para-hyperk\"ahler} 
if it is para-hyperhermitian and ${\rm d}\omega_{i}=0$, where $\omega_{i}(\cdot,\cdot)=g(I_{i}\cdot,\cdot)$, $i=1,2,3$.
\end{enumerate}
\end{definition}

One can show the following:

\begin{proposition}
Let $(g,K)$ be an almost-para-Hermitian structure on a real manifold $M$. Then $(g,K)$ is para-K\"ahler if and only if $K$ is 
parallel with respect to the Levi-Civita connection of $g$. 
Likewise, an almost-para-hyperhermitian structure $(g,I_1,I_2,I_3)$ on $M$ is para-hyperk\"ahler if and only if 
$I_i$ is parallel w.r.t. the Levi-Civita connection of $g$ for all $i=1,2,3$.
\end{proposition}

\subsection{Generalized geometry and algebroids}

\subsubsection{Elementary notions}\label{Sec:basicGG}

Let $V$ be a $d$-dimensional vector space and $V^{*}$ its dual, and consider the space $V\oplus V^{*}$. 
We will denote elements of $V\oplus V^{*}$ by $(X,\a), (Y,\b)$, etc.
There is a natural inner product $\la \cdot,\cdot \ra$ on $V\oplus V^{*}$ given by
$ \la (X,\a), (Y,\b) \ra = \a(Y)+\b(X)$
\footnote{A factor of $1/2$ is 
often included in the right hand side of this equation, but for our purposes this is not important.}.
Any orthogonal endomorphism, i.e. any elemeny of $\mathfrak{so}(V\oplus V^{*})$, can be written as
\begin{equation*}
 \left( \begin{matrix} A & \beta \\ B & -A^{*} \end{matrix} \right)
\end{equation*}
where $A\in{\rm End}(V)$, $B$ can be viewed as a 2-form in $\wedge^{2}V^{*}$ and $\beta$ as a bivector in $\wedge^{2}V$.
Exponentiating, one gets elements in ${\rm SO}(V\oplus V^{*})$, that can be separated into
\begin{equation}\label{OrthogonalTransformations}
 \exp(B)= \left( \begin{matrix} 1 & 0 \\ B & 1  \end{matrix} \right), \quad 
 \exp(\beta)= \left( \begin{matrix} 1 & \beta \\ 0 & 1  \end{matrix} \right), \quad  
 \exp(A)= \left( \begin{matrix} \exp A & 0 \\ 0 & (\exp A^{*})^{-1}  \end{matrix} \right).
\end{equation}

\begin{remark}\label{Remark:Shear}
The matrix $\exp(B)$ in \eqref{OrthogonalTransformations} is called a {\em B-transformation}, 
and it can be thought of as a {\em shear} transformation in the sense that it fixes projections to $V$ 
while shearing in the direction of $V^{*}$; see \cite[Chapter 2]{Gualtieri}.
This interpretation and terminology will be particularly well-suited to our purposes in later sections.
Analogous considerations apply to the {\em $\beta$-transformation} given by the matrix $\exp(\beta)$.
\end{remark}

While the above is just linear algebra and can be done for any vector space, in Generalized Geometry one applies this 
to the case where $V$ is the tangent space to a point in a $d$-dimensional manifold $M$, 
and the {\em generalized tangent bundle} is defined as $\mathbb{T}M:=TM\oplus T^{*}M$.
The construction above then describes the pointwise structure of $\mathbb{T}M$. 
The differential structure is described by generalizing the notion of Lie bracket $[\cdot,\cdot]$ of vector fields to a bracket 
operation on sections of $\mathbb{T}M$. There are two different bracket operations considered in the 
literature, the {\em Courant} bracket and the {\em Dorfman} bracket.
Here we choose the Dorfman bracket and denote it by $\llbracket \cdot,\cdot \rrbracket$; it is defined by
\begin{equation*}
 \llbracket (X,\a), (Y,\b) \rrbracket = ([X,Y], \pounds_{X}\b - \pounds_{Y}\a + {\rm d}(i_{Y}\alpha) )
\end{equation*}
where $\pounds$ denotes the Lie derivative and $i_{Y}$ the interior product. 
The Courant bracket is the skew-symmetrization of the Dorfman bracket.
While the ordinary Lie bracket is skew-symmetric and satisfies the Jacobi identity, the Dorfman bracket 
is not skew-symmetric, but it satisfies the Jacobi identity. In turn, the Courant bracket is skew-symmetric, 
but it does not satisfy the Jacobi identity.

\subsubsection{Lie and Courant algebroids}

In the literature one frequently encounters {\em real} Lie algebroids, but in this work we need the 
complexified versions (although the manifold is still real). This is also used in e.g. \cite{Gualtieri}, see Chapter 3 therein.

\begin{definition}\label{Def:LieAlgebroid}
A (complex) {\em Lie algebroid} is a triple $(L,[\cdot,\cdot]_{L},\rho)$ where $L\to M$ is a complex vector bundle, 
$[\cdot,\cdot]_{L}:\Gamma(L)\times\Gamma(L) \to \Gamma(L)$ is a bilinear map, 
and $\rho:L\to TM\otimes\mathbb{C}$ is a bundle map called {\em anchor}, such that, 
for all $X,Y,Z\in\G(L)$ and $f\in C^{\infty}(M)$, the following four conditions are satisfied\footnote{The anchor $\rho$ 
extends to a map between sections $\Gamma(L)\to\Gamma(TM\otimes\mathbb{C})$ that we also denote by $\rho$.}:
\begin{enumerate}
 \item $[X,Y]_{L}=-[Y,X]_{L}$.
 \item ${\rm Jac}_{L}(X,Y,Z) := [ X , [Y,Z]_{L} ]_{L}+[ Z, [X,Y]_{L} ]_{L}+[ Y, [Z,X]_{L} ]_{L}=0$.
 \item $\rho([X,Y]_{L})=[\rho(X),\rho(Y)]$, where $[\cdot,\cdot]$ is the Lie bracket of vector fields.
 \item $ [X,fY]_{L}=(\rho(X)f)Y+f[X,Y]_{L}$.
\end{enumerate}
\end{definition}

From the first item we see that $[\cdot,\cdot]_{L}$ is skew-symmetric, and the second means that the 
{\em Jacobiator} for $[\cdot,\cdot]_{L}$ vanishes, or, in other words, the
Jacobi identity for $[\cdot,\cdot]_{L}$ is satisfied. 
Together, these two items imply that $[\cdot,\cdot]_{L}$ is a Lie bracket.
The third item means that $\rho$ is a morphism, and from the fourth item 
we see that the anchor and the bracket are subject to the Leibniz rule.

The following are standard examples of Lie algebroids, taken from \cite[Chapter 3]{Gualtieri}:
\begin{example}
The tangent bundle $TM$ gives origin to the {\em tangent Lie algebroid} $(TM,[\cdot,\cdot],\mathbb{I})$,
where the bracket is the Lie bracket of vector fields and the anchor is the identity map.
\end{example}
\begin{example}\label{Ex:LieAlgFoliation}
If $L\subset TM$ is an involutive distribution, then $(L,[\cdot,\cdot]_L,\mbb{I}_L)$ is a Lie algebroid, 
with $[\cdot,\cdot]_L$ and $\mbb{I}_L$ the restrictions of, respectively, the Lie bracket and the identity map 
of $TM$ to the subbundle $L$.
\end{example}

We will need some objects that can be naturally constructed from the structures in a Lie algebroid:
\begin{definition}[Def. 3.7 and 3.8 in \cite{Gualtieri}]\label{Def:LieAlgDerivatives}
Let $(L,[\cdot,\cdot]_{L},\rho)$ be a Lie algebroid, let $\Lambda^{k}=\wedge^{k}L^{*}$ for $k=0,1,2,...$, 
and let $\omega\in\Gamma(\Lambda^{k})$ and $X_0,...,X_k \in \Gamma(L)$.
\begin{enumerate}
\item The {\em Lie algebroid exterior derivative} is the map ${\rm d}^{L}: \Gamma(\Lambda^{k}) \to \Gamma(\Lambda^{k+1})$
defined by 
\begin{align}
\nonumber ({\rm d}^{L}\omega)(X_0,...,X_k) ={}& \sum_{i}(-1)^{i}\rho(X_{i})\omega(X_{0},...,\hat{X}_{i},...,X_{k}) \\
  &+\sum_{i<j}(-1)^{i+j}\omega([X_i,X_j]_{L}, X_{0}, ..., \hat{X}_{i},..., \hat{X}_{j},..., X_{k}) \label{LieAlgExtDer}
\end{align}
where the notation $\hat{X}_{i}$ means that the vector $X_{i}$ is omitted.
\item The {\em generalized Lie derivative} $\pounds^{L}$ of $\omega$ along $X\in\Gamma(L)$
is the operator defined in terms of ${\rm d}^{L}$ by Cartan's formula
\begin{equation}
 \pounds^{L}_{X} \omega = i_{X}{\rm d}^{L}\omega+{\rm d}^{L}i_{X}\omega \label{GeneralizedLieDerivative}
\end{equation}
where $i_X$ denotes interior product, $i_{X}\omega = \omega(X,...)$.
\end{enumerate}
\end{definition}

As examples of the exterior derivative \eqref{LieAlgExtDer}, if $f\in\Gamma(\Lambda^{0})$ and $\omega\in\Gamma(\Lambda^{1})$ then
\begin{align}
 ({\rm d}^{L}f)(X) ={}& \rho(X)f, \label{LieAlgExtDer0} \\
 ({\rm d}^{L} \omega)(X,Y) ={}& \rho(X)\omega(Y)-\rho(Y)\omega(X)-\omega([X,Y]_{L}). \label{LieAlgExtDer1} 
\end{align}

\begin{remark}[de Rham complexes]\label{Remark:deRham}
The fact that the bracket of a Lie algebroid satisfies the Jacobi identity, together with the morphism property 
of the anchor, imply that the Lie algebroid exterior derivative ${\rm d}^{L}$ satisfies ${\rm d}^{L}\circ{\rm d}^{L}=0$. 
Therefore, $(\Gamma(\Lambda^{\bullet}),{\rm d}^{L})$ is a cochain complex (see e.g. \cite[Chapter 3]{Gualtieri}). 
This will be particularly important in sections \ref{Sec:algebroids4D}, \ref{Sec:algebroidsweight}.
\end{remark}

\begin{definition}[Def. 3.1 in \cite{Vaisman2}]\label{Def:CourantAlgebroid}
A {\em Courant algebroid} is a quadruple $(E, \langle\cdot,\cdot\rangle, \rho, \llbracket\cdot,\cdot\rrbracket)$
where: 
$E\to M$ is a vector bundle, 
$\langle\cdot,\cdot\rangle:\G(E)\times\G(E)\to C^{\infty}(M)$ is a non-degenerate symmetric bilinear form, 
$\rho:E\to TM$ is a bundle map called {\em anchor}, and
$\llbracket\cdot,\cdot\rrbracket : \G(E)\times\G(E)\to \G(E)$ is a bilinear operation called {\em Dorfman bracket}, 
subject to the following axioms for all $X,Y,Z\in\G(E)$:
\begin{enumerate}
 \item $\rho(X) \la Y,Z \ra = \la \llbracket X,Y \rrbracket, Z \ra + \la Y,\llbracket X,Z \rrbracket \ra $
 \item $\la \llbracket X,X \rrbracket,Y \ra = \tfrac{1}{2}\rho(Y)\la X,X \ra$ 
 \item $\llbracket X, \llbracket Y,Z \rrbracket \rrbracket = \llbracket \llbracket X,Y \rrbracket, Z\rrbracket +
 \llbracket Y, \llbracket X,Z \rrbracket \rrbracket$
\end{enumerate}
\end{definition}

In terms of the Dorfman bracket, the axioms in Definition \ref{Def:CourantAlgebroid} have the following meaning:
the first axiom means invariance of the inner product with respect to the bracket;
from the second axiom we see that the bracket is not skew-symmetric;
and the third axiom means that the bracket satisfies the Jacobi identity.

The definition of a Courant algebroid can be given using two different bracket operations: 
the Courant bracket $\llbracket\cdot,\cdot\rrbracket_{\rm Cour}$ or the Dorfman bracket $\llbracket\cdot,\cdot\rrbracket$. 
The difference is that one can require the bracket either to be skew-symmetric (Courant) or to satisfy 
the Jacobi identity (Dorfman), but not both. 
The original definition, introduced in \cite{ZWX}, is in terms of the Courant bracket.
The two brackets are related by
$\llbracket X,Y \rrbracket_{\rm Cour}=\tfrac{1}{2}(\llbracket X,Y\rrbracket -\llbracket Y,X\rrbracket)$.

\begin{remark}\label{rem-ma}
If only the first and second axioms in Definition \ref{Def:CourantAlgebroid} are required,
the bracket is said to be {\em metric-compatible}.
The quadruple $(E, \langle\cdot,\cdot\rangle, \rho, \llbracket\cdot,\cdot\rrbracket)$
is then called {\em metric algebroid} \cite{Vaisman1}.
\end{remark}

From the axioms in Definition \ref{Def:CourantAlgebroid} one can deduce two additional identities:
\begin{align}
 \llbracket X, f Y \rrbracket={}& (\rho(X) f ) Y+f\llbracket X,Y \rrbracket, \label{IdentityCourant1} \\
 \rho(\llbracket X,Y \rrbracket)={}& [\rho(X),\rho(Y)] \label{IdentityCourant2}
\end{align}
for any $f\in C^{\infty}(M)$. The first identity means that the anchor and the bracket satisfy the Leibniz rule; 
the second identity means that $\rho$ is a morphism of bundles.
These identities are sometimes included as part of the definition of a Courant algebroid, but it is known that they 
can  be obtained from the axioms: to prove \eqref{IdentityCourant1}, compute $\la \llbracket X, fY \rrbracket, Z \ra $ 
for $Z\in\Gamma(E)$ arbitrary and use the first axiom;
to prove \eqref{IdentityCourant2}, compute both sides independently applied to $\la Z,W \ra$ for arbitrary $Z,W\in\Gamma(E)$,
and use the first and third axioms.

As shown in \cite{ZWX}, a natural example of a Courant algebroid can be obtained from a Lie bialgebroid. 
A related example, which is key for our purposes, is the following: 

\begin{proposition}[See e.g. Example 2.6 in \cite{Svo18}]\label{Prop:CourantAlgebroid}
Let $(L, [\cdot,\cdot]_{L}, \rho_L)$ be a  Lie algebroid, and consider the vector bundle $L\oplus L^{*}$.
Denote sections of $L\oplus L^{*}$ by $(X,\a), (Y,\b)$ etc., and introduce
the following maps: 
\begin{align}
 & \pi_{L\oplus L^{*}}((X,\a))=X \label{anch}, \\
 & \la (X,\a),(Y,\b)\ra_{L\oplus L^{*}} = \a(Y)+\b(X), \label{innerp} \\
 & \llbracket (X,\a), (Y,\b) \rrbracket_{L\oplus L^{*}} = ([X,Y]_L, \pounds^L_X\b-\pounds^L_Y\a+{\rm d}^L(i_{Y}\alpha)) \label{Dorf}
\end{align}
where ${\rm d}^L$ and $\pounds^{L}$ were defined in Definition \ref{Def:LieAlgDerivatives}. 
Then the quadruple 
\begin{equation}
 (L\oplus L^{*}, \langle\cdot,\cdot\rangle_{L\oplus L^{*}}, \pi_{L\oplus L^{*}}, \llbracket\cdot,\cdot\rrbracket_{L\oplus L^{*}})
\end{equation}
is a Courant algebroid, with inner product $\langle\cdot,\cdot\rangle_{L\oplus L^{*}}$, anchor $\pi_{L\oplus L^{*}}$ 
and Dorfman bracket $\llbracket\cdot,\cdot\rrbracket_{L\oplus L^{*}}$.
\end{proposition}

This procedure for constructing Courant algebroids from Lie algebroids is particularly natural in manifolds with a foliation,
as described by Vaisman in \cite[Section 3]{Vaisman2}.
The following result is essentially Proposition 3.1 in \cite{Vaisman2} (see also \cite[Section 3.4]{Svo18}): 

\begin{proposition}\label{Prop:Vaisman}
Let $(g,K)$ be an almost-para Hermitian structure on a $d$-dimensional real manifold $M$ (Def. \ref{Def:almostParaHermitian}). 
Let $L$ and $\tilde{L}$ be the two eigenbundles of $K$, and assume that $L$ is involutive. 
Then the quadruple
\begin{equation*}
 (TM\otimes\mathbb{C},g,P,\llb\cdot,\cdot\rrb)
\end{equation*}
is a Courant algebroid, where the inner product $g$ is the metric, the anchor $P$ is the projector 
\eqref{projectors} to $L$, and the Dorfman bracket is
\begin{equation*}
\llbracket X, Y \rrbracket =
 \varphi^{-1}(\llbracket (x,\tilde{x}_{\flat}), (y,\tilde{y}_{\flat}) \rrbracket_{L\oplus L^{*}}).
\end{equation*}
Here, $\varphi$ and $\flat$ are the isomorphisms \eqref{isom} and \eqref{isom0}, 
$\llb\cdot,\cdot\rrb_{L\oplus L^{*}}$ is the Dorfman bracket \eqref{Dorf}, and 
$x=PX$, $\tilde{x}=\tilde{P}X$, $y=PY$, $\tilde{y}=\tilde{P}Y$ are the decompositions of vectors induced by $K$.
\end{proposition}

\begin{proof}
The tangent bundle splits as $TM\otimes\mathbb{C}=L\oplus\tilde{L}$, where $L$ and $\tilde{L}$ are the 
eigenbundles of $K$ and $L$ is involutive.
From Example \ref{Ex:LieAlgFoliation}, the distribution $L$ defines a Lie algebroid $(L,[\cdot,\cdot]_L,\mbb{I}_L)$.
Therefore, applying Proposition \ref{Prop:CourantAlgebroid}, the quadruple
\begin{equation}\label{Courant1}
 (L\oplus L^{*},\langle\cdot,\cdot\rangle_{L\oplus L^{*}}, \pi_{L\oplus L^{*}}, \llbracket\cdot,\cdot\rrbracket_{L\oplus L^{*}})
\end{equation}
is a Courant algebroid, where the anchor, inner product and Dorfman bracket are given by \eqref{anch}, 
\eqref{innerp} and \eqref{Dorf}.
Now one simply has to transfer this structure to $TM\otimes\mathbb{C}$ using 
the isomorphisms $\varphi$ and $\flat$ (eqs. \eqref{isom}, \eqref{isom0}).
Any elements $X,Y\in \Gamma(TM\otimes\mathbb{C})$ can be written as 
$X=x+\tilde{x}$ and $Y=y+\tilde{y}$, where $x,y,\in\Gamma(L)$ and $\tilde{x},\tilde{y}\in\Gamma(\tilde{L})$.
Using these decompositions, and the fact that $g|_{L}=0=g|_{\tilde{L}}$, a brief calculation shows that
\begin{equation*}
 g(X,Y) = \langle (x,\tilde{x}_{\flat}), (y,\tilde{y}_{\flat}) \rangle_{L\oplus L^{*}},
\end{equation*}
so the metric $g$ plays the role of the inner product in Def. \ref{Def:CourantAlgebroid}.
Furthermore, recalling the projector $P$ given in \eqref{projectors} and using again $X=x+\tilde{x}$, 
we have
\begin{equation*}
 P(x+\tilde{x})=x=\pi_{L\oplus L^{*}}((x,\tilde{x}_{\flat}))
\end{equation*}
so $P$ plays the role of the anchor. Finally, let us see the Dorfman bracket.
Since any elements $\a,\b\in \Gamma(L^{*})$ can be written as 
$\a=\tilde{x}_{\flat}$, $\b=\tilde{y}_{\flat}$ for some $\tilde{x},\tilde{y}\in\Gamma(\tilde{L})$,
the Dorfman bracket \eqref{Dorf} is
\begin{equation}
 \llbracket (x,\tilde{x}_{\flat}), (y,\tilde{y}_{\flat}) \rrbracket_{L\oplus L^{*}} = 
 ([x,y]_L, \pounds^{L}_{x} \tilde{y}_{\flat} - \pounds^{L}_{y} \tilde{x}_{\flat} + {\rm d}^{L}g(\tilde{x},y)) \label{dorf2}
\end{equation}
where we recall that by definition, $\tilde{x}_{\flat}(y)\equiv g(\tilde{x},y)$.
The right hand side of \eqref{dorf2} is, of course, an element of $L\oplus L^{*}$.
In order to map this to $L\oplus\tilde{L}$, one uses the inverse isomorphisms $\sharp:L^{*}\to \tilde{L}$ and
$\varphi^{-1}:L\oplus L^{*}\to L\oplus\tilde{L}$, which map
$L\oplus L^{*}\ni(z,\gamma)\mapsto \varphi^{-1}((z,\gamma))= z+\gamma^{\sharp} \in L\oplus\tilde{L}$.
Replacing $(z,\gamma)$ by the right hand side of \eqref{dorf2}, we get
\begin{align}
\nonumber \llbracket x+\tilde{x}, y+\tilde{y} \rrbracket \equiv {}& 
 \varphi^{-1}(\llbracket (x,\tilde{x}_{\flat}), (y,\tilde{y}_{\flat}) \rrbracket_{L\oplus L^{*}})\\
 ={}& [x,y]+\left( \pounds^L_{x}\tilde{y}_{\flat} - \pounds^L_{y}\tilde{x}_{\flat}+{\rm d}^L g(\tilde{x},y) \right)^{\sharp}.
 \label{dorfTB}
\end{align}
(This is the unnumbered equation above eq. (20) in \cite{Svo18}.)
\end{proof}

\subsection{Conformal geometry and the Lee form}\label{Sec:Conformal}

As noticed in Remark \ref{Remark:Conformal}, conformal invariance plays an important role in our work.
Recall that, given a manifold $M$ and a metric $g$ on it, the conformal class of $g$ 
is defined as $[g]=\{\Omega^{2}g \;|\; \Omega\in C^{\infty}(M), \Omega>0 \}$.
We refer to the pair $(M,[g])$ as a conformal structure.
The elements in $[g]$ are called conformal representatives.

\begin{definition}
Let $(M, [g])$ be a conformal structure. 
A {\em Weyl connection} is a linear, torsion-free connection ${}^{\rm w}\nabla$ 
such that for any conformal representative $g\in[g]$, it holds ${}^{\rm w}\nabla g = -2f\otimes g$ 
for some 1-form $f$. We call $f$ the {\em Weyl 1-form}. 
\end{definition}

Under a change of conformal representative $g\to\hat{g}=\Omega^{2}g$, by definition we must have 
${}^{\rm w}\nabla \hat{g} = -2\hat{f}\otimes \hat{g}$ for some 1-form $\hat{f}$. Replacing $\hat{g}=\Omega^{2}g$ on 
the left hand side, we see that $f$ and $\hat{f}$ are related by (from now on we will frequently use the abstract index notation)
\begin{equation}\label{TransformationWeylForm}
 \hat{f}_{a} = f_{a} - ({\rm d}\log\Omega)_{a}.
\end{equation}

If we choose a metric $g\in[g]$ with Levi-Civita connection $\nabla$, the relationship between 
${}^{\rm w}\nabla$ and $\nabla$ is given by a tensor field $Q:\Gamma(TM)\times\Gamma(TM)\to\Gamma(TM)$.
Explicitly, one has
\begin{equation*}
 g({}^{\rm w}\nabla_{X}Y,Z) = g(\nabla_{X}Y,Z)+g(Q(X,Y),Z),
\end{equation*}
where
\begin{equation*}
 g(Q(X,Y),Z)=f(X)g(Y,Z)+f(Y)g(X,Z)-f(Z)g(X,Y).
\end{equation*}
In index notation this can be expressed as 
\begin{equation}\label{WeylConnection1}
 {}^{\rm w}\nabla_{a}Y^{b}=\nabla_{a}Y^{b}+Q_{ac}{}^{b}Y^{c},
\end{equation}
where 
\begin{equation}\label{WeylConnection2}
 Q_{ac}{}^{b} = \delta^{b}_{c}f_{a} + \delta^{b}_{a}f_{c} - g^{bd}g_{ac}f_{d}.
\end{equation}

There are interesting relations between the 1-form $f$ associated to a Weyl connection and 
certain properties of almost para-Hermitian structures. This dates back to the original work of H. C. Lee in \cite{Lee-conformal} 
on almost symplectic manifolds $(M,a_{ab})$: $M$ is a $d$-dimensional manifold (with $d$ even) 
and $a_{ab}=a_{[ab]}$ is a non-degenerate 2-form.
The inverse of $a_{ab}$ is $(a^{-1})^{ab}$, so that $(a^{-1})^{ac}a_{cb}=\delta^{a}{}_{b}$.
Lee defines the ``curvature tensor'' of $a_{ab}$ as $({\rm d}a)_{abc}$, and the ``curvature vector'' as
$(a^{-1})^{bc}({\rm d}a)_{abc}$.
Two such spaces $(M, a_{ab})$ and $(\hat{M},\hat{a}_{ab})$ are said to be ``conformal'' to each other if there exists 
a scalar field $\phi$ such that 
\begin{equation}\label{LeeConformal}
 \hat{a}_{ab}=\phi a_{ab}.
\end{equation}
Assuming $d>2$ and defining the 1-form
\begin{equation*}
 k_{a} := \tfrac{1}{(d-2)}(a^{-1})^{bc}({\rm d}a)_{abc},
\end{equation*}
it is observed in \cite{Lee-conformal} that $\hat{k}_{a}$ and $k_{a}$ are related by
\begin{equation}\label{LeeConformalRelation2}
 \hat{k}_{a} = k_{a} - ({\rm d}\log\phi)_{a}
\end{equation}
It follows from this that $({\rm d}k)_{ab}=({\rm d}\hat{k})_{ab}$, so $({\rm d}k)_{ab}$ is called the ``first conformal curvature 
tensor'' by Lee. 
There is also a ``second conformal curvature tensor'': a 3-form $c_{abc}$ given by
\begin{equation}\label{LeeSecondConformalCurvature}
 c_{abc} := ({\rm d}a)_{abc} + (k\wedge a)_{abc}.
\end{equation}
The tensors $c_{abc}$ and $\hat{c}_{abc}$ are related by $\hat{c}_{abc}=\phi c_{abc}$.

Now, as already noticed, a manifold with an almost-para-Hermitian structure $(g,K)$ comes automatically with 
an almost-symplectic form $\omega_{ab}=g_{cb}K^{c}{}_{a}$. Therefore, we may take $a_{ab}\equiv \omega_{ab}$ in Lee's construction.
Since $K^{a}{}_{c}K^{c}{}_{b}=\delta^{a}{}_{b}$, we have $(\omega^{-1})^{ab}=g^{ac}g^{bd}\omega_{cd}\equiv \omega^{ab}$.
Under a conformal transformation of the metric, $g_{ab}\to\hat{g}_{ab}=\Omega^{2}g_{ab}$, the new almost-symplectic 
2-form is $\hat{\omega}_{ab} = \Omega^{2}\omega_{ab} $. Thus, we take $\phi = \Omega^{2}$ in \eqref{LeeConformal} 
and subsequent formulas, and we define:
\begin{definition}
Let $([g],K)$ be an almost (para-)Hermitian conformal structure on a $d$-dimensional manifold $M$, with $d>2$.
For a choice $g\in[g]$, let $\omega$ be the associated almost-symplectic 2-form. 
We define the {\em Lee form} as
\begin{equation}\label{LeeForm}
 \theta_{a} = \tfrac{1}{2(d-2)}\omega^{bc}({\rm d}\omega)_{abc}.
\end{equation}
\end{definition}

If $\nabla_{a}$ is the Levi-Civita connection of $g_{ab}$, a short calculation using $\omega^{ab}\omega_{ab}=-d$ 
shows that \eqref{LeeForm} can be written as
\begin{equation*}
 \theta_{a} = \tfrac{1}{(d-2)}K^{b}{}_{c}\nabla_{b}K^{c}{}_{a}.
\end{equation*}
This is equivalent to $\theta_{a}=-\frac{1}{(d-2)}K^{c}{}_{a}\nabla_{b}K^{b}{}_{c}$, or in index-free notation:
\begin{equation}\label{LeeForm2}
 \theta(X)=-\tfrac{1}{(d-2)}(\delta \omega)(KX)
\end{equation}
for any vector field $X$, where $\delta$ is the codifferential.

From \eqref{LeeConformalRelation2} and \eqref{LeeForm} we see that 
the Lee form $\theta_{a}$ has exactly the transformation property of a Weyl 1-form under conformal transformations 
of the metric $g_{ab}\to \hat{g}_{ab}=\Omega^{2}g_{ab}$ 
(i.e. $\theta_{a} \to \hat{\theta}_{a} = \theta_{a} - ({\rm d}\log\Omega)_{a}$).
Therefore, from the discussion above, we deduce:
\begin{proposition}\label{Prop:Lee}
An almost (para-)Hermitian structure $(g,K)$ on $M$ induces a natural Weyl connection in the conformal structure $(M,[g])$,
by taking the Lee form as the Weyl 1-form, $\theta_{a} \equiv f_{a}$.
\end{proposition}

This motivates the following definition:
\begin{definition}\label{Def:CompatibleWeylConnection}
Let $([g],K)$ be an almost (para-)Hermitian conformal structure on $M$, and let ${}^{\rm w}\nabla$ be a Weyl connection. 
We say that ${}^{\rm w}\nabla$ and $K$ are {\em compatible} if the Weyl 1-form associated to ${}^{\rm w}\nabla$ and 
the Lee form associated to $K$ coincide (in other words, ${}^{\rm w}\nabla$ is induced by $K$).
\end{definition}

\begin{remark}
In the literature, see e.g. \cite[Section 4]{Gover} (also \cite{Bailey}), the usual definition of compatibility of
a Weyl connection ${}^{\rm w}\nabla$ with a Hermitian structure $K$ is that they must satisfy
\begin{equation}\label{Compatibility2}
 {}^{\rm w}\nabla_{a}K^{a}{}_{b}=0.
\end{equation}
Using \eqref{WeylConnection1}-\eqref{WeylConnection2}, it follows easily that this is true if and only if 
the condition of Definition \ref{Def:CompatibleWeylConnection} holds, namely the Weyl 1-form is equal to the Lee form.
Thus, the two definitions coincide.
However, we have chosen Def. \ref{Def:CompatibleWeylConnection} since it is an immediate consequence of Lee's construction, 
while the geometric meaning of requiring \eqref{Compatibility2} as a compatibility condition is not clear to us.
\end{remark}

Finally, from Lee's results \cite{Lee-conformal} and the definitions given in Section \ref{Sec:paraHermitian} we have:
\begin{proposition}[Lee \cite{Lee-conformal}]
Let $(g,K)$ be an almost (para-)Hermitian structure on a $d$-dimensional manifold $M$. Then:
\begin{enumerate}
\item $(g,K)$ is almost (para-)K\"ahler if and only if the Lee form \eqref{LeeForm} vanishes.
\item For $d=4$, $(g,K)$ is locally conformally almost (para-)K\"ahler if and only if the Lee form is closed.
\item For $d>4$, $(g,K)$ is locally conformally almost (para-)K\"ahler if and only if $c_{abc}\equiv 0$ (where 
$c_{abc}$ is defined in \eqref{LeeSecondConformalCurvature}).
\end{enumerate}
\end{proposition}

\begin{proof}
The first item follows from the definitions of almost (para-)K\"ahler (Def. \ref{Def:almostparaKahler}, ${\rm d}\omega = 0$)
and the Lee form \eqref{LeeForm}.
The second and third items are Theorem 5 in \cite{Lee-conformal}.
\end{proof}

\section{Almost para-Hermitian structures in four dimensions}\label{Sec:4D}

\subsection{Self-dual forms}

Consider a 4-dimensional, orientable, real manifold $M$, equipped with a metric.
Let $\Lambda^{k}(M)=\wedge^{k} T^{*}M$ be the space of $k$-forms.
The Hodge star operator $*$ can be seen as a map $ * : \Lambda^{2}(M)\to\Lambda^{2}(M)$
satisfying $*^{2}=(-1)^{s}$, where $s$ is the number of $(-1)$'s appearing in the signature of the metric.
Thus, in Riemannian ($s=0$) and split ($s=2$) signature, $*$ always defines a para-complex structure in $\Lambda^{2}(M)$, 
whereas in Lorentz signature ($s=3$) it defines a complex structure in $\Lambda^{2}(M)$.
The eigenvalues of $*$ are $\pm \sqrt{(-1)^{s}}$. The space of 2-forms is then decomposed as 
\begin{equation}\label{SDdecomposition}
 \Lambda^{2}(M)=\Lambda^{2}_{+}(M)\oplus\Lambda^{2}_{-}(M)
\end{equation}
where $\Lambda^{2}_{+}(M)$ is the eigenspace corresponding to the eigenvalue $+\sqrt{(-1)^{s}}$, 
and $\Lambda^{2}_{-}(M)$ corresponds to $-\sqrt{(-1)^{s}}$.
Elements of $\Lambda^{2}_{+}(M)$ are called {\em self-dual} (SD) 2-forms, and elements of $\Lambda^{2}_{-}(M)$
are {\em anti-self-dual} (ASD) 2-forms.

\begin{theorem}\label{Theorem:SD}
Let $M$ be a real, 4-dimensional, orientable manifold with a metric $g$, let $A\in{\rm Aut}(TM\otimes\mathbb{C})$ 
and let $W$ be a bilinear map defined by
\begin{equation}\label{Def:W}
 W(X,Y) = g(AX,Y)
\end{equation}
for any vectors $X,Y$. Furthermore, let $a$ be a non-vanishing scalar\footnote{For the case $a=0$, see 
Remark \ref{Remark:NullKahler} below.}.
Then $A$ satisfies the conditions
\begin{align}
 A^{2} ={}& a \mathbb{I}, \label{Ksquare} \\
 g(AX,AY) ={}& -a g(X,Y) \label{gKK}
\end{align}
if and only if the map $W$ is either a self-dual or an anti-self-dual 2-form.
\end{theorem}

\begin{proof}
We will use the abstract index notation, and the metric and its inverse to raise and lower indices as convenient.
Notice that since $A$ is an automorphism, the map $W$ in \eqref{Def:W} is non-degenerate.
Appendix B in \cite{Wald} will be useful in the following.
Since $M$ is orientable, it has a volume 4-form $\varepsilon_{abcd}$, and
the metric allows to normalize it as $\varepsilon^{abcd}\varepsilon_{abcd}=(-1)^{s}4!$. 

Assume first that $A\in{\rm Aut}(TM\otimes\mathbb{C})$ satisfies \eqref{Ksquare} and \eqref{gKK}. In index notation this is
\begin{align}
 A^{a}{}_{c}A^{c}{}_{b} ={}& a \delta^{a}{}_{b}, \label{Ksquare-index} \\
 g_{cd}A^{c}{}_{a}A^{d}{}_{b} ={}& -a g_{ab} \label{gKK-index}
\end{align}
with $a \neq 0$. The bilinear map \eqref{Def:W} is $W_{ab} = g_{bc}A^{c}{}_{a}$.
Multiplying by $A^{b}{}_{d}$ and using \eqref{gKK-index}: 
\begin{equation*}
 W_{ab}A^{b}{}_{d} = g_{bc}A^{c}{}_{a}A^{b}{}_{d}=-a g_{ad}.
\end{equation*}
Contracting now with $A^{d}{}_{c}$ and using \eqref{Ksquare-index}:
\begin{equation*}
 W_{ab}A^{b}{}_{d}A^{d}{}_{c} = a W_{ac} = -a g_{ad}A^{d}{}_{c} = -a W_{ca},
\end{equation*}
which shows that $W_{ac}=-W_{ca}$ and therefore $W_{ab}$ is a 2-form. Now we want to show that it is (A)SD.
We can write the volume form as $\varepsilon_{abcd}=N W_{[ab}W_{cd]}$, where $N$ is 
determined by the normalization of $\varepsilon_{abcd}$. Expanding the skew-symmetrization:
\begin{equation*}
 \varepsilon_{abcd} = \tfrac{N}{3}(W_{ab}W_{cd}+W_{ac}W_{db}+W_{ad}W_{bc}).
\end{equation*}
Contracting with $W^{cd}$ and using that $W_{ab}=-W_{ba}$:
\begin{equation*}
 \varepsilon_{abcd}W^{cd} = \tfrac{N}{3} (W_{ab}W_{cd}W^{cd}+2W_{ac}W_{db}W^{cd}).
\end{equation*}
Now, using that $W_{db}=g_{eb}A^{e}{}_{d}$ and $W^{cd}=g^{ce}g^{df}W_{ef}$, we have
\begin{equation*}
 W_{db}W^{cd}=A^{c}{}_{d}A^{d}{}_{b}=a \delta^{c}{}_{b}
\end{equation*}
from which it also follows that $W_{cd}W^{cd}=-4a$. Therefore
\begin{equation}\label{dualW0}
 \varepsilon_{abcd}W^{cd} = -\tfrac{2}{3}Na W_{ab}.
\end{equation}
Using now the normalization of $\varepsilon_{abcd}$:
\begin{equation*}
 (-1)^{s}4!=\varepsilon^{abcd}\varepsilon_{abcd}=N\varepsilon^{abcd}W_{ab}W_{cd} = \tfrac{8}{3} N^{2}a^{2},
\end{equation*}
thus 
\begin{equation*}
 Na = \pm 3 \sqrt{(-1)^{s}}.
\end{equation*}
Replacing in \eqref{dualW0} and using the definition ${}^{*}W_{ab}=\frac{1}{2}\varepsilon_{abcd}W^{cd}$, 
we get
\begin{equation*}
 {}^{*}W_{ab} = \mp \sqrt{(-1)^{s}} \; W_{ab},
\end{equation*}
which shows that $W_{ab}$ must be (A)SD.

Suppose now that the map $W$ defined in \eqref{Def:W} is a 2-form. 
Writing the volume form again as $\varepsilon_{abcd} = N W_{[ab}W_{cd]}$, 
and using the normalization $\varepsilon^{abcd}\varepsilon_{abcd}=(-1)^{s}4!$,
a short calculation reveals that $(-1)^{s}4!=2N {}^{*}W^{ab}W_{ab}$, thus
\begin{equation}\label{normalization}
 N = 12(-1)^{s} ({}^{*}W^{ab}W_{ab})^{-1}.
\end{equation}
Now, from formula (B.2.13) in \cite{Wald} we have
\begin{equation*}
 (-1)^{s}3! \delta^{a}{}_{b} = \varepsilon^{cdea}\varepsilon_{cdeb}
\end{equation*}
Replacing $\varepsilon_{cdeb}=N W_{[cd}W_{eb]}$ and the expression for $N$, we get
the general identity
\begin{equation*}
 \delta^{a}{}_{b} = -4 ({}^{*}W^{de}W_{de})^{-1} \; {}^{*}W^{ac}W_{cb}.
\end{equation*}
Assume now that $W_{ab}$ is (A)SD:
\begin{equation*}
 {}^{*}W_{ab} = \e W_{ab},
\end{equation*}
where $\e=\pm\sqrt{(-1)^{s}}$, depending on the signature of $g$ and on whether we consider SD or 
ASD forms. In any case, the above identity becomes
\begin{equation}\label{Wsquare}
 \delta^{a}{}_{b} = -4 (W^{de}W_{de})^{-1} \; W^{ac}W_{cb},
\end{equation}
which shows that $A^{a}{}_{b}=g^{ac}W_{bc}$ satisfies \eqref{Ksquare-index}, with $a=-\frac{1}{4}W_{ef}W^{ef}$. 
Finally, to see that \eqref{gKK} holds, we compute
\begin{equation*}
 g_{cd}A^{c}{}_{a}A^{d}{}_{b} = -W_{da}A^{d}{}_{b} = -a g_{ab}
\end{equation*}
where the second equality is deduced form \eqref{Ksquare-index} by contracting with $g_{da}$.
\end{proof}

\begin{corollary}\label{Corollary:SD}
Let $A\in{\rm Aut}(TM\otimes\mathbb{C})$ and let $W$ be defined by \eqref{Def:W}.
Define a scalar field $\phi$, a bilinear map $\omega$, and a map $K\in{\rm Aut}(TM\otimes\mathbb{C})$ by, respectively,
\begin{align}
 \phi ={}& (-\tfrac{1}{4}W_{ab}W^{ab})^{1/2}, \label{Def:phi} \\
 \omega ={}& \phi^{-1} W, \label{Def:omega} \\
 g(KX,Y) ={}& \omega(X,Y) \label{Def:K}
\end{align}
for any vectors $X,Y$. Then $K$ satisfies the conditions
\begin{align}
 K^{2} ={}& \mathbb{I}, \label{paraHermitian1} \\
 g(KX,KY) ={}& -g(X,Y) \label{paraHermitian2}
\end{align}
if and only if $\omega$ is either a self-dual or an anti-self-dual 2-form.
\end{corollary}

\begin{remark}\label{Remark:reality}
Suppose that $\omega$ in \eqref{Def:omega} is an (A)SD 2-form. Then the map $K$ defined 
in \eqref{Def:K} satisfies equations \eqref{paraHermitian1} and \eqref{paraHermitian2},
and its eigenvalues are $+1,+1,-1,-1$ so the two eigenbundles of $K$ have the same rank. 
According to definitions \ref{Def:paraHermitianStructure} and \ref{Def:almostParaHermitian}, 
we say that $(g,K)$ is an almost para-Hermitian structure.
Notice that {\em $K$ is not necessarily real-valued}, this is the reason why we use $TM\otimes\mathbb{C}$
instead of $TM$ in Def. \ref{Def:almostParaHermitian}.
The reality of $K$ depends on the signature of $g$:
\begin{enumerate}
\item Riemannian signature $(++++)$: (A)SD forms can be chosen to be real, but $\phi$ in \eqref{Def:phi} 
is always purely imaginary, so $K$ is purely imaginary or complex-valued. 
\item Split signature $(++--)$: (A)SD forms can be chosen to be real, but $\phi$ can be real or complex, so
$K$ can be real- or complex-valued.
\item Lorentz signature $(+---)$: (A)SD forms are necessarily complex, so $K$ is complex-valued.
\end{enumerate}
\end{remark}

\subsection{Spaces of almost para-Hermitian structures}\label{Sec:spaceParaHermitian}

The results of Theorem \ref{Theorem:SD} and its Corollary \ref{Corollary:SD} show us that we can think of 
maps satisfying \eqref{paraHermitian1} and \eqref{paraHermitian2}
equivalently in terms of non-degenerate (A)SD 2-forms. 
Here we will give a convenient parametrization of the space of such maps at a point $x\in M$, that is of the space
\begin{equation}\label{SpaceParaHermitian}
 P_{x}:=\{ K\in{\rm Aut}(T_{x}M\otimes \mathbb{C}) \;|\; K^{2}=\mathbb{I} \;\;\text{and}\;\; g(K\cdot,K\cdot)=-g(\cdot,\cdot) \}.
\end{equation}
As emphasized in Remark \ref{Remark:reality}, we consider $K\in{\rm Aut}(T_{x}M\otimes \mathbb{C})$ 
since a tensor satisfying \eqref{paraHermitian1}-\eqref{paraHermitian2} is not necessarily real.
Imposing reality conditions changes the structure of $P_{x}$, and this will depend on the metric signature.
Our results in this subsection can be summarized in the following form:

\begin{theorem}\label{Theorem:spaceParaHermitian}
Let $M$ be a real, orientable 4-manifold with a metric $g$, and let $x\in M$. 
The space \eqref{SpaceParaHermitian} has two connected components, 
and depending on the signature of $g$, each component can be parametrized as follows:
\begin{enumerate}
\item Riemannian signature: the elements in the set \eqref{SpaceParaHermitian} are not real but 
can be chosen to be purely imaginary, and each of the two components in the space of such maps is 
\begin{equation}\label{ParaHermitianEuclidean}
 P^{(R)}_{x}\cong \mathbb{CP}^1 \cong S^{2}
\end{equation}
where $\mathbb{CP}^1$ is the complex projective line, and $S^{2}$ is the unit 2-sphere. 
\item Split signature: the elements in the set \eqref{SpaceParaHermitian} can be chosen to be real, 
and each of the two components in the space of such maps is 
\begin{equation}
 P^{(S)}_{x}\cong (\mathbb{RP}^{1}\times\mathbb{RP}^{1})\backslash\mathbb{RP}^{1} \cong H^{1}
\end{equation}
where $\mathbb{RP}^{1}$ is the real projective line, and $H^{1}$ is a hyperboloid of one sheet. 
\item Lorentz signature: the elements in the set \eqref{SpaceParaHermitian} are necessarily complex, 
and each of the two components in the space of such maps is 
\begin{equation}
  P^{(L)}_x \cong (\mathbb{CP}^{1}\times\mathbb{CP}^{1})\backslash\mathbb{CP}^{1} \cong \mathbb{C}S^{2}
\end{equation}
where $\mathbb{C}S^{2}$ is the complexified 2-sphere. 
\end{enumerate}
\end{theorem}

\begin{remark}
The result \eqref{ParaHermitianEuclidean} is analogous to the result about almost-complex structures 
in Riemannian geometry, which is already known in the literature (see Remark \ref{Remark:Euclidean} below).
\end{remark}

We will first analyse the general structure of \eqref{SpaceParaHermitian} and then study 
the different signatures separately.

First of all, the two connected components in $P_{x}$ refer to the fact that, if $K\in P_{x}$, 
then the 2-form $\omega(\cdot,\cdot)=g(K\cdot,\cdot)$ can be self-dual or anti-self-dual.
We will focus on only one of the components of $P_{x}$; the analysis for the other is analogous.

The parametrization of $P_{x}$ involves the use of spinors.
We refer to Appendix \ref{Appendix:spinors} for a brief review of spinors in four dimensions, as well as for 
notation and conventions (see also Appendix \ref{Appendix:CovariantFormalism} for spinor {\em fields}).
In particular, we raise and lower spinor indices with the symplectic forms $\e_{AB}$, $\e_{A'B'}$ and their inverses. 
In any signature, the spaces of SD and ASD 2-forms have the spinor decomposition
\begin{equation}\label{formsspinors}
 \Lambda^{2}_{+}(M) \cong \mathbb{S}'^{*}\odot\mathbb{S}'^{*}, 
 \qquad  \Lambda^{2}_{-}(M) \cong \mathbb{S}^{*}\odot\mathbb{S}^{*},
\end{equation}
where $\mathbb{S}'^{*}$ and $\mathbb{S}^{*}$ are the (dual) primed and unprimed spinor bundles 
(see the beginning of Appendix \ref{Appendix:CovariantFormalism}).
In indices, this means that if $W^{+}_{ab}$ is SD, and $W^{-}_{ab}$ is ASD, 
then there exist symmetric spinors $\varphi_{AB}=\varphi_{(AB)}$ and $\psi_{A'B'}=\psi_{(A'B')}$
such that 
\begin{align}
 W^{+}_{ab} ={}& \psi_{A'B'}\epsilon_{AB}, \label{W+0} \\
 W^{-}_{ab} ={}& \varphi_{AB}\epsilon_{A'B'}. \label{W-0}
\end{align}
These 2-forms are non-degenerate if and only if $\psi_{A'B'}\psi^{A'B'} \neq 0$ 
and $\varphi_{AB}\varphi^{AB} \neq 0$ respectively.
In the following, since we will focus on only one of the two components of $P_{x}$,
we are free to choose to work with either SD or ASD forms. 
For concreteness, from now on we will focus on the ASD case (the SD case being entirely analogous), 
that is on forms like \eqref{W-0}, and we omit the superscript ``$-$'' since it will not be needed.

\begin{remark}[Convention]\label{Remark:chirality}
If a tensor/spinor depends non-trivially only on unprimed spinors, we will often refer to it as having ``negative chirality'', 
while if it depends only on primed spinors we will say that it has ``positive chirality''. 
For example, we say that $\varphi_{A}$ or $T^{a}{}_{b}=\tau^{A}{}_{B}\delta^{A'}{}_{B'}$ have negative chirality, 
and that $U_{aB}=\sigma_{A'}\epsilon_{AB}$ or $V_{ab}=\rho_{A'B'}\epsilon_{AB}$ have positive chirality. 
More formally, the notions of negative and positive chirality refer to the two kinds of 
representations of the spin group (in the notation of appendix \ref{Appendix:spinors},
negative chirality corresponds to representations $(n,0)$, and positive chirality to $(0,m)$).
\end{remark}

\begin{remark}[Null-K\"ahler]\label{Remark:NullKahler}
We notice that the proof of Theorem \ref{Theorem:SD} is particularly simple if one uses
the isomorphisms \eqref{SDdecomposition}, \eqref{formsspinors} and the explicit expressions 
\eqref{W+0} and \eqref{W-0}. 
Actually, these decompositions also allow to deal with the case $a=0$ in \eqref{Ksquare}. 
Suppose that $N\in{\rm End}(TN\otimes\mathbb{C})$ satisfies $N^{2}=0$ and 
$g(NX,Y)+g(X,NY)=0$ for all $X,Y$. 
Then using \eqref{SDdecomposition}, \eqref{formsspinors}--\eqref{W-0}, it is straightforward to show that
any $N$ satisfying these conditions must be of the form
\begin{equation}
 N^{a}{}_{b}=\a^{A}\a_{B}\delta^{A'}{}_{B'} \qquad \text{or } \qquad N^{a}{}_{b}=\mu^{A'}\mu_{B'}\delta^{A}{}_{B}
\end{equation}
for some spinors $\a^{A}$, $\mu^{A'}$, depending on whether the 2-form $g(N\cdot,\cdot)$ is SD or ASD.
We may refer to the pair $(g,N)$ as an `almost-null K\"ahler structure'. 
It follows that the space of such maps at a point $x\in M$ is simply the space of spinors at $x$.
If the spinor $\a^{A}$ (or $\mu^{A'}$) is parallel under the Levi-Civita connection of $g$, 
then this is called a {\em Null-K\"ahler structure}, see \cite{Dun01} and \cite[Section 10.2.3]{Dun10}.
(If $N$ is real, then the signature of $g$ must be split, however here we allow $N$ to be complex.)
\end{remark}

Let $W_{ab}$ be an ASD 2-form at $x\in M$, with the spinor representation \eqref{W-0}, and let
$(\xi_{A},\eta_{A})$ be a basis of $\mathbb{S}^{*}|_{x}$, where $\xi_{A}\eta^{A}=\chi\neq0$. 
Then $\varphi_{AB}$ can expanded as (see Appendix \ref{Appendix:spinors})
\begin{equation}\label{ExpansionInBasis}
 \varphi_{AB}=\chi^{-2}[\varphi_{2}\xi_{A}\xi_{B}-\varphi_{1}(\xi_{A}\eta_{B}+\eta_{A}\xi_{B})+\varphi_{0}\eta_{A}\eta_{B}],
\end{equation}
where $\varphi_{0}=\varphi_{AB}\xi^{A}\xi^{B}$, $\varphi_{1}=\varphi_{AB}\xi^{A}\eta^{B}$ and 
$\varphi_{2}=\varphi_{AB}\eta^{A}\eta^{B}$.
The scalar \eqref{Def:phi} is then $\phi=(-\frac{1}{2}\varphi_{AB}\varphi^{AB})^{1/2}=\chi^{-1}(\varphi^{2}_1-\varphi_{0}\varphi_{2})^{1/2}$.
Thus, the tensor $K^{a}{}_{b}=\phi^{-1} g^{ac}W_{bc}$ satisfies \eqref{paraHermitian1} and \eqref{paraHermitian2} 
and is explicitly
\begin{equation}\label{KGR0}
 K^{a}{}_{b}=\frac{1}{\chi\sqrt{\varphi^{2}_1-\varphi_{0}\varphi_{2}}} 
 \left[\varphi_{2}\xi^{A}\xi_{B}-\varphi_{1}(\xi^{A}\eta_{B}+\eta^{A}\xi_{B}) + \varphi_{0}\eta^{A}\eta_{B} \right]\delta^{A'}{}_{B'}.
\end{equation}
This can be written in a more convenient form by expressing $\varphi_{AB}$ in terms of its principal spinors, 
see Appendix \ref{Appendix:spinors}: 
since $\varphi_{AB}\varphi^{AB}\neq0$, there exist two non-proportional spinors $\a_{A},\b_{A}$ such that 
\begin{equation}\label{decompositionPNDs}
 \varphi_{AB}=\a_{A}\b_{B}+\a_{B}\b_{A}.
\end{equation}
Independently of signature, the spinors $\a_{A}$ and $\b_{A}$ are in general {\em complex}, since they 
are obtained by finding the roots of a second order polynomial; see around eq. \eqref{PrincipalSpinors}. 
However, whereas in Lorentzian and Riemmanian signature they are {\em necessarily} complex,
in split signature the combination in the right hand side of \eqref{decompositionPNDs} is real.

The decomposition \eqref{decompositionPNDs} implies that $\phi=\a_{C}\b^{C}$, therefore
\begin{equation}\label{KGR}
  K^{a}{}_{b} = (\a_{C}\b^{C})^{-1} (\a^{A}\b_{B}+\b^{A}\a_{B})\delta^{A'}{}_{B'}.
\end{equation}
From Corollary \ref{Corollary:SD} we see that {\em any} map satisfying conditions \eqref{paraHermitian1} and \eqref{paraHermitian2} 
in four dimensions has the form \eqref{KGR} for some spinors $\a_{A},\b_{A}$ 
(the only other possibility being choosing positive instead of negative chirality).
However, \eqref{KGR} does not uniquely fix these spinors: we see from this equation that $K$ is invariant under
\begin{equation}\label{gaugefreedomK}
 \a_{A} \to \lambda \a_{A}, \qquad \b_{A} \to \mu \b_{A}
\end{equation}
for any non-vanishing scalars $\lambda$ and $\mu$.
Now, the space of spinors at a point (the fiber of the spin bundle) is $\mathbb{C}^{2}$ in signature 
$(+---)$ and $(++++)$, and $\mathbb{R}^2$ in signature $(++--)$.  
Since any element of the set \eqref{SpaceParaHermitian} (with negative chirality) can be put in the form \eqref{KGR}, 
and since the spinors are subject to the projective equivalence \eqref{gaugefreedomK},
we deduce that {\em a priori} each component of the set $P_{x}$ can be parametrized by 
the product of two projective spaces, that is $\mathbb{CP}^1\times\mathbb{CP}^1$.
(Notice that although in the split signature case the projective spin space is $\mathbb{RP}^{1}$, 
if we allow the rescalings in \eqref{gaugefreedomK} to be complex then we get the complexification 
of $\mathbb{RP}^{1}$, which is $\mathbb{CP}^{1}$.)
But since the spinors $\a_{A}$ and $\b_{A}$ are not allowed to be proportional, we must 
remove the ``diagonal'' from $\mathbb{CP}^1\times\mathbb{CP}^1$,
so we get
\begin{equation}\label{SpaceParaHermitian2}
 P_{x} \cong \{ (\a,\b)\in\mathbb{CP}^{1}\times\mathbb{CP}^{1} \;|\; \a\neq\b \} 
 \cong (\mathbb{CP}^{1}\times\mathbb{CP}^{1})\backslash\mathbb{CP}^{1}
\end{equation}
where we are now abusing notation and using $P_{x}$ for one of the two components of \eqref{SpaceParaHermitian}.
Requiring reality conditions for \eqref{KGR}, however, imposes additional restrictions on the spinors $\a_{A},\b_{A}$. 
We will analyse separately the different signatures.

\subsubsection{Riemannian signature}

In this case we can require the ASD 2-form $W_{ab}$ to be {\em real}. 
To analyse this in more detail, we need to introduce the Euclidean spinor conjugation. This is an 
involution $\dagger$ on the spin spaces, i.e. $\mathbb{S}\to\mathbb{S}$ and $\mathbb{S}'\to\mathbb{S}'$ 
(so it does not interchange the chirality as in the Lorentzian case), 
and if we write a spinor $\varphi_{A}$ in components as $\varphi_{A}=(a,b)$, then 
$\varphi^{\dagger}_{A}=(-\bar{b},\bar{a})$. It follows that $\dagger^{2}=-1$.
Since $\epsilon^{\dagger}_{A'B'}=\epsilon_{A'B'}$, we get that $W_{ab}$ is real
if and only if $\beta_{A}=i\alpha^{\dagger}_{A}$, so $W_{ab}=i(\a_{A}\a^{\dagger}_{B}+\a^{\dagger}_{A}\a_{B})\e_{A'B'}$. 
Furthermore, $(\alpha_{C}\beta^{C})^{-1}=-i/ \Vert\a\Vert^2$, where $\Vert\a\Vert^{2}:=\epsilon_{AB}\a^{A}\a^{\dagger B}$.
Therefore, with the additional requirement of {\em real} ASD 2-forms,
the map \eqref{KGR} becomes in the Riemannian case
\begin{equation}\label{KEuclidean}
 K^{a}{}_{b}=\frac{1}{\Vert\a\Vert^{2}}(\a^{A}\a^{\dagger}_{B}+\a^{\dagger A}\a_{B})\delta^{A'}{}_{B'},
\end{equation}
so we see explicitly that $K$ is parametrized by only {\em one} projective spinor. 
Each of the two components of the set of ``Riemannian para-Hermitian structures'' at $x\in M$ is therefore 
\begin{equation}
 P^{(R)}_{x} \cong \mathbb{CP}^{1}. \label{SEuclidean}
\end{equation}
The complex projective line $\mathbb{CP}^{1}$ is the Riemann sphere, and is topologically $\mathbb{CP}^{1}\cong S^{2}$. 
Thus, there is a 2-sphere of Riemannian para-Hermitian structures at any point.

\begin{remark}\label{Remark:Euclidean}
From \eqref{KEuclidean} we see immediately that $K$ is purely imaginary, since $\dagger^2=-1$.
Therefore, the map 
\begin{equation*}
 J:=iK
\end{equation*}
is real-valued and satisfies $J^{2}=-\mathbb{I}$, so it defines a complex structure 
in each tangent space, and the space of complex (Hermitian) structures at any point is again parametrized by $\mathbb{CP}^{1}$. 
In this way we recover the Riemannian results for almost-complex structures given in \cite[Section 9.1]{WardWells} and 
\cite{Atiyah} (see eq. (9.1.20) in the first reference and page 429 in the second).
\end{remark}

The 2-sphere \eqref{SEuclidean} of Riemannian para-Hermitian structures can also be seen by using non-projective spinors as follows.
From the general form \eqref{KGR0}, 
choosing different values for $\varphi_{0}, \varphi_{1}, \varphi_{2}$ we deduce the following particular cases
\begin{align}
 (K_{1})^{a}{}_{b} ={}& \chi^{-1}(\xi^{A}\eta_{B}+\eta^{A}\xi_{B})\delta^{A'}{}_{B'}, \label{K1} \\
 (K_{2})^{a}{}_{b} ={}& i\chi^{-1}(\xi^{A}\xi_{B}+\eta^{A}\eta_{B})\delta^{A'}{}_{B'}, \label{K2} \\
 (K_{3})^{a}{}_{b} ={}& \chi^{-1}(\xi^{A}\xi_{B}-\eta^{A}\eta_{B})\delta^{A'}{}_{B'} \label{K3}
\end{align}
(where reality conditions imply $\eta_{A}=i\xi^{\dagger}_{A}$).
Any $K$ can be expressed as a linear combination of these tensors\footnote{Note that \eqref{KGR0} is equivalent to
\begin{equation}
 K^{a}{}_{b} = \frac{(-\varphi_{1})}{\sqrt{\varphi_{1}^{2}-\varphi_{0}\varphi_{2}}}(K_{1}){}^{a}{}_{b} 
 +\frac{(-i)(\varphi_{2}+\varphi_{0})}{2\sqrt{\varphi_{1}^{2}-\varphi_{0}\varphi_{2}}}(K_{2}){}^{a}{}_{b} 
 +\frac{(\varphi_{2}-\varphi_{0})}{2\sqrt{\varphi_{1}^{2}-\varphi_{0}\varphi_{2}}}(K_{3}){}^{a}{}_{b} 
\end{equation}
}. It is straightforward to check the following identities:
\begin{equation}\label{triplet}
 (K_{1})^{2}=(K_{2})^{2}=(K_{3})^{2}=\mathbb{I}, \qquad K_{i}K_{j}=-K_{j}K_{i} \;\; \text{for } \;   i \neq  j.
\end{equation}
Using these properties, a short calculation shows that 
\begin{equation}\label{twosphere}
 (aK_1 +b K_2 +cK_3)^{2}=(a^2+b^2+c^2)\mathbb{I},
\end{equation}
where $a,b,c$ are real numbers. Thus, as long as $(a,b,c)\in S^{2}$, the combination $\mathcal{K}_{(a,b,c)}=aK_1 +b K_2 +cK_3$ is 
a new para-Hermitian structure parametrized by an element $(a,b,c)$ of $S^{2}$, so there is a 2-sphere of such structures.

Since the $K_{i}$ are all purely imaginary, we can define the real maps $J_{1}=iK_{1}$, $J_{2}=iK_{2}$ and $J_{3}=iK_{3}$; 
then it follows that these tensors satisfy 
\begin{equation}
 (J_{i})^{2}=-\mathbb{I}, \qquad J_{1}J_{2}=J_{3}, \quad J_{2}J_{3}=J_{1}, \quad J_{3}J_{1}=J_{2},
\end{equation}
so the triple $(J_1,J_2,J_3)$ defines an almost-hyperhermitian structure on $M$ (Def. \ref{Def:hypercomplex}).

\subsubsection{Split signature}

In signature $(++--)$ spinors can be chosen to be real. Therefore, the reality of the map \eqref{KGR0} 
depends on the sign of $\varphi^{2}_1-\varphi_{0}\varphi_{2}$.
This is equivalent to the reality conditions for the principal spinors $\a_{A}$ and $\b_{A}$ in \eqref{decompositionPNDs}: 
if $\varphi^{2}_1-\varphi_{0}\varphi_{2}>0$ then $\a_{A}$ and $\b_{A}$ are real, whereas if 
$\varphi^{2}_1-\varphi_{0}\varphi_{2}<0$ they are complex; see around equation \eqref{PrincipalSpinors}.
Thus, requiring \eqref{KGR0} to be real, and hence an ordinary para-Hermitian structure, is equivalent 
to requiring $\a_{A}$ and $\b_{A}$ in \eqref{KGR} to be real. We can thus restrict to real rescalings in 
\eqref{gaugefreedomK}, so the space of para-Hermitian structures at $x$ (of a definite chirality) is
\begin{equation}\label{SNeutral}
  P^{(S)}_{x} \cong \{ (\a,\b)\in\mathbb{RP}^{1}\times\mathbb{RP}^{1} \;|\; \a\neq\b \} 
  \cong  (\mathbb{RP}^{1}\times\mathbb{RP}^{1})\backslash\mathbb{RP}^{1}.
\end{equation}
The real projective line $\mathbb{RP}^{1}$ is topologically $\mathbb{RP}^{1}\cong S^{1}$, 
so the space of para-Hermitian structures at any point in split signature is the torus $S^{1}\times S^{1}$
with the line $\a=\b$ removed.
To visualize this structure, we consider again the three cases \eqref{K1}, \eqref{K2}, \eqref{K3}, 
where now $\xi_{A}$ and $\eta_{A}$ are real spinors.
We see that $K_1$ and $K_3$ are real, but $K_2$ is purely imaginary, so we define the real map $K'_2=iK_2$.
These tensors satisfy the following:
\begin{equation}
 (K_{1})^{2}=(K_{3})^{2}=\mathbb{I}=-(K'_{2})^{2}, \qquad K_{1}K'_{2}=-K'_{2}K_{1}=K_{3}.
\end{equation}
From these relations and definition \ref{Def:hypercomplex} we see that
the triple $(K_{1},K'_{2},K_{3})$ defines an almost-para-hyperhermitian structure.
If $a,b,c$ are real numbers, we now find
\begin{equation}\label{hyperboloid}
 (aK_1 +b K'_2 +cK_3)^{2}=(a^2-b^2+c^2)\mathbb{I},
\end{equation}
thus the condition for $aK_1 +b K'_2 +cK_3$ to be an almost-para-Hermitian structure is $a^2-b^2+c^2=1$, which 
describes a hyperboloid. 
Therefore, the space \eqref{SNeutral} is in this case a hyperboloid of one sheet.
One can also convince oneself of this by drawing $(S^{1}\times S^{1})\backslash S^{1}$ 
(where we think of the removed $S^1$ as a ``diagonal'' circle ---i.e. non-contractible).

\begin{remark}[Almost-complex structures]\label{Remark:complexsplit}
An analysis analogous to the one that leads to \eqref{SpaceParaHermitian2} can be carried out to find the set of 
almost {\em complex} structures compatible with $g$ in split signature 
(i.e. maps satisfying $J^{2}=-\mathbb{I}$ and $g(J\cdot,J\cdot)=+g(\cdot,\cdot)$).
Requiring $J$ to be real, it is straightforward to see (cf. around eq. \eqref{PrincipalSpinors}) 
that the principal spinors must be complex. Since they are complex conjugates of each other, i.e. $\beta_{A}=\bar{\alpha}_{A}$, 
$J$ is parametrized by only one spinor $\alpha_{A}$, and one finds
\begin{equation}
 J^{a}{}_{b} = -\frac{i}{(\alpha_{C}\bar{\alpha}^{C})}(\alpha^{A}\bar{\a}_{B}+\bar{\a}^{A}\a_{B})\delta^{A'}{}_{B'}
\end{equation}
(note that this is real since $\overline{(\alpha_{C}\bar{\alpha}^{C})}=-(\alpha_{C}\bar{\alpha}^{C})$).
We have the freedom $\alpha_{A}\to\lambda\alpha_{A}$ with $\lambda\in\mathbb{C}^{\times}$, and we must exclude 
the set $\bar{\alpha}_{A}\propto\alpha_{A}$, so it follows that the set of Hermitian structures at a point is 
$\mathbb{CP}^{1}\backslash\mathbb{RP}^{1}$. Topologically this is $S^{2}\backslash S^{1}$, 
which is homemorphic to a hyperboloid of {\em two} sheets. 
This can also be seen by simply requiring $a^2-b^2+c^2=-1$ in \eqref{hyperboloid}. 
We will invoke this result in section \ref{Sec:twistorsinliterature} below.
\end{remark}

\subsubsection{Lorentz signature}

As observed in Remark \ref{Remark:reality}, in this case the elements of the set \eqref{SpaceParaHermitian} 
are necessarily complex-valued, and $P^{(L)}_{x}$ is simply given by \eqref{SpaceParaHermitian2}. 
Topologically, this is the product $S^{2}\times S^{2}$ of two 2-spheres with the ``diagonal'' $\a=\b$ removed.
We can obtain an alternative description of this structure
by considering again the three maps \eqref{K1}, \eqref{K2}, \eqref{K3}.
These tensors satisfy \eqref{triplet}, but now they are complex-valued. 
Considering $\mathcal{K}_{(a,b,c)}=aK_1 +b K_2 +cK_3$, equation \eqref{twosphere} still holds, but
now $a,b,c$ are {\em complex} numbers. The condition for $\mathcal{K}_{(a,b,c)}$ to be a 
``Lorentzian para-Hermitian structure'' is thus 
\begin{equation}
 a^2+b^2+c^2=1, \qquad a,b,c\in\mathbb{C}
\end{equation}
which is the complexification of the 2-sphere, $\mathbb{C}S^{2}$.
(In this case it is perhaps not so easy to argue intuitively that the set $(S^{2}\times S^{2})\backslash S^{2}$ 
is a complex 2-sphere.)

Each of the two components in the set \eqref{SpaceParaHermitian}, without any reality conditions, 
is therefore a complexified sphere $\mathbb{C}S^{2}$.
Notice that this is true in any signature, and also for complex 4-manifolds.

\subsection{Gauge freedom and a covariant formalism}\label{sec:CovariantFormalism}

From Remark \ref{Remark:Conformal} and the results of section \ref{Sec:spaceParaHermitian}, we see that
fixing a (para-)complex structure $K$ at a point $x\in M$ is equivalent to fixing (ignoring reality conditions): 
\begin{enumerate}
\item a conformal structure $[g]$, and
\item two (different) points on the Riemann sphere at $x$.
\end{enumerate}
Over an open neighbourhood in $M$, instead of considering point objects we consider {\em fields}, 
i.e. smooth sections of the appropriate bundles.
In practice we will work with representatives, i.e. 
with some metric $g_{ab}$, or its spinor equivalents $\e_{AB}$ and $\e_{A'B'}$, 
and with two non-proportional spinors $\a_{A},\b_{A}$ (which in split signature are real, while in Riemannian 
signature one has $\b_{A}=i\a^{\dagger}_{A}$).
But since the only given data is $K$, the framework should not depend on the representatives chosen.
Here we give the essential points of a formalism that takes into account this issue, namely that is covariant under 
the transformations associated to this ``gauge freedom''.
A more detailed description of this framework is given in Appendix \ref{Appendix:CovariantFormalism}.

We have seen that the explicit expression of $K$ in terms of the representative spinors is \eqref{KGR}, that is:
\begin{equation}\label{KGR1}
 K^{a}{}_{b} = \frac{1}{(\e^{CD}\a_{C}\b_{D})}\e^{AE}(\a_{E}\b_{B}+\b_{E}\a_{B})\delta^{A'}{}_{B'}.
\end{equation}
We are free to assign conformal weights to the spinors $\a_{A},\b_{A}$, i.e. we can impose these spinors 
to change under a conformal transformation of the metric as
\begin{equation}\label{ConformalWeightSpinors}
 \a_{A} \to \Omega^{w_{0}+1}\a_{A}, \qquad \b_{A} \to \Omega^{w_{1}+1}\b_{A}
\end{equation}
for some real numbers $w_{0},w_{1}$. 
Recalling the conformal behavior
\begin{equation}\label{ConformalRescaling}
 \e_{AB} \to \Omega \e_{AB}, \qquad \e_{A'B'} \to \Omega \e_{A'B'}.
\end{equation}
we see explicitly from \eqref{KGR1} that $K$ is invariant under conformal rescalings.
(The conformal weights in \eqref{ConformalWeightSpinors} are chosen so that 
$\a^{A}\to\Omega^{w_{0}}\a^{A}$, $\b^{A}\to\Omega^{w_{1}}\b^{A}$.)

The ``gauge freedom'' is given by
conformal \eqref{ConformalRescaling} and projective \eqref{gaugefreedomK} rescalings.
In Lorentzian and Riemannian signature, 
these transformations define the group\footnote{$\mathbb{R}^{+}$ denotes the multiplicative group of 
positive real numbers, and $\mathbb{C}^{\times}$ ($\mathbb{R}^{\times}$) is the multiplicative group of complex (real) numbers.} 
$G_o=\mathbb{C}^{\times}\times\mathbb{C}^{\times}\times\mathbb{R}^{+}$, 
where $\mathbb{R}^{+}$ corresponds to conformal rescalings \eqref{ConformalRescaling} and 
$\mathbb{C}^{\times}\times\mathbb{C}^{\times}$ to projective rescalings \eqref{gaugefreedomK}. 
In split signature one replaces $\mathbb{C}^{\times}\times\mathbb{C}^{\times}$ by $\mathbb{R}^{\times}\times\mathbb{R}^{\times}$.

\begin{definition}\label{Def:weightedsfields}
We say that a spinor/tensor/scalar field $\varphi^{B...B'...}_{C...C'...}$ is {\em weighted} and has {\em type $(r,r';w)$}, 
where $r,r',w$ are real numbers, if under 
the transformations \eqref{ConformalRescaling}-\eqref{gaugefreedomK} it changes as 
\begin{equation}\label{weightedfields}
 \varphi^{B...B'...}_{C...C'...} \to \Omega^{w}\lambda^{r}\mu^{r'}\varphi^{B...B'...}_{C...C'...}.
\end{equation}
\end{definition}

This transformation law defines a representation of the group $G_{o}$, see Appendix \ref{Appendix:CovariantFormalism}.
Examples of weighted fields are $\alpha^{A},\beta^{A}$ and $\phi=\alpha_{A}\beta^{A}$; we give the specific weights 
for them in table \ref{Table:weights}.

\begin{table}
\begin{center}
\begin{tabular}{|c|c|c|c|}
\hline
 & $w$ & $r$ & $r'$ \\ \hline
$\b^{A}$ & $w_1$ & $0$ & $1$ \\ \hline 
$\b_{A}$ & $w_1+1$ & $0$ & $1$ \\ \hline 
$\a^{A}$ & $w_0$ & $1$ & $0$ \\ \hline
$\a_{A}$ & $w_0+1$ & $1$ & $0$ \\ \hline
$\phi$ & $w_0+w_1+1$ & $1$ & $1$ \\ \hline
\end{tabular}
\end{center}
\caption{Weights of the fields associated to an almost-para Hermitian structure.}
\label{Table:weights}
\end{table}

A non-trivial issue is how to define a derivative operator that is covariant under all the transformations involved,
since, to begin with, in a conformal structure one does not have the usual Levi-Civita connection.
The first point is to apply Lee's construction, seen in Section \ref{Sec:Conformal}, 
to use the natural Weyl connection ${}^{\rm w}\nabla_{a}$ induced by an almost-para Hermitian conformal 
structure $([g],K)$ (Proposition \ref{Prop:Lee}).
Then, after applying the appropriate machinery, one arrives at:

\begin{definition}\label{Def:GHPconformalCV}
Let $(M,[g],K)$ be a real, 4-dimensional, almost-para-Hermitian conformal structure, 
and let $\a_{A},\b_{A}$ be the spinor fields representing $K$ as in \eqref{KGR}, with $\e^{AB}\a_{A}\b_{B}=\phi\neq 0$.
Let ${}^{\rm w}\nabla_{AA'}$ and $f_{AA'}$ be the Weyl connection and Lee form induced by $K$,
and let $\varphi^{B...B'...}_{C...C'...}$ be a spinor field of type $(r,r';w)$.
We define the covariant derivative $\mathcal{C}_{AA'}$ by
\begin{equation}\label{cvC}
 \mathcal{C}_{AA'}\varphi^{B...B'...}_{C...C'...}
 = {}^{\rm w}\nabla_{AA'}\varphi^{B...B'...}_{C...C'...} +(wf_{AA'}+rL_{AA'}+r'M_{AA'})\varphi^{B...B'...}_{C...C'...}
\end{equation}
where 
\begin{align}
 L_{AA'} :={}& \phi^{-1}\b_{B}\nabla_{AA'}\a^{B} + \phi^{-1}\b_{A}f_{A'B}\a^{B} - w_{0} f_{AA'}, \label{connection1} \\
 M_{AA'} :={}& -\phi^{-1}\a_{B}\nabla_{AA'}\b^{B} - \phi^{-1}\a_{A}f_{A'B}\b^{B} - w_{1} f_{AA'} \label{connection2}
\end{align}
(with $w_0$ and $w_1$ defined by \eqref{ConformalWeightSpinors}), and where $\nabla_{AA'}$ is any 
Levi-Civita connection in the conformal class.
\end{definition}

To understand the construction of \eqref{cvC}, we refer to the discussion that leads to eq. \eqref{CVassociatedbundle} 
in Appendix \ref{Appendix:CovariantFormalism}.
Under a gauge transformation \eqref{weightedfields}, $\mathcal{C}_{AA'}$ satisfies
\begin{equation}\label{transfoCD}
 \mathcal{C}_{AA'}\varphi^{B...B'...}_{C...C'...} \to \Omega^{w}\lambda^{r}\mu^{r'}\mathcal{C}_{AA'}\varphi^{B...B'...}_{C...C'...}.
\end{equation}

The projection of $\mathcal{C}_{AA'}$ into the eigenbundles of $K$ deserves special attention, so 
we introduce:
\begin{definition}
Let $(M,[g],K)$ be a real, 4-dimensional, almost-para-Hermitian conformal structure, 
and let $\a_{A},\b_{A}$ be the spinor fields representing $K$ as in \eqref{KGR}. 
Let $\mathcal{C}_{AA'}$ be the associated covariant derivative, defined by \eqref{cvC}.
We define the operators
\begin{equation}
 \tilde{\mathcal{C}}_{A'}:=\a^{A}\mathcal{C}_{AA'}, \qquad \mathcal{C}_{A'}:=\b^{A}\mathcal{C}_{AA'}. 
 \label{CProjections}
\end{equation}
\end{definition}

\begin{remark}
Let $(g,K)$ be an almost-para-Hermitian structure, let $P$ and $\tilde{P}$ be the associated projectors 
\eqref{projectors} to its eigenbundles, and let $\a^{A},\b^{A}$ be the spinor fields representing $K$ (eq. \eqref{KGR}).
Furthermore, let $\nabla$ be an arbitrary derivative operator, and $X\in\Gamma(TM\otimes\mathbb{C})$. 
In the modern para-Hermitian approach to double field theory, see in particular
\cite{FRS2017} and \cite{Svo18}, an especially important role is played by the operators $\nabla_{\tilde{P}(X)}$ and $\nabla_{P(X)}$ 
(in arbitrary dimensions).
These operators are also crucial in the present work, and in our notation we have the equivalence:
\begin{subequations}
\begin{align}
 & \tilde{\nabla}_{A'}:=\a^{A}\nabla_{AA'} \quad \Leftrightarrow \quad \nabla_{\tilde{P}(X)}, \\
 & \nabla_{A'}:=\b^{A}\nabla_{AA'} \quad \Leftrightarrow \quad \nabla_{P(X)}.
\end{align}
\end{subequations}
In particular, $\tilde{\mathcal{C}}_{A'}$ corresponds to $\mathcal{C}_{\tilde{P}(X)}$, and 
$\mathcal{C}_{A'}$ corresponds to $\mathcal{C}_{P(X)}$.
\end{remark}

\section{Integrable structures}\label{Sec:Integrability}

\subsection{Special spinor fields}\label{Sec:specialspinors}

The characterization given by Theorem \ref{Theorem:spaceParaHermitian} of the space of almost para-Hermitian structures
allows immediately an explicit description of the eigenbundles:
given $K$, there exist projective spinors $\a^{A},\beta^{A}$ such that eq. \eqref{KGR} (or \eqref{KEuclidean} in 
the Riemmanian case) holds, and from this expression one can readily check that the $\pm 1$ eigenbundles are respectively
\begin{align}
 L ={}& \{u^{a}\in TM\otimes\mathbb{C} \;|\; u^{a}=\beta^{A}\nu^{A'}, \;\; \nu^{A'}\in\mathbb{S}'(0,-1;-w_1) \}, \label{LGR} \\
 \tilde{L} ={}& \{v^{a}\in TM\otimes\mathbb{C} \;|\; v^{a}=\alpha^{A}\mu^{A'}, \;\; \mu^{A'}\in\mathbb{S}'(-1,0;-w_0)\}, \label{LtildeGR}
\end{align}
where $\beta^{A}, \alpha^{A}$ are fixed and $\nu^{A'}, \mu^{A'}$ vary,
and where $\mathbb{S}'(r,r';w)$ represents the weighted bundles described in Appendix \ref{Appendix:CovariantFormalism}.
We have to include these weighted bundles because, since $\beta^{A}$ and $\alpha^{A}$ have non-trivial weights 
(see table \ref{Table:weights}), 
the spinors $\nu^{A'}, \mu^{A'}$ should also have non-trivial weights so that the 
elements of $L,\tilde{L}$ are ordinary vectors. 

As discussed in Section \ref{Sec:paraHermitian}, the integrability properties of an almost para-Hermitian structure 
like \eqref{KGR} refer to whether the distributions $L$ and $\tilde{L}$ are involutive.
In particular, according to Definition \ref{Def:LparaHermitian}, $K$ is half-integrable if and only if 
$N_{P} \equiv 0$ or $N_{\tilde{P}}\equiv 0$, and integrable if and only if $N_{P}=N_{\tilde{P}}=0$. 
For the case of \eqref{KGR}, the projectors \eqref{projectors} to $L$ and $\tilde{L}$ are respectively
\begin{subequations}
\begin{align}
 P^{a}{}_{b} ={}& \tfrac{1}{2}(\delta^{a}{}_{b}+K^{a}{}_{b}) = \phi^{-1}\b^{A}\a_{B}\delta^{A'}{}_{B'}, \label{projector4d} \\
 \tilde{P}^{a}{}_{b} ={}& \tfrac{1}{2}(\delta^{a}{}_{b}-K^{a}{}_{b}) = -\phi^{-1}\a^{A}\b_{B}\delta^{A'}{}_{B'}. \label{projectortilde4d}
\end{align}
\end{subequations}
where $\phi=\alpha_{A}\beta^{A}$.

\begin{lemma}\label{Lemma:paraHermitian}
Consider an almost (para-)Hermitian structure on a real 4-manifold $M$ (Def. \ref{Def:almostParaHermitian}),
and let $\nabla$ be the Levi-Civita connection of $g$. For concreteness suppose that the (para-)complex structure 
is of negative chirality (see Remark \ref{Remark:chirality}). Then the (para-)Hermitian structure is
\begin{enumerate}
\item half-integrable if and only if there exists a non-trivial unprimed spinor field $\varphi^{A}$ such that 
\begin{equation}\label{ShearFree}
 \varphi^{A}\varphi^{B}\nabla_{AA'}\varphi_{B} = 0.
\end{equation} 
\item integrable if and only if there are two linearly independent solutions to equation \eqref{ShearFree}.
\end{enumerate}
\end{lemma}

\begin{proof}
Denote the (para-)complex structure by $K$, and suppose it has negative chirality.
Then $K$ can be written as \eqref{KGR} for some (projective) spinor fields $\a_{A},\b_{A}$ 
(where in the Riemannian case we have $\b_{A}=i\a^{\dagger}_{A}$).
Then the projectors $P,\tilde{P}$ take the form \eqref{projector4d}-\eqref{projectortilde4d}.
Now let us compute $N_{P}(X,Y)$ for arbitrary $X,Y$, using \eqref{projector4d}-\eqref{projectortilde4d}. We have:
\begin{align*}
 (N_{P}(X,Y))^{AA'} ={}& -\phi^{-1}\a^{A}\b_{B}[PX,PY]^{BA'}\\
 ={}&-\phi^{-1}\a^{A}\beta_{B}\left(\phi^{-1}\b^{C}\a_{D}X^{DC'}\nabla_{CC'}(\phi^{-1}\b^{B}\a_{E}Y^{EA'}) \right. \\
  &  \left. - \phi^{-1}\b^{C}\a_{D}Y^{DC'}\nabla_{CC'}(\phi^{-1}\b^{B}\a_{E}X^{EA'})\right) \\
 ={}& -\phi^{-3}\a^{A}\beta_{B} 
(\b^{C}\a_{D}X^{DC'}\a_{E}Y^{EA'}\nabla_{CC'}\b^{B}-\b^{C}\a_{D}Y^{DC'}\a_{E}X^{EA'}\nabla_{CC'}\b^{B}) \\
 ={}& -\phi^{-3}\a^{A}\b_{B}\b^{C}\a_{E}\a_{D}(X^{DC'}Y^{EA'}-X^{DA'}Y^{EC'})\nabla_{CC'}\b^{B} \\
 ={}& -\phi^{-3}(\a_{E}\a_{D}X_{D}{}_{E'}Y^{EE'})\a^{A}\b^{B}\b^{C}\nabla_{C}{}^{A'}\b_{B},
\end{align*}
therefore we see that $N_{P} \equiv 0$ if and only if $\beta^{A}$ satisfies equation \eqref{ShearFree}, 
which proves the first item. 
An analogous calculation applies to $\a^{A}$, and this shows that $N_{\tilde{P}} \equiv 0$ iff $\a^{A}$
satisfies \eqref{ShearFree}. 
Therefore, the condition $N_{P}=N_{\tilde{P}}\equiv 0$ is equivalent to the condition that both spinors $\a_{A}$
and $\b_{B}$ satisfy eq. \eqref{ShearFree}; this proves the second item.
\end{proof}

Not any 4-manifold admits non-trivial solutions to \eqref{ShearFree}, so half-integrability of $K$ 
imposes restrictions. We may call equation \eqref{ShearFree} the ``shear-free condition'', since in 
Lorentzian general relativity its solutions define what are known as shear-free null geodesic congruences: 
the real vector field $\ell^{a}=\varphi^{A}\bar\varphi^{A'}$ is tangent to a null geodesic congruence which is 
shear-free (see e.g. \cite[Section 9.2]{Wald} for the definition of shear of a geodesic congruence).
If, for example, one also asks the metric to be Einstein, then by the Goldberg-Sachs theorem the Weyl 
tensor must be algebraically special, which is a restriction on the curvature.
Still, many physically interesting spacetimes in general relativity satisfy this property; in particular, 
stationary black hole solutions such as Kerr or Schwarzschild (with or without cosmological constant).

In the DFT language, a choice of $K$ means a choice of T-dual `spacetimes', where the physical spacetime 
is related to the involutive eigenbundle.
If we change $K$, the new $K$ might not be half-integrable so it would not define such T-dual spaces.
Requiring any $K$ to be ``appropriate'' in this sense imposes more severe restrictions:

\begin{lemma}\label{Lemma:SDWeyl}
Let $M$ be a real, 4-dimensional, orientable manifold equipped with a metric $g$. 
If all almost-para-Hermitian structures (of a definite chirality) are half-integrable, then the Weyl curvature of $g$ is (A)SD.
\end{lemma}

\begin{proof}
Let $K$ be an arbitrary almost para-Hermitian structure, and assume for concreteness that it has negative chirality.
Then we know that it is represented by two projective spinor fields $\a^{A},\b^{A}$. 
We also know that $K$ is half-integrable if and only if one of these spinor fields satisfies eq. \eqref{ShearFree}.
Let $\varphi^{A}$ denote any of these spinors.
From \eqref{ShearFree} we deduce that $\varphi^{A}\nabla_{AA'}\varphi_{B} = \varphi_{B}\pi_{B'}$ 
for some spinor field $\pi_{B'}$.
Applying now $\varphi^{C}\nabla_{C}{}^{A'}$ to equation \eqref{ShearFree}, one obtains the following 
integrability condition:
\begin{align}
\nonumber 0 ={}& \varphi^{C}\nabla_{C}{}^{A'}(\varphi^{A}\varphi^{B}\nabla_{AA'}\varphi_{B}) \\
\nonumber ={}& 2\pi^{A'}\varphi^{A}\varphi^{B}\nabla_{AA'}\varphi_{B} -\varphi^{C}\varphi^{A}\varphi^{B}\Box_{AC}\varphi_{B} \\
 ={}& -\Psi_{ABCD}\varphi^{A}\varphi^{B}\varphi^{C}\varphi^{D} \label{IntegrabilityConditionSFR}
\end{align}
where $\Box_{AC}=\nabla_{A'(A}\nabla_{C)}{}^{A'}$, and $\Psi_{ABCD}$ is the Weyl curvature spinor 
(see \cite{PR1} for details).
If any almost para-Hermitian structure $K$ is half-integrable, then \eqref{IntegrabilityConditionSFR} must be satisfied by 
all spinors $\varphi^{A}$, which implies that $\Psi_{ABCD}\equiv 0$ and thus the Weyl curvature of $g$ is SD.
(Choosing positive chirality instead leads to ASD Weyl curvature.)
\end{proof}

\begin{remark}\label{Remark:anyKintegrable}
When saying that ``any almost-para-Hermitian structure is {\em half}-integrable'', it is convenient to 
be more specific and say, using the terminology mentioned in Remark \ref{Remark:LparaHermitian},
that ``the manifold is $L$-para Hermitian with respect to any $K$'', i.e. the $(+1)$-eigenbundle of any $K$ is involutive.
From this perspective, if a manifold is $L$-para Hermitian with respect to any $K$, then 
it follows that it is para-Hermitian. 
To see this: let $K$ be an arbitrary $L$-para Hermitian structure, i.e. its $(+1)$-eigenbundle is involutive.
If any other $K$ is also $L$-para Hermitian, then in particular this is true for $K'=-K$. The $(+1)$-eigenbundle 
of $K'$ is then involutive, and it is the $(-1)$-eigenbundle of $K$. 
So $K$ is both $L$-para Hermitian and $\tilde{L}$-para Hermitian, i.e., it is para-Hermitian.
\end{remark}

The (A)SD restriction in lemma \ref{Lemma:SDWeyl} is quite strong.
In Lorentz signature, the left- and right-handed Weyl curvature spinors are complex conjugates of each other, 
which implies that (A)SD Weyl curvature means vanishing Weyl tensor, so the manifold must be conformally flat. 
In Riemannian and split signature, the two Weyl curvature spinors are independent from each other, so one of them can vanish 
while the other one is general. 
In the Riemannian case, if one also asks Ricci-flatness (as well as (A)SD Weyl curvature), then it can be shown 
that the manifold must be hyperk\"ahler.

In the DFT literature, a change of para-complex structure is referred to as a {\em change of polarization}.
As mentioned, an arbitrary change of polarization in general does not lead to a new pair of T-dual `spacetimes', 
unless all para-complex structures are half-integrable.
Another possibility is to consider ``small'' changes of polarization, by which we mean 
deformations of the para-complex structure. 
These issues will be addressed in section \ref{Sec:Twistors}, in connection also with twistor constructions.

\subsection{Lie and Courant algebroids}\label{Sec:algebroids4D}

We have seen that, given a 4-dimensional real manifold $M$ with a metric $g$, any $K\in{\rm Aut}(TM\otimes\mathbb{C})$ 
satisfying conditions \eqref{paraHermitian1} and \eqref{paraHermitian2} is characterized by two projective 
spinor fields $[\a_{A}], [\b_{A}]$. 
The tangent bundle splits as $TM\otimes\mathbb{C} = L\oplus \tilde{L}$,
where $L$ and $\tilde{L}$ are the eigenbundles of $K$, given explicitly in eqs. \eqref{LGR}, \eqref{LtildeGR}.
Because of this splitting, any vector in $TM\otimes\mathbb{C}$ can be decomposed into two pieces, 
as originally done in \eqref{splitvector0}.
We can explicitly describe these pieces using \eqref{LGR}, \eqref{LtildeGR}, 
the identity $P+\tilde{P}=\mathbb{I}$ and the fact that the projectors $P$ and $\tilde{P}$ are 
given by \eqref{projector4d}, \eqref{projectortilde4d}. 
That is, any $X\in\Gamma(TM\otimes\mathbb{C})$ is decomposed as
\begin{equation}\label{splitvector1}
 X^{AA'} = (PX)^{AA'} + (\tilde{P}X)^{AA'} \equiv x^{AA'} + \tilde{x}^{AA'},
\end{equation}
where 
\begin{align}
 & (PX)^{AA'} \equiv x^{AA'} = \beta^{A}x^{A'}, \qquad x^{A'}=\phi^{-1}\alpha_{B}X^{BA'}, \label{CompSplitVector1} \\
 & (\tilde{P}X)^{AA'} \equiv \tilde{x}^{AA'} = \alpha^{A}\tilde{x}^{A'}, \qquad \tilde{x}^{A'}=-\phi^{-1}\beta_{B}X^{BA'}. \label{CompSplitVector2}
\end{align}
Analogous decompositions and notation apply to any other vector $Y^{AA'}=y^{AA'}+\tilde{y}^{AA'}$, etc. 
Similarly, a 1-form $W\in\Gamma(T^{*}M\otimes\mathbb{C})$ is decomposed as 
\begin{equation}\label{split1form}
 W_{AA'}=\omega_{AA'} +\tilde{\omega}_{AA'}= \alpha_{A}\omega_{A'} + \beta_{A}\tilde{\omega}_{A'}, 
\end{equation}
with 
\begin{align*}
 \omega_{A'} ={}& \phi^{-1}\beta^{B}\omega_{BA'}, \\
 \tilde{\omega}_{A'} ={}& -\phi^{-1}\alpha^{B}\omega_{BA'}.
\end{align*}

\begin{remark}\label{Remark:weightsvector}
Ordinary tensor fields have trivial $r,r',w$ weights (in the notation of Definition \ref{Def:weightedsfields}), 
but the spinors $\alpha^{A}, \beta^{A}$ have non-trivial weights (given in Table \ref{Table:weights}).
Therefore, when decomposed into pieces like \eqref{splitvector1} or \eqref{split1form},
the primed spinor parts are also non-trivially weighted.
For example, the spinors $x^{A'}$ and $\tilde{x}^{A'}$ in the decomposition \eqref{CompSplitVector1}-\eqref{CompSplitVector2}
have non-trivial weights, which can be deduced from those of $\phi$ and $\alpha^{A},\beta^{A}$.
This is analogous to the fact that $\nu^{A'},\mu^{A'}$ in \eqref{LGR}, \eqref{LtildeGR} are also weighted.
\end{remark}

We now use the above decomposition to describe the Lie and Courant algebroid structures of 4-manifolds 
admitting non-trivial solutions to \eqref{ShearFree}.
Since this equation involves only a {\em conformal} structure and a {\em projective} spinor field, 
the covariant formalism of Section \ref{sec:CovariantFormalism} is here appropriate to give a formulation 
invariant under the associated gauge freedom.
We recall that the metric is allowed to have any signature.

\begin{lemma}\label{Lemma:LieAlgGR}
Let $(M,[g_{ab}])$ be a 4-dimensional, real conformal structure.
Let $g_{ab}\in[g_{ab}]$, with Levi-Civita connection $\nabla_{AA'}$.
Suppose that there is a projective spinor field $[\b_{A}]$, where any representative $\beta_{A}$ satisfies the equation 
\begin{equation}
 \b^{A}\b^{B}\nabla_{AA'}\b_{B}=0. \label{SFR1}
\end{equation}
Let $\a_{A}$ be any spinor field such that $\e^{AB}\a_{A}\b_{B}\equiv\phi\neq0$.
Let $L\subset TM\otimes\mathbb{C}$ be defined by \eqref{LGR}, 
and let $\mathcal{C}_{A'}$ be the operator introduced in \eqref{CProjections}.
\begin{enumerate}
\item The triple $(L,[\cdot,\cdot]_{L},P)$ is a Lie algebroid, where $P$ is the projector \eqref{projector4d},
and the Lie algebroid bracket is
\begin{equation}\label{LieAlgBracket4D}
  [X,Y]^{AA'}_{L} = \beta^{A} (x^{B'}\mathcal{C}_{B'}y^{A'} - y^{B'}\mathcal{C}_{B'}x^{A'}),
\end{equation}
for any vector fields $X,Y\in\Gamma(L)$.
\item Let $\Lambda^{k}=\wedge^{k}L^{*}, k=0,1,2,...$. The Lie algebroid exterior derivative of $(L,[\cdot,\cdot]_{L},P)$ is given by
\begin{align}
 ({\rm d}^{L}f)_{AA'} ={}& \phi^{-1}\a_{A}\mathcal{C}_{A'}f, \label{LieAlgExtDerGR} \\
  ({\rm d}^{L}\omega)_{AA'BB'} ={}& \phi^{-1}(\mathcal{C}_{C'}\omega^{C'})\a_{A}\a_{B}\epsilon_{A'B'}, \label{LieAlgExtDerGR1}
\end{align}
and ${\rm d}^{L}\omega=0$ for any $\omega\in\Gamma(\Lambda^{k})$ with $k\geq2$, where $\omega^{C'}$ 
in the right hand side of \eqref{LieAlgExtDerGR1} refers to the decomposition \eqref{split1form}.
\item For a 1-form $\omega\in\Gamma(L^{*})$, the generalized Lie derivative along $X\in\Gamma(L)$ is
\begin{equation}
 (\pounds^{L}_{X}\omega)_{AA'} = 
 \a_{A}[x_{A'}(\mathcal{C}_{B'}\omega^{B'})+\mathcal{C}_{A'}(x^{B'}\omega_{B'})].
 \label{GenLieDerivativeGR}
\end{equation}
\end{enumerate}
\end{lemma}

\begin{proof}
For the first item, we need to check that the conditions in Definition \ref{Def:LieAlgebroid} are satisfied. 
First of all, we know that $L$ is a vector bundle over $M$. 
We have also seen that the condition \eqref{SFR1} implies that sections of $L$ are in involution with 
respect to the usual Lie bracket $[\cdot,\cdot]$ of vector fields, thus its restriction to $L$, $[\cdot,\cdot]_{L}$, 
indeed maps $\Gamma(L)\times\Gamma(L)$ to $\Gamma(L)$, is skew-symmetric and satisfies the 
Jacobi identity. 
In addition, the identity in $TM$ restricted to $L$ coincides with $P$, and it satisfies the conditions to be an anchor.
So it only remains to prove the expression \eqref{LieAlgBracket4D} for the Lie algebroid bracket. To do this,
notice first that, since ordinary vector fields $X,Y$ have trivial $w,r,r'$ weights (see Remark \ref{Remark:weightsvector}), 
we have for example
\begin{equation*}
 \mathcal{C}_{b}X^{a} = \nabla_{b}X^{a} + Q_{bc}{}^{a}X^{c},
\end{equation*}
where $Q_{bc}{}^{a}$ was defined in \eqref{WeylConnection2}. Since $Q_{bc}{}^{a}=Q_{(bc)}{}^{a}$, we get
\begin{equation*}
 X^{b}\nabla_{b}Y^{a}-Y^{b}\nabla_{b}X^{a}  = X^{b}\mathcal{C}_{b}Y^{a}-Y^{b}\mathcal{C}_{b}X^{a}.
\end{equation*}
Using this and decompositions like \eqref{splitvector1}-\eqref{CompSplitVector1}-\eqref{CompSplitVector2} for $X,Y\in\Gamma(L)$,
we have:
\begin{align*}
 [X,Y]^{AA'}_{L} ={}& X^{BB'}\nabla_{BB'}Y^{AA'}-Y^{BB'}\nabla_{BB'}X^{AA'} \\
 ={}& x^{B'}\b^{B}\mathcal{C}_{BB'}(\b^{A}y^{A'})-y^{B'}\b^{B}\mathcal{C}_{BB'}(\b^{A}x^{A'}) \\
={}& \beta^{A} (x^{B'}\mathcal{C}_{B'}y^{A'} - y^{B'}\mathcal{C}_{B'}x^{A'})
\end{align*}
where in the third line we used \eqref{Cbeta} and \eqref{SFR1}, together with the definition \eqref{CProjections}
for $\mathcal{C}_{A'}$.

Now let us prove the second item. From Definition \ref{Def:LieAlgDerivatives}, the Lie algebroid 
exterior derivative acting on a 0-form $f\in\Gamma(\Lambda^{0})$ is given by $({\rm d}^{L}f)(X)=P(X)(f)$ 
(since here the anchor is the projector $P$).
Since $f$ has trivial weights, we can replace the ordinary $\partial_{AA'}$
implicit in this formula by the covariant derivative $\mathcal{C}_{AA'}$, so we get
\begin{equation*}
 ({\rm d}^{L}f)(X) = (PX)^{AA'}\mathcal{C}_{AA'}f.
\end{equation*}
Using \eqref{CompSplitVector1}, formula \eqref{LieAlgExtDerGR} follows.
Consider now a 1-form $\omega\in\Gamma(\Lambda^{1})$. 
The operator ${\rm d}^{L}$ acting on $\omega$ is given by the expression \eqref{LieAlgExtDer1}, with $\rho\equiv P$. 
Since $\omega,X,Y$ have trivial weights we can, as before, replace $\partial_{AA'}$ by $\mathcal{C}_{AA'}$:
\begin{equation*}
 ({\rm d}^{L}\omega)(X,Y) = (PX)^{AA'}\mathcal{C}_{AA'}(\omega_{BB'}Y^{BB'})
 -(PY)^{AA'}\mathcal{C}_{AA'}(\omega_{BB'}X^{BB'}) - \omega_{AA'}[X,Y]^{AA'}_{L}.
\end{equation*}
Replacing the expressions for $X,Y$ and $[X,Y]_{L}$, this is
\begin{align}
\nonumber  ({\rm d}^{L}\omega)(X,Y)  ={}& x^{B'} \mathcal{C}_{B'}(\phi\omega_{A'}y^{A'})-y^{B'} \mathcal{C}_{B'}(\phi\omega_{A'}x^{A'})
  -\phi\omega_{A'} (x^{B'}\mathcal{C}_{B'}y^{A'} - y^{B'}\mathcal{C}_{B'}x^{A'}) \\
\nonumber ={}& \phi \left[ x^{B'} \mathcal{C}_{B'}(\omega_{A'}y^{A'})-y^{B'} \mathcal{C}_{B'}(\omega_{A'}x^{A'})
  -\omega_{A'} x^{B'}\mathcal{C}_{B'}y^{A'} +\omega_{A'} y^{B'}\mathcal{C}_{B'}x^{A'} \right] \\
  ={}& \phi (x^{A'}y^{B'}-y^{A'}x^{B'})\mathcal{C}_{A'}\omega_{B'} \label{LieAlgExtDerGR1-aux}
\end{align}
where in the second line we used \eqref{Cchi}. Now: 
\begin{align*}
\nonumber x^{A'}y^{B'}-y^{A'}x^{B'} ={}& (\epsilon_{C'D'}x^{C'}y^{D'})\epsilon^{A'B'} \\
\nonumber ={}&(\epsilon_{C'D'}\phi^{-2}\alpha_{C}X^{CC'}\alpha_{D}Y^{DD'})\epsilon^{A'B'} \\
 ={}& (\phi^{-2}\alpha_{C}\alpha_{D}\epsilon_{C'D'}\epsilon^{A'B'}) X^{CC'}Y^{DD'}
\end{align*}
where in the second line we replaced the expression \eqref{CompSplitVector1} for $x^{A'}$, and similarly for $y^{A'}$.
Replacing now in \eqref{LieAlgExtDerGR1-aux}, and noticing that 
\begin{equation*}
  ({\rm d}^{L}\omega)(X,Y) =  ({\rm d}^{L}\omega)_{CC'DD'}X^{CC'}Y^{DD'},
\end{equation*}
formula \eqref{LieAlgExtDerGR1} follows. 
If now $\omega\in\Gamma(\Lambda^{k})$, $k\geq2$, then ${\rm d}^{L}\omega$ vanishes 
since it is a $(k+1)$-form with $k\geq 2$ in a 2-dimensional space.

Finally, to show the third item, from the general expression \eqref{GeneralizedLieDerivative} we have
\begin{equation*}
 (\pounds^{L}_{X}\omega)_{AA'} = X^{BB'}({\rm d}^{L}\omega)_{BB'AA'} + [{\rm d}^{L}(X^{BB'}\omega_{BB'})]_{AA'}.
\end{equation*}
Using the second item, this is
\begin{align*}
\nonumber (\pounds^{L}_{X}\omega)_{AA'} ={}& X^{BB'}\phi^{-1}\alpha_{B}\alpha_{A}\epsilon_{B'A'}\mathcal{C}_{C'}\omega^{C'}
 + \phi^{-1}\alpha_{A}\mathcal{C}_{A'}(\phi x^{B'}\omega_{B'}) \\
 ={}& \alpha_{A}x_{A'}\mathcal{C}_{C'}\omega^{C'}+\alpha_{A}\mathcal{C}_{A'}(x^{B'}\omega_{B'}),
\end{align*}
which proves \eqref{GenLieDerivativeGR}.
\end{proof}

\begin{lemma}\label{Lemma:CourantAlgGR}
Let $(M,[g_{ab}])$ be a 4-dimensional, real conformal structure.
Let $g_{ab}\in[g_{ab}]$, with Levi-Civita connection $\nabla_{AA'}$.
Suppose that there is a projective spinor field $[\b_{A}]$, where any representative $\beta_{A}$ satisfies the equation 
\begin{equation}
 \b^{A}\b^{B}\nabla_{AA'}\b_{B}=0. \label{SFR2}
\end{equation}
Let $\a_{A}$ be any spinor field such that $\e^{AB}\a_{A}\b_{B}=\phi\neq0$.
Then the quadruple
\begin{equation}\label{CourantGR}
 (TM\otimes\mathbb{C}, g, P, \llb \cdot,\cdot \rrb)
\end{equation}
is a Courant algebroid, where the anchor $P$ and the Dorfman bracket are respectively
\begin{align}
 P^{AA'}{}_{BB'} ={}& \phi^{-1}\beta^{A}\a_{B}\delta^{A'}{}_{B'}, \\
 \llb X,Y \rrb^{AA'} ={}& 
 \beta^{A} (x^{B'}\mathcal{C}_{B'}y^{A'} - y^{B'}\mathcal{C}_{B'}x^{A'})
 +\a^{A} \left( x^{A'}\mathcal{C}_{B'}\tilde{y}^{B'} - y^{A'}\mathcal{C}_{B'}\tilde{x}^{B'} 
 +\mathcal{C}^{A'}(x^{B'}\tilde{y}_{B'}) \right) \label{Dorfman4D}
\end{align}
with $\mathcal{C}_{A'}$ the operator defined in \eqref{CProjections}.
\end{lemma}

\begin{proof}
Define the bundle $L\subset TM\otimes\mathbb{C}$ by \eqref{LGR}; then
from Lemma \ref{Lemma:LieAlgGR}, the triple $(L,[\cdot,\cdot]_{L},P)$ is a Lie algebroid.
Thus, applying Proposition \ref{Prop:CourantAlgebroid}, the vector bundle $L\oplus L^{*}$ 
together with the maps defined by \eqref{anch}, \eqref{innerp} and \eqref{Dorf}, constitute  
a Courant algebroid. 
Using now the pair of spinors $(\a_{A},\b_{A})$ to define, via \eqref{KGR}, a tensor field $K$ satisfying 
conditions \eqref{paraHermitian1} and \eqref{paraHermitian2}, the eigenbundles of $K$ are given by 
\eqref{LGR} and \eqref{LtildeGR} and we have the splitting $TM\otimes\mathbb{C}=L\oplus\tilde{L}$.
We can then apply Proposition \ref{Prop:Vaisman}, so the quadruple \eqref{CourantGR} is also a Courant algebroid.
The anchor is simply the projector \eqref{projector4d}.
What remains now is to translate the expression \eqref{dorfTB} for the Dorfman bracket to the language we are using in this section.
Using \eqref{CompSplitVector1}, \eqref{CompSplitVector2}, \eqref{LieAlgExtDerGR} and \eqref{GenLieDerivativeGR}, we have
\begin{align}
 & (\tilde{x}_{\flat})_{AA'} = \epsilon_{AB}\epsilon_{A'B'}\tilde{x}^{BB'} = \a_{A}\tilde{x}_{A'}, \\
 & (\tilde{y}_{\flat})_{AA'} = \epsilon_{AB}\epsilon_{A'B'}\tilde{y}^{BB'} = \a_{A}\tilde{y}_{A'}, \\
 & g(\tilde{x},y) = \epsilon_{AB}\epsilon_{A'B'}\tilde{x}^{AA'}y^{BB'} = \phi\tilde{x}_{A'}y^{A'}, \\
 & ({\rm d}^{L}g(\tilde{x},y))_{AA'} = \a_{A}\mathcal{C}_{A'}(\tilde{x}_{B'}y^{B'}), \label{AuxDorf4D1} \\
 & [\pounds^{L}_{x}(\tilde{y}_{\flat})]_{AA'} = \a_{A}[x_{A'}\mathcal{C}_{B'}\tilde{y}^{B'} 
 +\mathcal{C}_{A'}(x^{B'}\tilde{y}_{B'})]. \label{AuxDorf4D2}
\end{align}
The Dorfman bracket \eqref{dorfTB} is
\begin{equation*}
 \llbracket X,Y \rrbracket^{AA'} = [x,y]^{AA'} + (\pounds^{L}_{x}(\tilde{y}_{\flat})-\pounds^{L}_{y}(\tilde{x}_{\flat})+{\rm d}^{L}g(\tilde{x},y))^{AA'}.
\end{equation*}
Since $x,y$ are sections of $L$, we have $[x,y]^{AA'}=[x,y]^{AA'}_{L}$, so we can use formula \eqref{LieAlgBracket4D}.
Replacing also \eqref{AuxDorf4D1}, \eqref{AuxDorf4D2}, we get
\begin{align*}
\llbracket X,Y \rrbracket^{AA'} ={}&  \beta^{A} (x^{B'}\mathcal{C}_{B'}y^{A'} - y^{B'}\mathcal{C}_{B'}x^{A'}) 
 +  \a^{A}[x^{A'}\mathcal{C}_{B'}\tilde{y}^{B'} +\mathcal{C}^{A'}(x^{B'}\tilde{y}_{B'})]  \\
& - \a^{A}[y^{A'}\mathcal{C}_{B'}\tilde{x}^{B'} +\mathcal{C}^{A'}(y^{B'}\tilde{x}_{B'})] 
+\a^{A}\mathcal{C}^{A'}(\tilde{x}_{B'}y^{B'}) \\
 ={}& \beta^{A} (x^{B'}\mathcal{C}_{B'}y^{A'} - y^{B'}\mathcal{C}_{B'}x^{A'}) 
  +\alpha^{A}\left( x^{A'}\mathcal{C}_{B'}\tilde{y}^{B'} - y^{A'}\mathcal{C}_{B'}\tilde{x}^{B'} + \mathcal{C}^{A'}(x^{B'}\tilde{y}_{B'})\right)
\end{align*}
which is \eqref{Dorfman4D}.
\end{proof}

Recall from Remark \ref{Remark:deRham} that a Lie algebroid has naturally associated a 
cochain complex $(\Gamma(\Lambda^{\bullet}),{\rm d}^{L})$.
In our case this is
\begin{equation}\label{derham0}
 0 \to \Gamma(\Lambda^{0}) \xrightarrow{{\rm d}^{L}} \Gamma(\Lambda^{1}) 
 \xrightarrow{{\rm d}^{L}} \Gamma(\Lambda^{2}) \to 0.
\end{equation}
We can think of this as a (twisted) de Rham complex. Since a twisted de Rham complex is locally exact, if an 
object $\varphi$ satisfies ${\rm d}^{L}\varphi=0$, then there exists, {\em locally}\footnote{Of course, {\em global} 
existence of potentials is a different story, where one has to take into account the cohomology of the complex \eqref{derham0}.}, 
another field $\psi$ such that $\varphi={\rm d}^{L}\psi$. 
Therefore, the Lie algebroid structure leads to a ``potentialization'' scheme. 
For example, suppose $W_{AA'}=\alpha_{A}\omega_{A'}\in\Gamma(\Lambda^{1})$ satisfies 
${\rm d}^{L}W=0$, then there exists $f\in\Gamma(\Lambda^{0})$ such that $W={\rm d}^{L}f$.
Using \eqref{LieAlgExtDerGR}-\eqref{LieAlgExtDerGR1}, the equation ${\rm d}^{L}W=0$ 
is $\mathcal{C}_{A'}\omega^{A'}=0$, and the equation $W={\rm d}^{L}f$ is $W_{AA'}=\phi^{-1}\alpha_{A}\mathcal{C}_{A'}f$.
In other words: 
\begin{equation}\label{potentials}
 \mathcal{C}_{A'}\omega^{A'}=0 \quad \Rightarrow \quad \omega_{A'}=\phi^{-1}\mathcal{C}_{A'}f
\end{equation}
for some $f$. Such a basic `potentialization' scheme turns out to be very powerful when applied to problems of interest in relativity 
(this will be the subject of a separate work). 
Notice, however, that the weights of the potentials here are restricted by the fact that the spaces $\Gamma(\Lambda^{k})$ 
in the complex \eqref{derham0} are composed of objects with zero weights.
If we want to study whether the complex \eqref{derham0} can be generalized to fields with other weights, 
we are led to the question of constructing Lie algebroids for weighted fields.
This turns out to be a bit tricky, and will be analysed in the next subsection.

\subsection{Weighted algebroids}\label{Sec:algebroidsweight}

The algebroids constructed in lemmas \ref{Lemma:LieAlgGR} and \ref{Lemma:CourantAlgGR} consist of vectors with 
trivial weights (in the sense of def. \ref{Def:weightedsfields}). 
From the discussion at the end of the previous subsection, we are also interested in 
analysing whether spaces of non-trivially weighted objects can be given an algebroid structure. 
We will see here that ordinary Lie algebroids do not seem to be compatible with weighted fields.

Let $L$ be the bundle \eqref{LGR}, where $\beta^{A}$ satisfies \eqref{ShearFree}.
Consider the weighted vector bundle $L_{(r,r';w)}:=L\otimes\mathbb{S}(r,r';w)$, where $\mathbb{S}(r,r';w)$ 
is the line bundle defined in \eqref{weightedlinebundles}. An element of $L_{(r,r';w)}$ is a vector 
with weights $(r,r';w)$ (so $L$ is simply $L_{(0,0;0)}$):
\begin{equation*}
 X\in L_{(r,r';w)} \quad \Leftrightarrow \quad X^{a}=\beta^{A}x^{A'} 
 \quad \text{and } \quad X\xrightarrow{G_{o}} \lambda^{r}\mu^{r'}\Omega^{w}X.
\end{equation*}
$L_{(r,r';w)}$ is a vector bundle over $M$. In order to construct a Lie algebroid, 
we also need an anchor and a Lie bracket satisfying the conditions given in Def. \ref{Def:LieAlgebroid}.

Let us first focus on the bracket.
Given two weighted vector fields $X\in\Gamma(L_{(r_1,r'_1;w_1)})$, $Y\in\Gamma(L_{(r_2,r'_2;w_2)})$, 
we have a natural candidate for a bracket operation:
\begin{equation}\label{wLiebracket}
 [X,Y]^{a}_{\mathcal{C}} := X^{b}\mathcal{C}_{b}Y^{a} - Y^{b}\mathcal{C}_{b}X^{a}.
\end{equation}
This is skew-symmetric, and, since $\mathcal{C}_{a}$ satisfies \eqref{transfoCD} and $\mathcal{C}_{b}\beta^{A}=0$, 
the result of $[X,Y]^{a}_{\mathcal{C}}$ is again a vector of the form $\beta^{A}\pi^{A'}$ with well-defined weights. 
However, the first problem that arises is that the weights of $[X,Y]^{a}_{\mathcal{C}}$ are 
$(r_1+r_2,r'_1+r'_2;w_1+w_2)$, so $[X,Y]^{a}_{\mathcal{C}}$ is neither an element 
of $L_{(r_1,r'_1;w_1)}$ nor of $L_{(r_2,r'_2;w_2)}$, and consequently, we cannot use $[\cdot,\cdot]_{\mathcal{C}}$ 
as a Lie bracket for a particular $L_{(r,r';w)}$.  
More precisely, what we have is
\begin{equation}\label{gradedbracket}
 [\cdot,\cdot]_{\mathcal{C}}:\Gamma(L_{(r_1,r'_1;w_1)})\times\Gamma(L_{(r_2,r'_2;w_2)})
 \to \Gamma(L_{(r_1+r_2, r'_1+r'_2; w_1+w_2)}).
\end{equation}
This suggests that we consider the vector bundle with graded fibers 
\begin{equation}\label{gradedAlgebroid}
 \mathbb{L}=\bigoplus_{(r,r';w)} L_{(r,r';w)},
\end{equation}
where the sum runs over all possible values of $(r,r';w)$.
We then have $[\cdot,\cdot]_{\mathcal{C}}:\Gamma(\mathbb{L})\times\Gamma(\mathbb{L})\to\Gamma(\mathbb{L})$. 
So the first item in Def. \ref{Def:LieAlgebroid} is satisfied, and we can try to 
give a Lie algebroid structure to $\mathbb{L}$.

The second item in Def. \ref{Def:LieAlgebroid} requires the Jacobi identity for $[\cdot,\cdot]_{\mathcal{C}}$ to hold.
Since $\mathcal{C}_{b}$ {\em cannot} be replaced by $\partial_{b}$ in \eqref{wLiebracket} 
(because $X$ and $Y$ now have non-trivial weights), whether or not the Jacobi identity for $[\cdot,\cdot]_{\mathcal{C}}$ 
is satisfied is a non-trivial issue. To investigate this, we use the Jacobiator:
\begin{equation}\label{Jacobiator}
 {\rm Jac}_{\mathcal{C}}(X,Y,Z)
 :=[X,[Y,Z]_{\mathcal{C}}]_{\mathcal{C}} + [Z,[X,Y]_{\mathcal{C}}]_{\mathcal{C}} + [Y,[Z,X]_{\mathcal{C}}]_{\mathcal{C}}.
\end{equation}
The following result shows that in order for this to vanish, we also need to require (half-)algebraic speciality of the Weyl curvature:
\begin{lemma}\label{Lemma:Jacobiator}
Suppose that the spinor field $\beta^{A}$ is shear-free. 
Then the Jacobiator \eqref{Jacobiator} vanishes if and only if $\beta^{A}$ is a repeated principal spinor of the ASD Weyl curvature.
\end{lemma}

\begin{proof}
Take three arbitrary vector fields in $\Gamma(\mathbb{L})$ (with possibly different weights), 
$X^{a}=\beta^{A}x^{A'}$, $Y^{a}=\beta^{A}y^{A'}$, $Z^{a}=\beta^{A}z^{A'}$. 
A tedious but straightforward calculation shows that 
\begin{equation*}
 ({\rm Jac}_{\mathcal{C}}(X,Y,Z))^{AA'}
 =\beta^{A}[(x_{C'}y^{C'})(\mathcal{C}_{B'}\mathcal{C}^{B'}z^{A'}) - (x_{C'}z^{C'})(\mathcal{C}_{B'}\mathcal{C}^{B'}y^{A'})
 +(y_{C'}z^{C'})(\mathcal{C}_{B'}\mathcal{C}^{B'}x^{A'})].
\end{equation*}
Therefore, the Jacobi identity for $[\cdot,\cdot]_{\mathcal{C}}$ is satisfied if and only if 
\begin{equation}\label{flatconnection}
 \mathcal{C}_{A'}\mathcal{C}^{A'} = 0
\end{equation}
when acting on any weighted primed spinor field. This equation is investigated in Appendix \ref{Appendix:curvature},
where the proof of Lemma \ref{Lemma:CC0} shows that equation \eqref{flatconnection} is satisfied
if and only if $\Psi_{ABCD}\beta^{B}\beta^{C}\beta^{D}=0$, i.e., $\beta^{A}$ is a repeated 
principal spinor of the ASD Weyl curvature.
\end{proof}

Let us now focus on the construction of the anchor.
From Def. \ref{Def:LieAlgebroid}, this should be a map $\rho:\mathbb{L}\to TM\otimes\mathbb{C}$.
An element $X\in L_{(r,r';w)}\subset\mathbb{L}$ has non-trivial weights, so we cannot use the 
projector $P$ as an anchor as before, since $PX$ has also non-trivial weights and thus it is not an element of $TM\otimes\mathbb{C}$.
Therefore, to any weighted vector $X$ we must somehow associate a vector $\rho(X)$ with trivial weights.
To this end, consider three scalar fields $A,B,C$ of types $(1,0;0)$, $(0,1;0)$ and $(0,0;1)$ respectively; that is:
\begin{equation*}
 A\xrightarrow{G_{o}} \lambda A, \qquad  B\xrightarrow{G_{o}} \mu B, \qquad  C\xrightarrow{G_{o}} \Omega C.
\end{equation*}
Then we can define a map 
\begin{equation}\label{weightedanchor}
 \rho : L_{(r,r';w)}\to TM\otimes\mathbb{C}, \qquad  \rho(X) := A^{-r}B^{-r'}C^{-w}X.
\end{equation}
Although this is rather {\em ad hoc} since we have not specified anything about the scalars $A,B,C$, 
we restrict the possible choices of them by requiring 
\begin{equation}\label{conditionscalars}
 \mathcal{C}_{A'}A = \mathcal{C}_{A'}B = \mathcal{C}_{A'}C = 0
\end{equation}
where as before $\mathcal{C}_{A'}=\beta^{A}\mathcal{C}_{AA'}$. 
In Lemma \ref{Lemma:CC0} it is shown that, if $\beta^{A}$ is shear-free and a repeated principal spinor,
the integrability condition $\mathcal{C}^{A'}\mathcal{C}_{A'}\Phi=0$ is satisfied for any weighted scalar field $\Phi$, 
so there are indeed solutions to \eqref{conditionscalars}.
With the choice \eqref{conditionscalars}, if $X\in\Gamma(L_{(r_1,r'_1;w_1)})$, $Y\in\Gamma(L_{(r_2,r'_2;w_2)})$, we have
\begin{align*}
 \rho([X,Y]_{\mathcal{C}})^{a} ={}& A^{-(r_1+r_2)}B^{-(r'_1+r'_2)}C^{-(w_1+w_2)}[X,Y]^{a}_{\mathcal{C}} \\
 ={}& A^{-(r_1+r_2)}B^{-(r'_1+r'_2)}C^{-(w_1+w_2)}(X^{b}\mathcal{C}_{b}Y^{a} - Y^{b}\mathcal{C}_{b}X^{a}) \\
 ={}&  A^{-r_1}B^{-r'_1}C^{-w_1}X^{b}\mathcal{C}_{b}(A^{-r_2}B^{-r'_2}C^{-w_2}Y^{a}) 
 - A^{-r_2}B^{-r'_2}C^{-w_2}Y^{b}\mathcal{C}_{b}(A^{-r_1}B^{-r'_1}C^{-w_1}X^{a}) \\
 ={}& \rho(X)^{b}\mathcal{C}_{b}(\rho(Y)^{a}) -  \rho(Y)^{b}\mathcal{C}_{b}(\rho(X)^{a}) \\
 ={}& [\rho(X),\rho(Y)]^{a}
\end{align*}
where in the third line we used \eqref{conditionscalars}, in the fourth we used the definition \eqref{weightedanchor}, 
and in the fifth the fact that $\rho(X), \rho(Y)$ have trivial weights. 
Therefore, we see that with the choice of the anchor \eqref{weightedanchor}-\eqref{conditionscalars} and 
the bracket \eqref{wLiebracket}, the map \eqref{weightedanchor} is a morphism 
and the third item in the definition \ref{Def:LieAlgebroid} of a Lie algebroid is satisfied.

The last condition that we need to investigate is the fourth item in Def. \ref{Def:LieAlgebroid}, namely 
the Leibniz rule for the bracket and the anchor: 
\begin{equation}\label{Leibniz?}
 [X,fY]_{\mathcal{C}} \stackrel{?}{=} (\rho(X)f) Y +f [X,Y]_{\mathcal{C}}
\end{equation}
for any $X\in\Gamma(L_{(r_1,r'_1;w_1)})$, $Y\in\Gamma(L_{(r_2,r'_2;w_2)})$ and $f\in C^{\infty}(M)$. 
However, the gradation \eqref{gradedbracket} of the bracket implies 
immediately that we cannot expect this condition to be satisfied:
If $f\in C^{\infty}(M)$ is an ordinary function, 
then $fY\in\Gamma(L_{(r_2,r'_2;w_2)})$, so, because of \eqref{gradedbracket},
the left hand side of \eqref{Leibniz?} is an element of $\Gamma(L_{(r_1+r_2,r'_1+r_2;w_1+w_2)})$; 
but in the right hand side, while the second term is in this space as well,
the first term lives in $\Gamma(L_{(r_2,r'_2;w_2)})$, so eq. \eqref{Leibniz?} is generally inconsistent. 
(There is the possibility however of taking $\rho\equiv 0$, but this is not something that we wish to consider here; 
see \eqref{ordinaryLeibniz} below.)
In other words, asking the condition \eqref{Leibniz?} to hold is incompatible with the gradation 
(assuming $\rho\neq0$). 
The structure we appear to need is then some sort of ``graded Lie algebroid'' 
(if such a thing can be defined at all).
Here we simply observe that computation of the left hand side of \eqref{Leibniz?} gives
\begin{equation}\label{ordinaryLeibniz}
 [X,fY]_{\mathcal{C}} = X(f) \; Y + f [X,Y]_{\mathcal{C}},
\end{equation}
which is the Leibniz rule satisfied by the ordinary Lie bracket of vector fields.

In conclusion, the triple $(\mathbb{L}, \rho, [\cdot,\cdot]_{\mathcal{C}})$, where $\mathbb{L}$, $\rho$ and 
$ [\cdot,\cdot]_{\mathcal{C}}$ are defined respectively by \eqref{gradedAlgebroid}, 
\eqref{weightedanchor}-\eqref{conditionscalars} and \eqref{wLiebracket}, satisfies the conditions 
to be a Lie algebroid {\em except} for the Leibniz rule \eqref{Leibniz?}, which must be replaced by the 
ordinary Leibniz rule \eqref{ordinaryLeibniz}.
More precisely, the Leibniz rule \eqref{Leibniz?} {\em cannot} be compatible with the gradation of a vector bundle 
(when $\rho\neq0$), so it seems that one would need to find an appropriate modification of 
it in order to accommodate a situation with graded fibers.

Another possibility is to introduce a notion of ``weighted algebroids'', by simply 
replacing the tangent bundle in the original definition \ref{Def:LieAlgebroid} by the weighted tangent 
bundle $TM^{\mathbb{C}}_{(r,r';w)}=TM\otimes\mathbb{C}\otimes\mathbb{S}(r,r';w)$. 
The bracket \eqref{wLiebracket} is also defined for any section of $TM^{\mathbb{C}}_{(r,r';w)}$, 
so the ordinary Lie bracket of the tangent bundle could be replaced by \eqref{wLiebracket}.
A ``weighted Lie algebroid'' would then be given by definition \ref{Def:LieAlgebroid} and its four items, 
but using $TM^{\mathbb{C}}_{(r,r';w)}$ instead of the tangent bundle.
Then, by simply taking the inclusion map $L_{(r,r';w)}\hookrightarrow TM^{\mathbb{C}}_{(r,r';w)}$ as ``anchor'' 
(in other words, $\rho=P$ again), the triple $(\mathbb{L}, P, [\cdot,\cdot]_{\mathcal{C}})$ 
is a ``weighted Lie algebroid'' in the above sense.

The above definition is completely {\em ad hoc}, and it is not clear if one gains anything at all 
by introducing such a structure, apart from perhaps adapting the Lie algebroid definition to 
a situation with graded fibers in this particular case. 
However, our motivation for analysing this issue was, as mentioned at the end subsection \ref{Sec:algebroids4D}, 
to see if the differential complex \eqref{derham0} can be generalized to objects with arbitrary weights.
The fact that \eqref{derham0} is indeed a differential complex is a consequence of: 
$(i)$ the morphism property of the anchor, and $(ii)$ the vanishing Jacobiator for the Lie algebroid bracket.
If we now take the inclusion $L_{(r,r';w)}\hookrightarrow TM^{\mathbb{C}}_{(r,r';w)}$ as anchor, then 
it is trivially a morphism, and the Jacobiator for the bracket \eqref{wLiebracket} has been investigated
in Lemma \ref{Lemma:Jacobiator}.
Therefore, it still makes sense to try to construct an analogue of the complex \eqref{derham0}. 
To this end, we first introduce weighted differential forms as sections of the space\footnote{Note that 
$\Lambda^{k}_{(r,r';w)}\neq \wedge^{k}(L\otimes\mathbb{S}(r,r';w))^{*}$.} 
$\Lambda^{k}_{(r,r';w)}:=(\wedge^{k}L^{*})\otimes\mathbb{S}(r,r';w)$, and we define the 
``weighted Lie algebroid differential'',
\begin{equation}
 {\rm d}^{\mathbb{L}} : \Gamma(\Lambda^{k}_{(r,r';w)}) \to \Gamma(\Lambda^{k+1}_{(r,r';w)}), \label{weightedd}
\end{equation}
by its action on weighted 0- and 1-forms, in close analogy to \eqref{LieAlgExtDer0}-\eqref{LieAlgExtDer1}:
\begin{align}
 ({\rm d}^{\mathbb{L}}f)(X) :={}& \mathcal{C}_{X}f, \label{weightedd0} \\
 ({\rm d}^{\mathbb{L}}\omega)(X,Y) :={}& \mathcal{C}_{X}\omega(Y) - \mathcal{C}_{Y}\omega(X) - \omega([X,Y]_{\mathcal{C}}).
  \label{weightedd1}
\end{align}
Explicitly, these expressions are formally equal to the right hand sides of \eqref{LieAlgExtDerGR}-\eqref{LieAlgExtDerGR1}, 
but now $f$ and $\omega$ are allowed to have {\em arbitrary} weights. 
We have: 
\begin{lemma}
Suppose that $\beta^{A}$ is shear-free. Then the following sequence 
is a differential complex if and only if $\beta^{A}$ is also a repeated principal spinor of the ASD Weyl 
curvature\footnote{A similar complex was obtained in \cite{Ara20}, but not in the context of Lie algebroids.}:
\begin{equation}\label{wderham}
 0 \to \Gamma(\Lambda^{0}_{(r,r';w)}) \xrightarrow{{\rm d}^{\mathbb{L}}} \Gamma(\Lambda^{1}_{(r,r';w)}) 
 \xrightarrow{{\rm d}^{\mathbb{L}}} \Gamma(\Lambda^{2}_{(r,r';w)}) \to 0.
\end{equation}
\end{lemma}

\begin{proof}
We have to check ${\rm d}^{\mathbb{L}}\circ {\rm d}^{\mathbb{L}} = 0$. 
We only have to see this for (weighted) 0-forms, since for a 1-form $\omega$, ${\rm d}^{\mathbb{L}}{\rm d}^{\mathbb{L}}\omega$
is a 3-form in a 2-dimensional space and thus it vanishes automatically.
Consider an arbitrary weighted scalar $f\in\Gamma(\Lambda^{0}_{(r,r';w)})$, 
and two arbitrary weighted vectors $X^{a}=\beta^{A}x^{A'}$, $Y^{a}=\beta^{A}y^{A'}$, then, using 
\eqref{weightedd0}, \eqref{weightedd1}:
\begin{align*}
 ({\rm d}^{\mathbb{L}}{\rm d}^{\mathbb{L}}f)(X,Y) ={}& 
 \mathcal{C}_{X}\mathcal{C}_{Y}f - \mathcal{C}_{Y}\mathcal{C}_{X}f - \mathcal{C}_{[X,Y]_{\mathcal{C}}}f \\
 ={}& X^{a}Y^{b}(\mathcal{C}_{a}\mathcal{C}_{b}-\mathcal{C}_{b}\mathcal{C}_{a})f \\
 ={}& x^{A'}y^{B'}\beta^{A}\beta^{B}[\mathcal{C}_{a},\mathcal{C}_{b}]f \\
 ={}& x_{B'}y^{B'}\mathcal{C}_{A'}\mathcal{C}^{A'}f
\end{align*}
where in the fourth line we used \eqref{decompositionCurvatureC}. Therefore, ${\rm d}^{\mathbb{L}}{\rm d}^{\mathbb{L}}f=0$ 
if and only if $\mathcal{C}_{A'}\mathcal{C}^{A'}f=0$, which, because of Lemma \ref{Lemma:CC0}, 
is equivalent to requiring $\beta^{A}$ to be a repeated principal spinor.
\end{proof}

Therefore, we see that as long as $\beta^{A}$ is shear-free {\em and} a repeated principal spinor, 
the ``potentialization'' \eqref{potentials} is also valid for fields with {\em arbitrary} weights.

\begin{remark}
If we are only interested in a particular metric $g_{ab}$, i.e. not in a conformal structure, then instead of
$\mathcal{C}_{AA'}$ we can use the so-called ``GHP'' covariant derivative $\Theta_{AA'}$ (which consists of 
only some parts of $\mathcal{C}_{AA'}$) and its projections $\Theta_{A'},\tilde{\Theta}_{A'}$. 
One can show that if $\beta^{A}$ is shear-free and a repeated principal spinor, then $\Theta_{A'}=\beta^{A}\Theta_{AA'}$ 
also satisfies $\Theta_{A'}\Theta^{A'}=0$ (on scalars and primed spinor fields), 
so we again have a ``potentialization'' scheme for $\Theta_{A'}$.
This is particularly useful for analysing integrability of the Einstein equations.
\end{remark}

\begin{remark}
Although the Leibniz rule is not included in the definition \ref{Def:CourantAlgebroid} of a Courant algebroid, 
we saw in eq. \eqref{IdentityCourant1} that it is actually a consequence of the axioms; thus, the vector bundle $\mathbb{L}$ 
cannot give an ordinary Courant algebroid structure either.
\end{remark}

\section{Twistors}\label{Sec:Twistors}

The twistor programme was introduced by Roger Penrose and it describes a non-local correspondence 
between spacetime and another space called twistor space. 
In this section we describe some natural relations between the constructions we have seen so far and twistor theory.
We first give the basic definition of a twistor, 
then we elaborate in sections \ref{Sec:2dimtwistorspace} and \ref{Sec:twistorspaces}
on the relation between the integrable structures studied in previous sections and twistor spaces.

According to Penrose, the basic premise in twistor theory is that light rays are to be 
regarded as more basic entities than spacetime events (see e.g. \cite[Section 6.1]{PR2}). 
Consider real, 4-dimensional Minkowski spacetime $\mathbb{M}$, and a point with coordinates $x^{AA'}_{o}$ 
(with respect to an arbitrary origin). 
A light ray $\gamma$ through this point is described by the null geodesic
\begin{equation}\label{NullRay}
x^{AA'}(\tau)=\tau\lambda^{A}\bar{\lambda}^{A'}+x^{AA'}_{o},
\end{equation}
where $\tau$ is a real parameter, and
we are using that the vector field $p^{AA'}$ tangent to the geodesic is null, so it can be expressed as
$p^{AA'}\equiv \lambda^{A}\bar{\lambda}^{A'}$ for some spinor field $\lambda^{A}$.
Contraction with $\lambda_{A}$ shows that $x^{AA'}(\tau)\lambda_{A}=x^{AA'}_{o}\lambda_{A}$
for all $\tau$. Since this is true for all $\tau$, the contraction defines a spinor $\mu^{A'}$ by
\begin{equation}\label{incidenceRelation}
 x^{AA'}\lambda_{A} \equiv i\mu^{A'}
\end{equation}
where $x^{AA'}$ denotes any point in $\gamma$, and
the factor $i$ is here introduced only to follow standard twistor conventions.
The entire light ray $\gamma$ is then in principle represented by the pair of spinors $(\lambda_{A},\mu^{A'})$.
However, the same $\gamma$ is obtained if we multiply $\lambda_{A}$ and $\mu^{A'}$ by a non-zero 
complex number. Thus, $\gamma$ can be represented as a point in $(\mathbb{C}^{2*}\oplus\bar{\mathbb{C}}^{2})/\sim$, 
where $\sim$ refers to the equivalence relation just mentioned. The space 
$\mathbb{PT}:=(\mathbb{C}^{2*}\oplus\bar{\mathbb{C}}^{2})/\sim$ is called (projective) twistor space 
(we see that it is isomorphic to the complex projective 3-space $\mathbb{CP}^{3}$ \footnote{Twistor space 
is actually $\mathbb{CP}^{3}\backslash\mathbb{CP}^{1}$ (see e.g. \cite{Huggett, WardWells}), 
but we do not need to make this distinction for our elementary presentation.}).
Actually, real light rays lie in a submanifold of $\mathbb{PT}$, since for real $x^{AA'}$ one has the constraint 
$\mu^{A'}\bar\lambda_{A'}+\bar\mu^{A}\lambda_{A}=0$, as follows from \eqref{incidenceRelation}.
On the other hand, given a point $(\lambda_{A},\mu^{A'})$ in twistor space $\mathbb{PT}$, what does this correspond to in spacetime?
The answer is the set of $x^{AA'}$ such that \eqref{incidenceRelation} holds. 
If $\mu^{A'}\bar\lambda_{A'}+\bar\mu^{A}\lambda_{A}\neq0$, \eqref{incidenceRelation} has no 
real solutions for $x^{AA'}$, so one needs to complexify spacetime in order to have solutions. 
In the complexified spacetime $\mathbb{CM}$, suppose that $y^{AA'}$ is a solution to \eqref{incidenceRelation}, 
i.e. $y^{AA'}\lambda_{A}=i\mu^{A'}$, then we see that $x^{AA'}=\lambda^{A}\pi^{A'}+y^{AA'}$ is also a solution 
for any spinor $\pi^{A'}$. 
Thus the set of points in $\mathbb{CM}$ satisfying \eqref{incidenceRelation} is a 2-dimensional (complex) plane $\Sigma$. 
For any two points $x,y$ in $\Sigma\subset\mathbb{CM}$ we have
$x^{AA'}-y^{AA'}=\lambda^{A}\pi^{A'}$ for some $\pi^{A'}$, thus vectors tangents to $\Sigma$ 
are of the form $\lambda^{A}\pi^{A'}$ where $\lambda^{A}$ is fixed and $\pi^{A'}$ varies. 
This implies that any tangent to $\Sigma$ is null, and likewise the inner product 
between any two tangent vectors vanishes. In other words,
\begin{equation}\label{twistorMinkowski}
 \eta|_{T\Sigma} \equiv 0
\end{equation}
where $T\Sigma$ is the tangent bundle to $\Sigma$ and $\eta$ is the Minkowski metric. 
Twistor space $\mathbb{PT}$ is then the moduli space of the 2-surfaces $\Sigma$ in which \eqref{twistorMinkowski} holds.

The above construction of twistor space starts from the explicit expression \eqref{NullRay} for light rays (null geodesics) 
in flat spacetime (which leads to the fundamental relation \eqref{incidenceRelation}),
thus it does not apply to curved spacetimes. For the curved case one has to generalize \eqref{twistorMinkowski}.
We allow also different signatures for the metric:

\begin{definition}\label{Def:Twistor}
Let $M$ be a (possibly complex) 4-dimensional manifold equipped with a metric $g$.
A {\em twistor surface} (or simply {\em twistor}) in $M$ is a totally null 2-surface:
a 2-dimensional surface $\Sigma$ in which the induced metric vanishes identically,
\begin{equation}\label{twistorCurved}
 g|_{T\Sigma} \equiv 0
\end{equation}
where $T\Sigma$ is the tangent bundle to $\Sigma$. We define {\em twistor space} as the moduli space 
of the 2-surfaces \eqref{twistorCurved}. 
\end{definition}

In flat spacetime one can also define twistors via solutions to the twistor equation (see \cite{PR2}), 
but in general curved spacetimes this equation has no solutions so the definition via \eqref{twistorCurved} 
is slightly more general, see \cite{Penrose76, Huggett, WardWells}.
The existence of twistor surfaces is however still quite restrictive; 
in particular, the usual twistor space of twistor theory is 3-dimensional and it requires
the Weyl curvature to be (anti-)self-dual.

Having introduced these basic concepts, we can now describe the connection of these ideas with the 
structures studied in previous sections.
We notice first that, from the basic definition \eqref{twistorCurved} (or \eqref{twistorMinkowski}) of a twistor, 
we see that it coincides with the Lagrangian submanifolds considered in double field theory, eq. \eqref{Isotropic} 
(but the motivation here is completely different, as we just saw).

\begin{remark}
Since here we are interested in twistors, which are generally {\em complex} surfaces, 
in this section we assume the manifold to be real-analytic, so that it admits 
a complexification, which will be denoted by $\mathbb{C}M$.
\end{remark}

\subsection{Two-dimensional twistor spaces}\label{Sec:2dimtwistorspace}

An immediate connection between twistors and the para-Hermitian structures considered in this work is as follows:
A half-integrable (para-)Hermitian structure in four dimensions defines a two-dimensional twistor space.

To see this, consider an almost-para Hermitian structure $(g,K)$ on a 4-manifold $M$, which produces a splitting 
$TM\otimes\mathbb{C}=L\oplus\tilde{L}$, where $L$ and $\tilde{L}$ are both 2-dimensional.
If the structure is half-integrable, then the distribution $L$ (say) is involutive. 
Then the integral manifolds associated to $L$ are 2-surfaces $\Sigma$ such that $T\Sigma = L|_{\Sigma}$. 
Now, the para-Hermitian property implies, as we know, that $g|_{T\Sigma}=0$, 
therefore the integral manifolds (i.e. leaves) $\Sigma$ are twistor surfaces. 
Furthermore, the space of leaves is two-dimensional, thus we have, according to definition \ref{Def:Twistor}, 
a two-dimensional twistor space.
Notice that this is true for any signature.

If the almost-para Hermitian structure is moreover {\em integrable}, i.e. the two eigenbundles $L$ and $\tilde{L}$ 
are involutive, then there are two different sets of twistor surfaces, i.e. two different two-dimensional twistor spaces. 
(In the next subsection we will analyse possible ways of understanding these spaces in relation to the usual three-dimensional 
space of twistor theory.)
Note that each involutive eigenbundle defines a pair of (possibly complex) coordinates, $z^{i}=(z^{1},z^{2})$ 
and $\tilde{z}_{i}=(\tilde{z}_{1},\tilde{z}_{2})$,
which give in particular a basis for the cotangent bundle, $T^{*}M\otimes\mathbb{C}={\rm span}({\rm d}z^{i}, {\rm d}\tilde{z}_{i})$.
Furthermore, the fact that the metric satisfies $g(K\cdot,K\cdot)=-g(\cdot,\cdot)$ implies that it can be 
written in this basis as
\begin{equation}
 g = g_{i}{}^{j}{\rm d}z^{i}\otimes{\rm d}\tilde{z}_{j} + g^{i}{}_{j}{\rm d}\tilde{z}_{i}\otimes{\rm d}z^{j}.
\end{equation}
According to the DFT interpretation, the interchange of the foliations $z^{i}={\rm const.}$ and $\tilde{z}_{i}={\rm const.}$ 
is understood as a T-duality transformation.

\smallskip
If only one of the distributions $L,\tilde{L}$ is integrable,
it is worth mentioning that, apart from the above interpretation in terms of the twistor programme, 
(complex) 4-manifolds admitting a foliation by the associated twistor surfaces have been extensively studied by 
Pleba\'nski and collaborators, who called these surfaces {\em null strings},
and the manifolds admitting them {\em HH-} or {\em hyper-heavenly} spaces; the name originates in the fact that they are a 
`generalization' of the H-spaces, or {\em heavens}, introduced by 
Newman\footnote{This `generalization' is in the sense that while an H-space is a Ricci-flat, 
complex 4-manifold with self-dual curvature, an HH-space was originally defined as a Ricci-flat, complex 4-manifold 
with (half) algebraically special curvature.
But the definition of Newman's H-space is more subtle than this simple characterization in the sense that it is, by construction, 
associated to general, asymptotically flat spacetimes (with real-analytic null infinity), via solutions to the good-cut equation. 
See also section \ref{Sec:twistorsinliterature}.}.
Remarkably, Pleba\'nski and Robinson showed in \cite{PlebanskiRobinson} that 
the existence of such integral submanifolds is sufficient for reducing the full non-linear 
Einstein vacuum equations to a scalar, second-order, non-linear PDE known as `hyper-heavenly equation'.

\smallskip
A 2-dimensional twistor space like the one above has been considered in \cite{Ara20} 
in relation to problems in perturbation theory in relativity. 
It is interesting to analyse how this 2-dimensional twistor space is related to the usual 
3-dimensional space of twistor theory; we will do this next.

\subsection{Three-dimensional twistor spaces}\label{Sec:twistorspaces}

In this section we will show the following:
If all almost-para Hermitian structures are half-integrable, then there exists a three-dimensional (3D)
twistor space $\mathbb{P}\mathcal{T}$, which is a mono-parametric family of two-dimensional (2D) twistor spaces 
parametrized by projective spinor fields.
This family is not necessarily a fibration over projective spinors: it is a fibration iff the vacuum Einstein equations hold.

Our discussion will be valid for any signature (also for complex manifolds). 
We will generically denote the space of projective spinors (at a point) as $\mathbb{P}^{1}$,
that is $\mathbb{P}^{1}=\mathbb{CP}^{1}$ in Riemannian and Lorentz signature (and in complex manifolds), 
and $\mathbb{P}^{1}=\mathbb{RP}^{1}$ in split signature.

We will first see how to get from a 2D twistor space to a 3D one, then we will discuss the fibration structure.
First, recall from Remark \ref{Remark:anyKintegrable} that, if any $K$ is half-integrable, then any $K$ is 
actually integrable. Now let $K$ be an arbitrary integrable para-Hermitian structure. 
From Section \ref{Sec:spaceParaHermitian} we know that it is represented by two projective spinor fields, 
say $\alpha_{A}$ and $\beta_{A}$, in the form \eqref{KGR}, where (since $K$ is integrable) both fields satisfy \eqref{ShearFree}.
Any other para-Hermitian structure $K'$ is represented by two other projective spinor fields $\varphi_{A}$ and $\psi_{A}$ 
in the same way, and, since spin-space is 2-dimensional, at any point these spinors can be expressed as linear combinations 
of $\alpha_{A}$ and $\beta_{A}$, say $\varphi_{A}=a\alpha_{A}+b\beta_{A}$ and $\psi_{A}=c\alpha_{A}+d\beta_{A}$.
But since $\varphi_{A}$ and $\psi_{A}$ should really be thought of as {\em projective} spinor fields, the scaling is irrelevant, 
so we can take $\varphi_{A}=\alpha_{A}+\lambda_{1}\beta_{A}$ and $\psi_{A}=\alpha_{A}+\lambda_{2}\beta_{A}$ 
for some $\lambda_{1}, \lambda_{2}$. 
This means that the principal spinor fields of any para-Hermitian structure can be reached by different values of 
$\lambda$ in 
\begin{equation}
 \gamma_{A}(\lambda):=\alpha_{A}+\lambda\beta_{A}.
\end{equation} 
Note that $\alpha_{A}$ corresponds to $\lambda=0$ whereas $\beta_{A}$ corresponds to $\lambda=\infty$.
At any fixed point of $M$,
one can think of $\lambda$ as a stereographic coordinate on the Riemann sphere $\mathbb{P}^{1}\cong\mathbb{C}\cup\{\infty\}$ 
(for Riemannian and Lorentz signature) or on the circle $\mathbb{P}^{1} \cong \mathbb{R}\cup\{\infty\}$ 
(for split signature).

For any fixed $\lambda$, $\gamma_{A}(\lambda)$ is a principal spinor field of an integrable para-Hermitian structure, thus,
from section \ref{Sec:2dimtwistorspace}, it defines a 2D twistor space. 
Therefore, varying $\lambda$ we add one more dimension and we get a {\em three}-dimensional twistor space: 
equivalently, a mono-parametric family of 2D twistor spaces. 
We denote the 3D twistor space as $\mathbb{P}\mathcal{T}$.
Since varying $\lambda$ means varying $\gamma_{A}(\lambda)$, 
the family is parametrized by projective spinor fields.

Notice that in Riemannian and Lorentz signature, twistor surfaces are necessarily complex, 
and the twistor space $\mathbb{P}\mathcal{T}$ is complex. 
Twistor surfaces live in the complexification $\mathbb{C}M$, and
at each point of $\mathbb{C}M$ there is a 2-sphere $\mathbb{CP}^{1}$ worth of them.
In terms of $M$, what we have is a 2-sphere bundle of twistor {\em planes}, where 
a twistor plane is the tangent space to a twistor surface.
In split signature, we can restrict to {\em real} twistor surfaces, and 
the associated twistor space $\mathbb{P}\mathcal{T}$ is real. 
In this case, at each point of $M$ there is a circle $\mathbb{RP}^{1}$ worth of twistor surfaces; 
i.e. we have a circle bundle of twistor planes over $M$.
For the connection of these twistor spaces with other twistor constructions in the literature, 
see section \ref{Sec:twistorsinliterature} below.

\begin{remark}
From the constructions seen in section \ref{Sec:Integrability}, it follows that if any almost-para Hermitian 
structure is half-integrable, then there is a one-parameter family of Lie and Courant algebroids for the tangent bundle, 
where each algebroid structure is defined by a projective spinor field.
\end{remark}

Let us now discuss the possible fibration structure of $\mathbb{P}\mathcal{T}$.
What we will argue is that the family of 2D twistor spaces is fibered over projective spinors if and only if the Einstein 
vacuum equations are satisfied\footnote{In twistor theory, the fact that the twistor space of a Ricci-flat complex 4-manifold 
is fibered over $\mathbb{CP}^{1}$ is already known; see e.g. Remark (2) in p. 444 in \cite{WardWells}.}.
Essentially, the obstruction for a fibration is that, given an arbitrary integrable para-Hermitian structure $K$, defined 
by two projective spinor fields $[\alpha_{A}]$, $[\beta_{A}]$, we can achieve
\begin{equation}\label{derivativealongsurfaces}
 \alpha^{B}\nabla_{BB'}\alpha_{A}=0, \qquad \beta^{B}\nabla_{BB'}\beta_{A}=0
\end{equation}
but in general we have
\begin{equation}\label{derivativeoffsurfaces}
  \beta^{B}\nabla_{BB'}\alpha_{A}\neq 0, \qquad \alpha^{B}\nabla_{BB'}\beta_{A}\neq 0.
\end{equation}

First, recall that the assumption that any para-Hermitian structure is integrable implies that the Weyl curvature is self-dual 
(i.e. $\Psi_{ABCD}\equiv0$, we saw this in Lemma \ref{Lemma:SDWeyl}\footnote{The restriction to self-dual 
curvature in relation to the existence of a 3D twistor space has of course been known in twistor theory since its origins.}).
Recall also that any representative spinor field $\alpha_{A}$ satisfies the shear-free condition \eqref{ShearFree}. 
We can now ask a stronger condition and
choose the scaling such that $\alpha_{A}$ is covariantly constant over any twistor surface $\Sigma$ 
associated to it:
\begin{equation}\label{spinorconstantoversurface}
 X^{b}\nabla_{b}\alpha_{A} = 0 \quad \forall X\in \Gamma(T\Sigma),
\end{equation}
that is, $\alpha^{B}\nabla_{BB'}\alpha_{A}=0$. We can do this because the integrability condition for it
is $\Psi_{ABCD}\alpha^{C}\alpha^{B}\alpha^{D} = 0$, which is satisfied since $\Psi_{ABCD}\equiv0$.

Consider two neighboring points in $\mathbb{P}\mathcal{T}$. 
These correspond to two twistor surfaces $\Sigma, \Sigma'$ in $\mathbb{C}M$, which we take to be associated 
to the same congruence defined by the spinor field $\alpha_{A}$.
We now use the concept of a `connecting vector' as defined by Penrose and Rindler in \cite[Section 7.1]{PR2}, 
but slightly generalized to a congruence of surfaces instead of curves.
This is used by Ward and Wells in \cite[Section 9.1]{WardWells}:
a vector field $v$ defined over $\Sigma$ is a `connecting' vector to the nearby twistor surface 
$\Sigma'$ if the Lie derivative of $v$ along vectors tangent to $\Sigma$ is again tangent to $\Sigma$, 
i.e. $\pounds_{X}v\in\Gamma(T\Sigma)$ for all $X\in\Gamma(T\Sigma)$.
Given the connecting vector $v$, the change in the spinor field $\alpha_{A}$ between the nearby 
twistor surfaces $\Sigma$ and $\Sigma'$ is represented by the spinor field
\begin{equation}\label{changeinspinor}
 \delta_{A} = v^{b}\nabla_{b}\alpha_{A}.
\end{equation}
If $\delta_{A}\neq 0$, the different twistor surfaces defined by the same spinor {\em field} $\alpha_{A}$ (i.e. the differents 
points in the associated 2D twistor space) cannot be associated to a same {\em point} spinor.
So, there is no projection from the 2D twistor space to a point in $\mathbb{P}^{1}$, and consequently,
the 3D twistor space $\mathbb{P}\mathcal{T}$ is not fibered over $\mathbb{P}^{1}$.

Now, we can decompose the vector $v$ as $v^{b}=a^{B'}\alpha^{B}+b^{B'}\beta^{B}$ for some $a^{B'}, b^{B'}$. 
Thus, using \eqref{spinorconstantoversurface}, equation \eqref{changeinspinor} is equivalently
\begin{equation}
 \delta_{A} = b^{B'}\beta^{B}\nabla_{BB'}\alpha_{A}.
\end{equation}
Therefore, the change in $\alpha_{A}$ between the different twistor surfaces of the congruence 
is governed by $\beta^{B}\nabla_{BB'}\alpha_{A}$. 
If $\beta^{B}\nabla_{BB'}\alpha_{A}\neq 0$, $\alpha_{A}$ changes along the congruence and thus 
there is no projection $\mathbb{P}\mathcal{T} \to \mathbb{P}^{1}$.
If $\beta^{B}\nabla_{BB'}\alpha_{A}= 0$, we have $v^{b}\nabla_{b}\alpha_{A}=0$ for any connecting vector 
and so $\alpha_{A}$ does not change between the different twistor surfaces of the associated 2D twistor 
space. All such surfaces can then be associated to the same point spinor, i.e. 
to the same element of $\mathbb{P}^{1}$.
So in this case there {\em is} a fibration $\mathbb{P}\mathcal{T}\to \mathbb{P}^{1}$, where the fibers are 
2D twistor spaces. We represent this in Figure \ref{Figure:3Dtwistorspace}.
In addition, note that, since $\alpha_{A}\beta^{A}\neq 0 $, if we have both $\alpha^{B}\nabla_{BB'}\alpha_{A}=0$ 
and $\beta^{B}\nabla_{BB'}\alpha_{A}=0$ then we actually have $\nabla_{BB'}\alpha_{A}=0$. 
In particular this implies that $[\nabla_{a},\nabla_{b}]\alpha_{C}=0$, so the Riemann curvature must be half-flat 
(since $\alpha_{A}$ was arbitrary). 
Since we already had $\Psi_{ABCD}=0$, the new information is 
\begin{equation}
 \Phi_{ABC'D'} = 0 = \Lambda
\end{equation}
where $\Phi_{ABC'D'}$ is the spinor analogue of the trace-free Ricci tensor, and $\Lambda$ is proportional 
to the curvature scalar (see \cite[Section 4.6]{PR1}).
Therefore, the manifold must be Ricci-flat.

Notice that in the Riemannian case, the conditions $\Psi_{ABCD}=0=\Phi_{ABC'D'}=\Lambda$ imply that $(M,g)$ is hyperk\"ahler 
(see e.g. Theorem 9.3.3 in \cite{Dun10}).

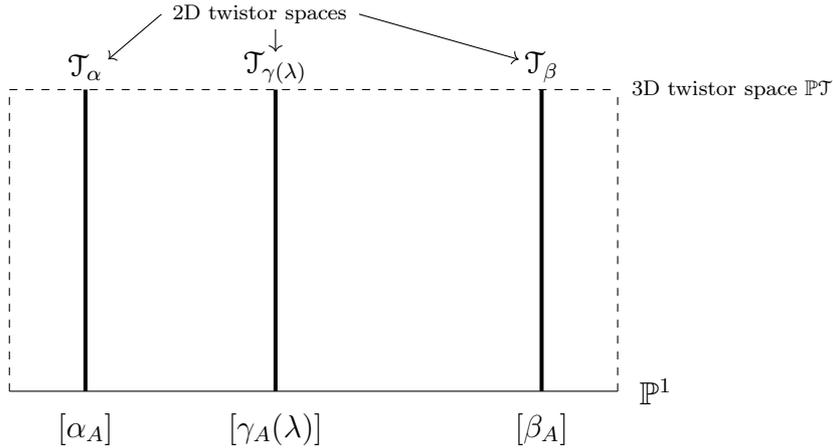
\begin{figure}

\begin{center}
\begin{tikzpicture}

\draw (0,0)--(8,0);
\node (cp1) at (8.5,0) {$\mathbb{P}^{1}$};
\draw[line width=1.5pt] (1,0) -- (1,4);
\node (alpha) at (1,-0.5) {$[\alpha_{A}]$};
\node (Talpha) at (1,4.3) {$\mathcal{T}_{\alpha}$};
\draw[line width=1.5pt] (3.5,0) -- (3.5,4);
\node (gamma) at (3.5,-0.5) {$[\gamma_{A}(\lambda)]$};
\node (Tgamma) at (3.5,4.3) {$\mathcal{T}_{\gamma(\lambda)}$};
\draw[line width=1.5pt] (7,0) -- (7,4);
\node (beta) at (7,-0.5) {$[\beta_{A}]$};
\node (Tbeta) at (7,4.3) {$\mathcal{T}_{\beta}$};

\draw[dashed] (0,0) -- (0,4);
\draw[dashed] (0,4)--(8,4);
\draw[dashed] (8,0) -- (8,4);
\node (T) at (9.5,4) {\scriptsize 3D twistor space $\mathbb{P}\mathcal{T}$};

\draw[<-] 
   (1.3,4.4) to (2,5) 
   node[right] {\scriptsize 2D twistor spaces};
\draw[<-]
   (3.5,4.5) to (3.5,4.8);
\draw[<-]
   (6.7,4.4) to (4.6,5);

\end{tikzpicture}
\end{center}
\caption{3D twistor space $\mathbb{P}\mathcal{T}$ seen as a fibration 
of 2D twistor spaces over projective spinors $\mathbb{P}^{1}$, in the case that the Einstein vacuum equations are satisfied.
In Riemannian and Lorentz signature, $\mathbb{P}^{1}=\mathbb{CP}^{1}\cong S^{2}$, 
while in split signature $\mathbb{P}^{1}=\mathbb{RP}^{1}\cong S^{1}$.
If the manifold is not Ricci-flat, $\mathbb{P}\mathcal{T}$ still consists of a family of 2D twistor spaces, but it 
is not fibered over $\mathbb{P}^{1}$. The 2D twistor spaces $\mathcal{T}_{\alpha}$ and $\mathcal{T}_{\beta}$ 
are ``T-dual'' in the sense that they define two T-dual Lagrangian foliations of the manifold $M$.}
\label{Figure:3Dtwistorspace}
\end{figure}

\subsubsection{Connections with other twistor constructions}\label{Sec:twistorsinliterature}

Let us briefly comment on some relationships with the twistor space constructions 
in the literature, following mainly \cite{Penrose76, WardWells, Dun10, Atiyah, LebrunMason}.
(We are only concerned with {\em projective} twistor space.)
The original construction of twistor space was done by Penrose in \cite{Penrose76}, where it is called non-linear graviton 
and is specific to {\em complex} (anti-)self-dual 4-manifolds. 
The non-linear graviton is a 3-complex-dimensional manifold
and its points are totally null surfaces in the (A)SD 4-manifold. 
The precise relationship between the two spaces can be encoded in a double fibration, where 
the correspondence space is the projective spin bundle. 
For real 4-manifolds, one can, according to signature, introduce different `real structures' (involutions)
on the twistor space of the complex 4-manifold;
but there are also other routes to construct the twistor space of a real 4-manifold.

In the Riemannian case, the reformulation of Penrose's twistor construction was done by Atiyah, Hitchin and Singer in \cite{Atiyah}.
This formulation associates a natural almost-complex manifold to any real Riemannian 4-manifold $(M,g)$.
The basic idea is to incorporate all almost-complex structures in $M$ into {\em one} 
almost-complex structure in a larger manifold; in other words, to provide the space of almost-complex structures 
{\em itself} with an almost-complex structure. 
The space of complex structures at a point is $\mathbb{CP}^{1}$ (see Remark \ref{Remark:Euclidean}), 
so over the whole manifold one gets the projective spin bundle $\mathbb{PS}$.
The almost-complex structure of $\mathbb{PS}$ is described by Ward and Wells 
\cite[Section 9.1]{WardWells} as follows.
Since $\mathbb{PS}$ has a connection, at any point $p=(x,\a)\in\mathbb{PS}$ the tangent space 
decomposes as $T_{p}\mathbb{PS} = T_{x}M\oplus T_{\a}\mathbb{CP}^{1}$. 
On $T_{x}M$ we have the complex structure $J=iK$ defined by $\a=[\a_{A}]$, with $K$ given by \eqref{KEuclidean}, 
and $T_{\a}\mathbb{CP}^{1}$ is the tangent space to a complex manifold so it naturally has a complex structure, 
say $J_{\mathbb{CP}^1}$. Thus, the complex structure of $T_{p}\mathbb{PS}$ is $\mathcal{J}={\rm diag}(J,J_{\mathbb{CP}^1})$.
Atiyah {\em et al.} proved \cite[Theorem 4.1]{Atiyah} that $\mathcal{J}$ is integrable if and only 
if the Weyl curvature of $g$ is (A)SD.
In that case, the space $\mathbb{PS}$ becomes a 3-dimensional complex manifold: this is the
twistor space of $M$. It coincides with the space of twistor surfaces in the complexification of $M$.

\begin{remark}[Twistor families]
Suppose that $\pi_{1}:\mathbb{P}\mathcal{T}\to\mathbb{CP}^{1}$ is the fibration by 2D twistor spaces, 
so that $\pi^{-1}_{1}(\alpha)$ is the 2D twistor space $\mathcal{T}_{\alpha}$.
In the Riemannian setting, there is an alternative way to define a fibration of 
$\mathbb{P}\mathcal{T}$ over $\mathbb{CP}^{1}$. Since $\mathbb{P}\mathcal{T}$ is actually 
the projective spin bundle of the real 4-manifold $M$, one can think of 
a different projection $\pi_{2}:\mathbb{P}\mathcal{T}=\mathbb{PS}\to\mathbb{CP}^{1}$
where a typical fiber $\pi^{-1}_{2}(\alpha)$ is the {\em complex} 2-manifold $(M,g,J_{\alpha})$, 
where $J_{\alpha}$ is the complex structure determined by $\alpha=[\alpha_{A}]$ 
(in our terms this is $J_{\alpha}=iK_{\alpha}$, with $K_{\alpha}$ given by the right hand side of \eqref{KEuclidean}), 
and $g$ is Hermitian w.r.t. $J_{\alpha}$. 
Therefore, $\mathbb{P}\mathcal{T}$ can also be seen as a $\mathbb{CP}^{1}$-family of complex surfaces
where the fibers are isometric as Riemannian manifolds but the complex structure changes from fiber to fiber. 
This is known in the literature as a {\em twistor family}. 
(Particularly relevant is the case in which $M$ is a K3 surface, so that
$\mathbb{P}\mathcal{T}$ is a $\mathbb{CP}^{1}$-family of K3 surfaces.)
The relation between the two fibrations $\pi_{1},\pi_{2} : \mathbb{P}\mathcal{T}\to\mathbb{CP}^{1}$
can be seen by considering the complexification of $M$, denoted $\mathbb{C}M$. 
An integral submanifold in $\mathbb{C}M$ of the involutive distribution $L_{\alpha}\subset TM\otimes\mathbb{C}=T\mathbb{C}M$ 
is, on the one hand, a point in $\mathcal{T}_{\alpha}=\pi^{-1}_{1}(\alpha)$ (i.e. a twistor surface defined by $\alpha=[\alpha_{A}]$), 
and on the other hand, a complex surface $(M,g,J_{\alpha})=\pi^{-1}_{2}(\alpha)$.
In other words, a fiber of $\pi_{1}$ is the space of leaves of the foliation of $\mathbb{C}M$ induced by a projective spinor $\alpha$, 
and a fiber of $\pi_{2}$ is a leaf in this foliation. Both typical fibers are 2-dimensional complex manifolds.
\end{remark}

In split signature, the twistor construction for ASD conformal structures was done by LeBrun and Mason in \cite{LebrunMason}, 
and it also involves, in a way, the consideration of the space of almost-complex structures over $M$.
In section \ref{Sec:twistorspaces} we restricted to real twistor planes (recall that a twistor plane is the tangent space to 
a twistor surface), but we can also
consider complex twistor planes by simply complexifying the space of spinors, i.e. taking 
$\mathbb{R}^{2}\otimes\mathbb{C}=\mathbb{C}^{2}$.
Then the bundle of complex twistor planes over $M$, which we denote by $\mathcal{Z}$ following \cite{LebrunMason},
has fibers $\mathbb{CP}^{1}$ (the twistor planes we consider here are the $\beta$-planes of \cite{LebrunMason}). 
The subbundle of real twistor planes is denoted by $F$ in \cite{LebrunMason}, and its fibers are $\mathbb{RP}^{1}$. 
Now, as we saw in Remark \ref{Remark:complexsplit}, the bundle of almost-complex structures over $M$, say $\mathcal{H}$, 
has fibers $\mathbb{CP}^{1}\backslash\mathbb{RP}^{1}$. This is the same as removing real twistor planes in $\mathcal{Z}$,
so it coincides with the bundle $\mathcal{Z}\backslash F$ in \cite{LebrunMason}.
Using that $\mathbb{CP}^{1}\backslash\mathbb{RP}^{1}$ is a hyperboloid of two sheets,
and identifying each sheet with a disk,
we see that $\mathcal{H}$ consists of two connected components, say $\mathcal{H} = U_{+} \cup U_{-}$,
each of which is a disk bundle over $M$. 
The twistor space of $M$ constructed in \cite{LebrunMason} 
is obtained essentially by giving a complex structure to the manifold-with-boundary $U_{+}\cup F$ 
(see \cite[Theorem 7.3]{LebrunMason}), 
which we can see as ``one half'' (with a boundary attached) of the space of almost-complex structures over $M$.
A twistor space construction involving the two halves of the hyperboloid $H^{2}$ together also with the equator $\mathbb{RP}^{1}$
(in other words, the Riemann sphere $\mathbb{CP}^{1}=H^{2}\cup\mathbb{RP}^{1}$)
can be found in \cite[Section 10.5]{Dun10}.
Here, the twistor space is obtained in a similar way to the Atiyah-Hitchin-Singer construction \cite{Atiyah},
i.e. by finding an involutive subbundle for $T\mathbb{PS}$\footnote{Whether this actually gives a complex 
structure to $\mathbb{PS}$ in the ordinary sense is a little subtle, see \cite{Eastwood}.};
and it consists of two open sets that are separated by a 3-dimensional real boundary. This boundary
is the real twistor space that we considered in section \ref{Sec:twistorspaces}.

\smallskip
Finally, in Lorentz signature, the Weyl curvature spinors of opposite chirality are complex conjugates of each other, 
so (A)SD Weyl tensor implies conformally flat curvature. 
Since the basic twistor construction is conformally invariant, the Lorentzian case is not very different 
from the Minkowski case.
Regarding spaces of complex structures, we saw that any para-Hermitian structure is necessarily complex-valued, 
so by simply multiplying by $i=\sqrt{-1}$ it can be converted into a ``complex'' structure. 
We showed that the space of such structures is $\mathbb{CP}^{1}\times\mathbb{CP}^{1}$ with 
the ``diagonal'' removed, which turns out to be a complex sphere $\mathbb{C}S^{2}$. 
The bundle of ``almost-complex structures'' compatible with $g$ over $M$ has then $\mathbb{C}S^{2}$ fibers,
and since $S^{2}\cong\mathbb{CP}^{1}$, $\mathbb{C}S^{2}$ is a kind of ``second complexification'' 
of $\mathbb{CP}^{1}$ regarded as a real 2-manifold. 
It seems that one is then considering a situation similar to the Riemannian case (in which the corresponding 
bundle has fibers $\mathbb{CP}^{1}$) but where the fibers are now the complexified version.
While in this case we are not aware of a twistor construction analogous to the one in other signatures, 
we find worth mentioning that complexified spheres are actually crucial in the construction of H-space and 
asymptotic twistor space (since these involve the complexification of null infinity). 
This suggests a different connection between Lorentzian ``almost-complex structures''
and twistor constructions: Since such ``almost-complex structures'' are complex-valued, 
we can just consider the bundle of ``almost-complex structures'' compatible with a complex metric 
on a genuine complex 4-manifold. 
As before, the fibers of this bundle are again $\mathbb{C}S^{2}$, so we have a 4-parameter family 
of complexified spheres. 
If all such ``almost-complex structures'' are ``integrable'', then $\Psi_{ABCD}=0$ (Lemma \ref{Lemma:SDWeyl}), 
which suggests that we may identify the complex 4-manifold of $\mathbb{C}S^{2}$'s
with Newman's H-space, and the $\mathbb{C}S^{2}$'s with the `good cuts' of the complexified null infinity 
of a {\em real}, asymptotically flat (Lorentzian) spacetime\footnote{H-space is also Ricci-flat, so 
presumably we would need to add some additional condition in order to make this suggestion more precise.}.

\subsection{Deformations of the para-complex structure}\label{Sec:Deformations}

We now wish to define a notion of deformations of a para-complex structure, and study the integrability 
of such deformations. 
Notice that we explicitly distinguish deformations of the {\em para-complex} structure $K$ from
deformations of the {\em para-Hermitian} structure $(g,K)$; in the former, we only deal with 
a family $K(t)$ of para-complex structures satisfying the para-Hermitian condition with 
respect to some {\em fixed} metric $g$, while in the latter, one should consider families of {\em both} 
para-complex structures $K(t)$ and metrics $g(t)$.
For the definition of deformations, we will follow \cite[Chapter 5]{Gualtieri} and \cite[Chapter 6]{Huybrechts}.
The setting in these works is however slightly different than ours in the sense that they consider deformations 
of {\em complex} structures, where the eigenbundles are complex conjugates, whereas in our 
case they are independent since we consider {\em para}-complex structures (or even complex-valued maps).
Our definition will agree with that of \cite{Svo18}.

\smallskip
In the following we leave the dimension unspecified; we particularize to four dimensions after this discussion.
Let $(g,K)$ be an almost-para Hermitian structure on a $d$-manifold $M$, 
and let $TM\otimes\mathbb{C}\cong L\oplus \tilde{L}$ be the splitting induced by $K$. 
We consider a continuous, smooth family $K(t)=K_{t}$ of para-complex structures, with $K_{0}=K$ 
and $g(K_{t}\cdot,K_{t}\cdot)=-g({\cdot,\cdot})$ for all $t$.
Then for each $t$ we have a splitting $TM\otimes\mathbb{C} \cong L_{t}\oplus\tilde{L}_{t}$ 
into the eigenbundles of $K_{t}$.
For small $t$, an element of $L_{t}$ can be described as $x+B_{t}x$ 
where $x\in L$, $B_{t}x \in \tilde{L}$, and the deformation has been encoded in $B_{t}$, 
which from the above we deduce is a map $B_{t}:L\to \tilde{L}$.
Since $L_{t}$ must be isotropic, we must have $g(x+B_{t}x,y+B_{t}y)=0$ 
for all $x,y\in L$, which translates into $g(x,B_{t}y)+g(B_{t}x,y)=0$, 
that is, $B_{t}$ is skew-symmetric and can be thought of as an element of $\wedge^2 L^{*}$.
Analogously, the deformation of $\tilde{L}$ is described by a skew-symmetric map $\tilde{B}_{t}:\tilde{L}\to L$,
i.e. an element of $\wedge^2 L$.
For the case of {\em Hermitian} structures, one has $\tilde{B}_{t}=\bar{B_{t}}$ since the 
eigenbundles are complex conjugates; but for {\em para-Hermitian} (or complex-valued) 
structures the maps $\tilde{B}_{t}$ and $B_{t}$ are independent.
The small deformation is then described by an endomorphism $A_{t}:L\oplus\tilde{L}\to L\oplus\tilde{L}$
whose matrix representation is
\begin{equation*}
 A_{t}=\left( \begin{matrix} \mathbb{I}_{L} & \tilde{B}_{t} \\ B_{t} & \mathbb{I}_{\tilde{L}} \end{matrix} \right).
\end{equation*}
Thus we see that the deformation is given by (one half of) the sum of the $B$- and $\beta$-transformations 
in \eqref{OrthogonalTransformations}, with parameters $B=2B_{t}$ and $\beta=2\tilde{B}_{t}$,
and according to the interpretation mentioned in Remark \ref{Remark:Shear}, it represents a simultaneous {\em shearing} 
in the directions of $L$ and $\tilde{L}$.
Provided $A_{t}$ is invertible, the deformed para-complex structure is $K_{t}=A_{t}KA^{-1}_{t}$.
Now, we are interested in the case where the original para-Hermitian structure is half-integrable,
which means that (say) the eigenbundle $L$ is involutive while $\tilde{L}$ is not. 
Since we want to study integrable deformations, we will focus on deformations of only $L$, 
that is, we set $\tilde{B}_{t}=0$. 
Then the deformation corresponds to a $B$-transformation \eqref{OrthogonalTransformations}, 
and noticing that (in the splitting $L\oplus\tilde{L}$) $K={\rm diag}(\mathbb{I}_{L}, -\mathbb{I}_{\tilde{L}})$,
the deformed (para-)Hermitian structure $K_{t}=A_{t}KA^{-1}_{t}$ is
\begin{equation}
  K_{t} = \left( \begin{matrix} \mathbb{I}_{L} & 0 \\ 2B_{t} & -\mathbb{I}_{\tilde{L}} \end{matrix} \right). 
  \label{DeformationKmatrix}
\end{equation}
This coincides with the definition given in \cite{Svo18}.
For the calculations below, we will need some more details about this matrix representation.
Let $E_{I}=(e_{i},\tilde{e}^{i})$ be a basis for $L\oplus\tilde{L}$, with $e_{i}$ and $\tilde{e}^{i}$ bases for 
$L$ and $\tilde{L}$ respectively. Any vector $X\in L\oplus\tilde{L}$ can then be decomposed as 
$X=x^{i}e_{i}+\tilde{x}_{i}\tilde{e}^{i}$.
Let $F^{I}=(\theta^{i},\tilde{\theta}_{j})$ be the dual basis, so that $F^{I}(E_{J})=\delta^{I}_{J}$.
The map $K_{t}\in{\rm Aut}(L\oplus\tilde{L})$ can be expanded as 
$K_{t}=(K_{t})^{I}{}_{J}E_{I}\otimes F^{J}$, and the matrix representation \eqref{DeformationKmatrix} means that
\begin{equation}
 K_{t} = e_{i}\otimes\theta^{i} - \tilde{e}^{i}\otimes\tilde{\theta}_{i} +2(B_{t})_{ij}\tilde{e}^{i}\otimes\theta^{j}.
 \label{DeformationK}
\end{equation}

\begin{remark}
In four dimensions, we have seen that half-integrability of $K$ is equivalent to the fact that one 
of the spinor fields representing $K$ is shear-free.
Therefore, the result that deformations of $K$ are encoded in $B$-transformations is somewhat natural,
since, as mentioned in Remark \ref{Remark:Shear}, $B$-transformations in generalized geometry are interpreted as a shearing.
\end{remark}

Now let us restrict to $d=4$ dimensions. 
Suppose that $\alpha^{A}$ and $\beta^{A}$ are the spinor fields representing $K$, eq. \eqref{KGR}.
Let $\mu^{iA'}=(\mu^{0A'},\mu^{1A'})$ and $\nu^{A'}_{i}=(\nu^{A'}_{0},\nu^{A'}_{1})$ 
be two bases for the primed spin bundle $\mathbb{S}'$. 
We have $\epsilon_{A'B'}\mu^{0A'}\mu^{1B'}=\tilde{N}\neq0$ and $\epsilon_{A'B'}\nu^{A'}_{0}\nu^{B'}_{1}=N\neq0$.
From these relations one can deduce\footnote{Here, $\epsilon_{ij}$ is defined by 
$\epsilon_{ij}=-\epsilon_{ji}$ ($i,j=0,1$) and $\epsilon_{01}=1$.} 
\begin{align*}
 \epsilon_{ij}\mu^{iA'}\mu^{jB'} = \tilde{N}\epsilon^{A'B'}, \qquad \epsilon_{A'B'}\mu^{iA'}\mu^{jB'} = \tilde{N}\epsilon^{ij}, \\
 \mu^{iA'}\mu_{jA'}=\tilde{N}\delta^{i}{}_{i}, \qquad \mu^{iA'}\mu_{iB'}=\tilde{N}\delta^{A'}{}_{B'},
\end{align*}
where $\mu_{jB'}\equiv \e_{A'B'}\epsilon_{ij}\mu^{iA'}$, and similarly for $\nu^{A'}_{i}$.
(Indices $i,j$ are raised and lowered with $\epsilon^{ij}$ and $\epsilon_{ij}$, analogously to $A',B'$.)
Using these identities and $F^{I}(E_{J})=\delta^{I}_{J}$, we deduce the following expressions for the frames and their duals:
\begin{align*}
  e_{i} = \beta^{A}\nu^{A'}_{i}\partial_{AA'}, \qquad & \qquad \tilde{e}^{i} = \alpha^{A}\mu^{iA'}\partial_{AA'}, \\
  \theta^{i} = (\phi N)^{-1}\alpha_{A}\nu^{i}_{A'}{\rm d}x^{AA'}, \qquad & 
  \qquad \tilde{\theta}_{i} = -(\phi \tilde{N})^{-1}\beta_{A}\mu_{iA'}{\rm d}x^{AA'}
\end{align*}
where we recall that $\phi = \alpha_{A}\beta^{A}$.
It then follows that $B_{t}$ is given by
\begin{equation*}
 B_{t} = (B_{t})_{ij}\tilde{e}^{i}\otimes\theta^{j} = 
 (\phi N)^{-1} (B_{t})_{ij}\alpha^{A}\alpha_{B}\mu^{iA'}\nu^{j}_{B'}\partial_{AA'}\otimes{\rm d}x^{BB'}
\end{equation*}
Using now that $(B_{t})_{ab}=g_{ac}(B_{t})^{c}{}_{b}$ is skew-symmetric, we get
\begin{equation*}
 (B_{t})_{ab} = - \phi^{-1}\varepsilon \alpha_{A}\alpha_{B}\epsilon_{A'B'}, \label{DeformationB}
\end{equation*}
where the scalar field $\varepsilon \equiv -\frac{1}{2}N^{-1}(B_{t})_{ij}\epsilon^{C'D'}\mu^{i}_{C'}\nu^{j}_{D'}$ 
encodes the deformation.
From \eqref{DeformationB} we get $(B_{t})^{a}{}_{b}=\phi^{-1}\varepsilon \alpha^{A}\alpha_{B}\delta^{A'}{}_{B'}$,
thus, the deformed para-complex structure \eqref{DeformationK} is 
\begin{align}
\nonumber (K_{t})^{a}{}_{b} ={}& K^{a}{}_{b} + 2(B_{t})^{a}{}_{b} \\
\nonumber ={}& \phi^{-1} (\alpha^{A}\beta_{B}+\beta^{A}\alpha_{B})\delta^{A'}{}_{B'} + 2\phi^{-1}\varepsilon \alpha^{A}\alpha_{B}\delta^{A'}{}_{B'} \\
 ={}& (\alpha_{C}\beta^{C}_{\varepsilon})^{-1}(\alpha^{A}\beta_{\varepsilon B}+\beta^{A}_{\varepsilon}\alpha_{B})\delta^{A'}{}_{B'},
 \label{DeformationK4D}
\end{align}
where we introduced 
\begin{equation}
 \beta^{A}_{\varepsilon}:= \beta^{A}+\varepsilon \alpha^{A}. \label{betaepsilon}
\end{equation}
Note that $t$ is the parameter of the continuous family $K_{t}$, while $\varepsilon$ is a (possibly complex) scalar field on the manifold.
The eigenbundles of $K_{t}$ are analogous to \eqref{LGR}-\eqref{LtildeGR}:
\begin{align}
 L_{t} ={}& \{u^{a}\in TM\otimes\mathbb{C} \;|\; u^{a}=\beta^{A}_{\varepsilon}\nu^{A'}, \;\; \nu^{A'}\in\mathbb{S}'(0,-1;-w_1) \}, \label{LDeformedK}\\
 \tilde{L}_{t} ={}& \{v^{a}\in TM\otimes\mathbb{C} \;|\; v^{a}=\alpha^{A}\mu^{A'}, \;\; \mu^{A'}\in\mathbb{S}'(-1,0;-w_0) \}. \label{LtildeDeformedK}
\end{align}
We see that the bundle $\tilde{L}$ is not modified since $\tilde{L}_{t}=\tilde{L}$, whereas $L$ is deformed to $L_{t}$
by using $\beta^{A}_{\varepsilon}$ instead of $\beta^{A}$.

Half-integrability of the deformation $K_{t}$ refers to the involutivity properties of $L_{t}$.
Such conditions were given in Lemma \ref{Lemma:paraHermitian}, where we saw that one has to study equation \eqref{ShearFree}. 
Therefore, the spinor field \eqref{betaepsilon} codifies the integrability properties of the deformation \eqref{DeformationK}. 
We first prove the following identity:

\begin{proposition}
Let $\alpha^{A},\beta^{A}$ be arbitrary spinor fields with $\alpha_{A}\beta^{A}=\phi\neq0$, let $\mathcal{C}_{AA'}$
be the associated covariant derivative (Def. \ref{Def:GHPconformalCV}) and let 
$\mathcal{C}_{A'}, \tilde{\mathcal{C}}_{A'}$ be the projections \eqref{CProjections}. 
Let $\beta^{A}_{\varepsilon}$ be given by \eqref{betaepsilon}, where $\varepsilon$ is an arbitrary (possibly complex) scalar field.
Then we have the identity
\begin{equation}
 \beta^{A}_{\varepsilon}\beta^{B}_{\varepsilon}\nabla_{AA'}\beta_{\varepsilon B} = 
 \beta^{A}\beta^{B}\nabla_{AA'}\beta_{B} 
 + \phi\mathcal{C}_{A'} \varepsilon + \tfrac{1}{2}\phi\tilde{\mathcal{C}}_{A'}\varepsilon^{2}
 + (\alpha^{A}\alpha^{B}\nabla_{AA'}\alpha_{B})\varepsilon^{3}. \label{DeformedShear}
\end{equation}
\end{proposition}

\begin{proof}
We first check that $\nabla_{AA'}$ can be replaced by $\mathcal{C}_{AA'}$ in the left hand side of \eqref{DeformedShear}.
Note that the weights of $\beta_{\varepsilon B}$ are the same as those of $\beta_{B}$, so using Table \ref{Table:weights} 
and formula \eqref{cvCexplicit} we have
\begin{equation*}
 \mathcal{C}_{AA'}\beta_{\varepsilon B} = \nabla_{AA'}\beta_{\varepsilon B}-f_{A'B}\beta_{\varepsilon A}
 +((w_{1}+1)f_{AA'}+M_{AA'})\beta_{\varepsilon B},
\end{equation*}
thus, contracting with $\beta^{A}_{\varepsilon}\beta^{B}_{\varepsilon}$, we get
\begin{equation*}
 \beta^{A}_{\varepsilon}\beta^{B}_{\varepsilon}\nabla_{AA'}\beta_{\varepsilon B} = 
 \beta^{A}_{\varepsilon}\beta^{B}_{\varepsilon}\mathcal{C}_{AA'}\beta_{\varepsilon B}.
\end{equation*}
Now,
\begin{align*}
\mathcal{C}_{AA'}\beta_{\varepsilon B} ={}& \mathcal{C}_{AA'}\beta_{B} + \alpha_{B}\mathcal{C}_{AA'}\varepsilon 
 +\varepsilon\mathcal{C}_{AA'}\alpha_{B}\\
 ={}& \phi^{-2}\alpha_{A}\alpha_{B}(\beta^{C}\beta^{D}\nabla_{CA'}\beta_{D}) + \alpha_{B}\mathcal{C}_{AA'}\varepsilon 
 + \varepsilon\phi^{-2}\beta_{A}\beta_{B}(\alpha^{C}\alpha^{D}\nabla_{CA'}\alpha_{D}) 
\end{align*}
where in the second line we used identities \eqref{Calpha} and \eqref{Cbeta}. 
Contracting now with $\beta^{A}_{\varepsilon}\beta^{B}_{\varepsilon}$ and using $\beta^{A}_{\varepsilon}\beta_{A}=-\phi\varepsilon$ 
and $\beta^{A}_{\varepsilon}\alpha_{A}=\phi$, we get
\begin{equation*}
 \beta^{A}_{\varepsilon}\beta^{B}_{\varepsilon}\mathcal{C}_{AA'}\beta_{\varepsilon B} = 
 \beta^{A}\beta^{B}\nabla_{AA'}\beta_{B} 
 + \phi\mathcal{C}_{A'} \varepsilon + \phi\varepsilon\tilde{\mathcal{C}}_{A'}\varepsilon
 + (\alpha^{A}\alpha^{B}\nabla_{AA'}\alpha_{B})\varepsilon^{3}
\end{equation*}
therefore \eqref{DeformedShear} follows.
\end{proof}

\begin{theorem}\label{Theorem:IntegrableDeformations}
Let $(g,K)$ be a half-integrable para-Hermitian structure (Def. \ref{Def:LparaHermitian}) on a 
4-manifold $M$, and let $\a_{A},\b_{A}$ be the spinor fields representing $K$, where $\b_{A}$ satisfies \eqref{ShearFree}.
Let $K_{t}$ be a small deformation of the para-complex structure $K$, represented  by the scalar field $\varepsilon$
as in \eqref{DeformationK4D}, \eqref{betaepsilon}.
Then the eigenbundle $L_{t}$ is involutive  (i.e. $K_{t}$ is a half-integrable deformation) if and only if 
\begin{equation}
 \mathcal{C}_{A'} \varepsilon = 0. \label{IntegrableDeformation}
\end{equation}
\end{theorem}

\begin{proof}
From Lemma \ref{Lemma:paraHermitian} and its proof we know that involutivity of $L_{t}$ 
is equivalent to $N_{P_{t}}\equiv 0$, which in turn is true if and only if $\beta^{A}_{\varepsilon}$
satisfies equation \eqref{ShearFree}. The left hand side of \eqref{ShearFree} has been computed for $\beta^{A}_{\varepsilon}$
in \eqref{DeformedShear}. Since $\beta^{A}$ satisfies \eqref{ShearFree} by assumption, and ``small'' deformation 
means that we only keep terms linear in $\varepsilon$, the result \eqref{IntegrableDeformation} follows.
\end{proof}

\begin{remark}
Equation \eqref{IntegrableDeformation} has non-trivial integrability conditions: applying $\mathcal{C}^{A'}$,
we get $\mathcal{C}^{A'}\mathcal{C}_{A'} \varepsilon = 0$. Using the first equation in \eqref{CCscalars}, we have
\begin{equation*}
 \mathcal{C}_{A'}\mathcal{C}^{A'} \varepsilon = (\beta^{A}\beta^{B}\nabla_{A'A}\psi^{A'}_{B})\varepsilon
\end{equation*}
where $\psi_{a}$ is given by \eqref{internalconnection} with $r,r',w$ the weights of $\varepsilon$.
These weights can be deduced from \eqref{betaepsilon}: $w(\varepsilon)=w_{1}-w_{0}$, $r(\varepsilon)=-1$, $r'(\varepsilon)=+1$.
Replacing these values in identity \eqref{bbdpsi2}:
\begin{equation*}
 \mathcal{C}_{A'}\mathcal{C}^{A'} \varepsilon  = 4\phi^{-1}\Psi_{ABCD}\alpha^{A}\beta^{B}\beta^{C}\beta^{D}\; \varepsilon.
\end{equation*}
Therefore, if $\varepsilon\neq0$ satisfies \eqref{IntegrableDeformation}, then we must have
$\Psi_{ABCD}\alpha^{A}\beta^{B}\beta^{C}\beta^{D}=0$, or equivalently $\Psi_{ABCD}\beta^{B}\beta^{C}\beta^{D}=0$ 
(since $\beta^{A}$ satisfies \eqref{IntegrabilityConditionSFR}). 
This means that {\em in order for integrable deformations to exist, the Weyl tensor must be (half) algebraically special}.
\end{remark}

We end this section with a few remarks.

First, notice that in view of \eqref{weightedd0} (see also \eqref{LieAlgExtDerGR}), 
we can write \eqref{IntegrableDeformation} equivalently as
\begin{equation}\label{MC}
 {\rm d}^{\mathbb{L}}\varepsilon=0
\end{equation}
where ${\rm d}^{\mathbb{L}}$ is the weighted Lie algebroid exterior derivative \eqref{weightedd}.
Taking into account the fact that $\varepsilon$ is a scalar field, we can interpret this as a Maurer-Cartan equation
(since the commutator term in the usual Maurer-Cartan eq. vanishes because $\varepsilon$ is a scalar).
The reason for interpreting \eqref{MC} in such a way is that it
fits with the usual theory of integrable deformations of complex structures (see e.g. \cite[Chapter 6]{Huybrechts}).

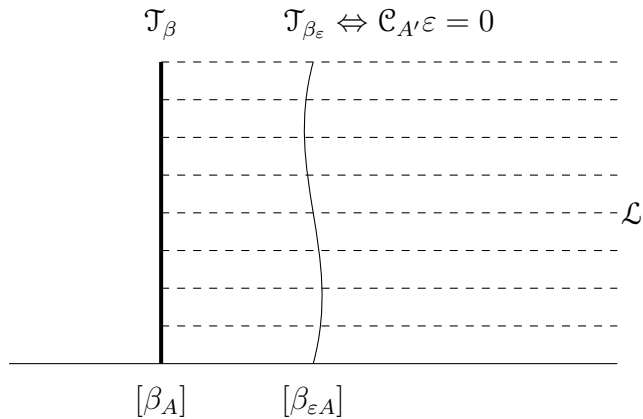
\begin{figure}

\begin{center}
\begin{tikzpicture}

\draw (0,0)--(8,0);

\draw[line width=1.5pt] (2,0) -- (2,4);
\node (beta) at (2,-0.5) {$[\beta_{A}]$};
\node (Tbeta) at (2,4.5) {$\mathcal{T}_{\beta}$};

\draw[dashed]  (2,0.5) -- (8,0.5);
\draw[dashed]  (2,1) -- (8,1);
\draw[dashed]  (2,1.5) -- (8,1.5);
\draw[dashed]  (2,2) -- (8,2);
\draw[dashed]  (2,2.5) -- (8,2.5);
\draw[dashed]  (2,3) -- (8,3);
\draw[dashed]  (2,3.5) -- (8,3.5);
\draw[dashed]  (2,4) -- (8,4);
\node (L) at (8.2,2) {$\mathcal{L}$};

\draw (4,0) to[out=75,in=255] (4,4);
\node (betae) at (4,-0.5) {$[\beta_{\varepsilon A}]$};
\node (Tbetae) at (5,4.5) {$\mathcal{T}_{\beta_{\varepsilon}} \Leftrightarrow \mathcal{C}_{A'}\varepsilon = 0$};

\end{tikzpicture}
\end{center}
\caption{Illustration of small deformations of a para-complex structure in $M$ in terms of 2D twistor spaces. 
A principal spinor field $\beta_{A}$ of a half-integrable 
para-complex structure $K$ defines a 2D twistor space $\mathcal{T}_{\beta}$ (thick vertical straight line).
The dashed horizontal lines represent a line bundle $\mathcal{L}$ over $\mathcal{T}_{\beta}$.  
A section of $\mathcal{L}$ (vertical curved line) corresponds to a scalar field $\varepsilon$ on $M$ such that $\mathcal{C}_{A'}\varepsilon=0$. 
It represents a small half-integrable deformation of $K$, so it defines a 2D twistor space $\mathcal{T}_{\beta_{\varepsilon}}$. 
Integrability conditions imply that the Weyl tensor must be half-algebraically special.}
\label{Figure:deformations}
\end{figure}

Second, recall from section \ref{Sec:2dimtwistorspace} that a half-integrable para-Hermitian structure defines 
a 2D twistor space, say $\mathcal{T}_{\beta}$.
In \cite{Ara20}, functions $\varepsilon$ satisfying \eqref{IntegrableDeformation} were interpreted as 
defining sections of line bundles over $\mathcal{T}_{\beta}$. 
Here we see that \eqref{IntegrableDeformation} appear as integrability conditions for deformations of a para-complex structure in $M$.
We can attempt a more or less ``intuitive'' understanding of this by considering a picture similar to 
figure \ref{Figure:3Dtwistorspace}, in which 2D twistor spaces are represented as vertical lines. 
We illustrate this in figure \ref{Figure:deformations}.
If $\mathcal{L}$ is a line bundle over $\mathcal{T}_{\beta}$, a section is a smooth map 
$\mathcal{T}_{\beta}\to\mathcal{L}$, defined by the condition that it should be covariantly constant over 
the twistor surfaces associated to $\mathcal{T}_{\beta}$, i.e. we can think of it as a scalar field $\varepsilon$ on $M$
such that $\mathcal{C}_{A'}\varepsilon=0$.
So a section of $\mathcal{L}$ defines, by theorem \ref{Theorem:IntegrableDeformations}, a half-integrable small deformation, 
and consequently a 2D twistor space $\mathcal{T}_{\beta_{\varepsilon}}$, ``close to'' $\mathcal{T}_{\beta}$, 
see fig. \ref{Figure:deformations}.

Finally, in the literature on DFT (see \cite{Svo18}), there are weaker forms of integrability where one 
replaces the ordinary Lie bracket by the so-called {\em D-bracket} (a generalization of the Dorfman bracket).
In particular, in \cite[Proposition 5.6]{Svo18} it is shown that a small deformation is `weakly integrable with respect to $K$'
if and only if a certain Maurer-Cartan equation is satisfied, where the operators involved are associated to the Lie algebroid structure. 
The notion of weak integrability refers to the involutivity of a subbundle in terms of the {\em canonical D-bracket} 
$\llbracket \cdot,\cdot \rrbracket^{\rm can}$, see \cite{Svo18} for details.
Adapting these concepts to our formulation, the subbundle $L_{t}$ is weakly integrable if, for all $x,y\in\Gamma(L_{t})$, 
it holds $(\llbracket x,y \rrbracket^{\rm can})^{AA'}=\beta^{A}_{\varepsilon}\pi^{A'}$ for some $\pi^{A'}$. 
Equivalently, this is $\beta_{\varepsilon A }(\llbracket x,y \rrbracket^{\rm can})^{AA'}=0$. 
By a direct calculation (using the definition of $\llbracket \cdot,\cdot \rrbracket^{\rm can}$ given in \cite{Svo18})
one can show that this is always true in our case,
so any small deformation of $K$ is weakly integrable with respect to $K$.

\section{Summary and conclusions}\label{Sec:Conclusions}

In this work we have attempted to show that there is a close connection between the para-Hermitian approach to
double field theory (a formulation of string theory designed to be manifestly covariant under T-duality) 
and the geometry of special four-dimensional manifolds that are of particular interest in general relativity.
This approach to DFT is based on the consideration of para-complex structures 
where one of the eigenbundles is integrable, 
so the tangent bundle of the extended manifold has the structure of a Courant algebroid, 
which is one of the main objects studied in generalized geometry.
We have analysed the four-dimensional version of this approach 
while also allowing the para-complex structure to be {\em complex-valued}, which in turn allows us to 
work with different metric signatures.

One of our original motivations was to point out the similarity between some of the geometrical 
structures in the para-Hermitian approach to DFT, and Pleba\'nski's hyper-heavenly construction in
general relativity. 
This connection is revealed when one observes that Pleba\'nski's congruence of ``null strings'' 
can be understood as the foliation by integral manifolds of an involutive eigenbundle of a 
complex-valued ``almost-complex structure''.
Such a structure has independent eigenbundles and so it is actually `half-integrable'. 
In real geometry, this phenomenon is captured by {\em para}-complex structures.

Since a spacetime manifold $M$ is, in general, not naturally equipped with a (para-)complex structure 
(see also below), we first deduced the general form of any almost (para-)Hermitian structure in four dimensions,
and we described the space of such structures over $M$.
We did this by showing their equivalence to (anti-)self-dual 2-forms and then describing the corresponding 
spaces in terms of projective spaces, for the three different possible metric signatures. 
We showed that these spaces are certain fiber bundles over the 4-manifold, where the fibers 
can be real 2-spheres (Riemannian signature), 1-sheet hyperboloids (split signature), 
or complex 2-spheres (Lorentz signature).
Noticing that an almost para-Hermitian structure is associated only to a {\em conformal} metric 
and to {\em projective} spinor fields, we showed how to deal with the associated ``gauge freedom'' 
by employing a formalism manifestly covariant under gauge transformations.
The corresponding covariant objects are {\em weighted fields}.

Having the general form of any almost para-Hermitian structure, we proceeded to analyse 
integrability issues, showing that half-integrability corresponds to the existence of 
special spinor fields (``shear-free congruences''), and then constructing 
the associated Lie and Courant algebroids. 
We were naturally led to the problem of generalizing the algebroid structures to fields with arbitrary weights, 
and we showed that ordinary Lie algebroids seem to be incompatible with them, 
since the Leibniz rule is not compatible with the gradation of the vector bundle of weighted fields 
(assuming that the anchor does not vanish identically).
Nevertheless, we were able to generalize the Lie algebroid differential complex, 
under the additional assumption of algebraically special curvature.
Further applications of these generalized geometric structures are left for future work.

Finally, we discussed connections of this approach with twistors.
The first observation is that the Lagrangian submanifolds of DFT correspond to the basic object in 
twistor theory, namely twistor surfaces in a 4-manifold.
Such twistor surfaces arise as integral submanifolds of a half-integrable (para-)Hermitian structure $K$.
This $K$ defines a two-dimensional (2D) twistor space, by means of the corresponding projective spinor 
field satisfying the shear-free condition. 
If all (para-)Hermitian structures are integrable, we showed that one then gets
a three-dimensional (3D) twistor space, which is a one-parameter family of 2D twistor spaces, 
parametrized by projective spinor fields.
In general this is not a fibration over projective spinors: the fibration structure is obtained iff the 
vacuum Einstein equations hold.
We discussed connections of this construction with other developments 
such as the Atiyah-Hitchin-Singer and the LeBrun-Mason approaches to twistor theory, 
and also with the $\mathbb{CP}^{1}$ twistor families known in the literature.
Lastly, we analysed deformations of a para-complex structure $K$ in $M$, and we showed that a 
small deformation is half-integrable iff it corresponds to a section of a line bundle over the 
2D twistor space associated to $K$.
We also found that small half-integrable deformations can only exist 
if the Weyl tensor is half-algebraically special.

\medskip
As mentioned, a choice of (para-)complex structure in a spacetime manifold is in general an extra assumption, 
and this is one of the reasons why we first focused on classifying all such possible choices. 
However, for 4-manifolds there are, at least, two situations of interest where natural almost (para-)complex manifolds arise:
\begin{enumerate}
\item Einstein manifolds with algebraically special Weyl tensor,
\item Twistor theory.
\end{enumerate}
The first case is perhaps the one of immediate interest in general relativity. The Goldberg-Sachs theorem implies that 
a vacuum solution (cosmological constant allowed) to the Einstein equations which is algebraically special,
has at least one shear-free congruence. 
If the corresponding (projective) spinor field is $\beta_{A}$, one can choose any other spinor field $\alpha_{A}$ 
with $\beta^{A}\alpha_{A}\neq0$ and define a half-integrable para-Hermitian structure by \eqref{KGR}.
In the second case, for 4-manifolds $M$ with a metric of Riemannian or split signature,
the twistor construction associates a natural almost-complex manifold: it gives an almost-complex structure 
to the space of all almost-complex structures in $M$.

\medskip
To conclude, we make a few comments about some possible directions that will be 
explored in future works.

First, as mentioned before, the Lie algebroid structure of a manifold with a shear-free spinor field has associated 
natural differential complexes that, in our case, allow to find certain potentials for 
problems of interest in general relativity. This particular point will be the subject of a separate work.
In addition, we notice that in the current work we have merely given an explicit description of the Courant algebroid 
associated to such special 4-manifolds, but we have not applied the resulting construction in any interesting way. 
It may be worth analysing if this explicit connection between generalized geometry and spacetimes 
of relevance in relativity has potentially interesting consequences.
More broadly, the generalized notions of integrability encountered in double field theory (such as the D-bracket or 
the associated metric algebroids) were not applied in this work.

In addition, we have not attempted to analyse applications of the notion of T-duality that is obtained 
in general relativity by particularizing the DFT framework to 4-manifolds. 
However, an immediate observation is that, for Einstein manifolds 
of Petrov type D, there is a natural para-Hermitian structure where both eigenbundles are integrable, 
and one eigenbundle is associated to ingoing principal null directions while the other one is 
associated to outgoing p.n.d.'s.
The transformation between ingoing and outgoing p.n.d.'s can then be understood 
as a T-duality transformation.
The question of whether this has some interesting consequences for e.g. black hole spacetimes 
(or perturbations thereof) is left for future work.

Another point that will be elaborated on in a separate work is the relationship between the 
bundle of complex-valued almost-complex structures in a 4-manifold and the theory 
of H-space and asymptotic twistor space (see section \ref{Sec:twistorsinliterature}).

Finally, a possible generalization to higher dimensional manifolds of some of the ideas developed in this work  
will likely follow the perspective and results in \cite{Mason2010}. 
A spinorial treatment should involve the use of {\em pure} spinors.

\subsubsection*{Acknowledgements}

I gratefully acknowledge the support of the Alexander von Humboldt Foundation
through a Humboldt Research Fellowship for Postdoctoral Researchers.

\appendix

\section{Spinors in four dimensions}\label{Appendix:spinors}

In this appendix we review some basic facts about spinors at a point in four dimensions, in any signature. 
For spinor {\em fields}, see the beginning of Appendix \ref{Appendix:CovariantFormalism}.

Given a $d$-dimensional vector space with a non-degenerate inner product of signature $(p,q)$, 
with $p+q=d$, the spin group can be defined as the double covering of the identity component of the orthogonal group, that is 
${\rm Spin}(p,q)=\widetilde{{\rm SO}_o}(p,q)$.
In four dimensions, one has the following isomorphisms depending on the signature:
\begin{align}
 {\rm Spin}(1,3) \cong{}& {\rm SL}(2,\mathbb{C}), \label{spin13} \\
 {\rm Spin}(4) \cong{}& {\rm SU}(2)\times{\rm SU}(2), \label{spin4} \\
 {\rm Spin}(2,2) \cong{}& {\rm SL}(2,\mathbb{R})\times{\rm SL}(2,\mathbb{R}). \label{spin22}
\end{align}
The natural (inequivalent) representations of ${\rm SL}(2,\mathbb{C})$ are $\mathbb{C}^{2}$ and $\bar{\mathbb{C}}^{2}$.
For ${\rm SU}(2)$, the natural representation is $\mathbb{C}^2$ (which in this case is equivalent as a representation 
to $\bar{\mathbb{C}}^{2}$), and for ${\rm SL}(2,\mathbb{R})$ the natural representation is $\mathbb{R}^{2}$.
Therefore, in all cases the finite dimensional irreducible representations are classified by two integers or half-integers $(n,m)$, 
that are said to be of opposite chirality. The fundamentals are $(\frac{1}{2},0)$ and $(0,\frac{1}{2})$.
We see that, in the Lorentzian and Riemannian cases, spinors are complex, while in neutral signature they are real.
Furthermore, in Lorentz signature, complex conjugation interchanges chirality (i.e. it maps an element of 
$\mathbb{C}^{2}$ to an element of $\bar{\mathbb{C}}^{2}$), but in the Riemannian case it is an involution.
In the following everything is valid for any signature.

We denote elements in the $(\frac{1}{2},0)$ representation by e.g. $\varphi^{A}$, 
and elements in $(0,\frac{1}{2})$ by e.g. $\psi^{A'}$.
The vector representation is $(\frac{1}{2},\frac{1}{2})=(\frac{1}{2},0)\otimes(0,\frac{1}{2})$, thus, a vector $v^{a}$ can 
be represented as an object with a pair of spinor indices $v^{AA'}$. The map from $v^{a}$ to $v^{AA'}$ is 
$v^a \to v^{AA'} = \sigma_{a}{}^{AA'}v^{a}$, where $\sigma_{a}$ represents the identity and the 
three Pauli matrices.
It is customary to omit $\sigma_{a}{}^{AA'}$ in this correspondence, so that one simply associates 
a vector index with a pair of spinor indices of opposite chirality, e.g. $v^{a}\equiv v^{AA'}$, etc.;
in this work we follow this convention.

The spin spaces are equipped with the symplectic structures $\e_{AB}$ and $\e_{A'B'}$. 
The inverses are denoted by $\e^{AB}$ and $\e^{A'B'}$, and they satisfy 
$\e^{AC}\e_{BC}=\d^{A}{}_{B}$ and $\e^{A'C'}\e_{B'C'}=\d^{A'}{}_{B'}$.
These objects allow to raise and lower spinor indices; our convention is 
\begin{equation}\label{RaisingLowering}
 \varphi^{A}:=\e^{AB}\varphi_{B}, \qquad \phi_{A}:=\e_{BA}\phi^{B},
\end{equation}
and analogously for primed indices. 
The relation between $\e_{AB}$, $\e_{A'B'}$ and the metric is simply 
$g_{ab}=g_{AA'BB'}=\epsilon_{AB}\epsilon_{A'B'}$.

Given two spinors $\xi_{A},\eta_{A}$ with $\e^{AB}\xi_{A}\eta_{B}=\chi\neq 0$, 
and using the convention \eqref{RaisingLowering} to define $\xi^{A},\eta^{A}$, we have
\begin{align}
 \e_{AB} ={}& \chi^{-1}(\xi_{A}\eta_{B}-\xi_{B}\eta_{A}), \\
 \delta^{A}{}_{B} ={}& -\chi^{-1}(\xi^{A}\eta_{B}-\eta^{A}\xi_{B}). \label{KroneckerSpinors}
\end{align}
These identities are used repeatedly throughout this work. For example, consider a spinor $\varphi_{AB}=\varphi_{(AB)}$,
then using the identity
\begin{equation}
 \varphi_{AB} = \delta^{C}{}_{A}\delta^{D}{}_{B}\varphi_{CD} 
\end{equation}
and replacing \eqref{KroneckerSpinors} one gets the expression \eqref{ExpansionInBasis}.

Any symmetric spinor can be decomposed into a symmetrized product of 1-index spinors, 
which are then called its {\em principal spinors}. 
This is a consequence of the fact that the field of complex numbers is algebraically closed, see Proposition (3.5.18) in \cite{PR1}.
We illustrate this with a simple example.
Using a basis $\xi^{A},\eta^{A}$, with $\xi_{A}\eta^{A}=\chi\neq 0$, let $\zeta^{A}=\xi^{A}+z \eta^{A}$, and consider 
$\varphi_{AB}=\varphi_{(AB)}$. Then 
\begin{equation}
 p(z) := \varphi_{AB}\zeta^{A}\zeta^{B} = \varphi_{0}+2\varphi_{1} z + \varphi_{2} z^{2} = (z-z_{+})(z-z_{-}),
\end{equation}
where we assume $\varphi_{2}\neq 0$, and the roots $z_{\pm}$ are
\begin{equation}
 z_{\pm} = \varphi^{-1}_{2}(-\varphi_{1} \pm (\varphi^{2}_{1}-\varphi_{0}\varphi_{2})^{1/2}).
\end{equation}
Then the principal spinors of $\varphi_{AB}$ are $\zeta^{\pm}_{A}=\xi_{A}+z_{\pm}\eta_{A}$, and
\begin{equation}\label{PrincipalSpinors}
 \varphi_{AB} = \chi^{-2} \varphi_{2} \zeta^{+}_{(A}\zeta^{-}_{B)}.
\end{equation}
It is important to note, however, that the principal spinors may be complex. 
This is of course always true in Riemannian and Lorentz signature, but in split signature, a 
symmetric spinor can be real while its principal spinors are complex.
This depends on whether the roots $z_{\pm}$ are real or complex, which in turn depends on the sign 
of $-\frac{1}{2}\varphi_{AB}\varphi^{AB}=\chi^{-2}(\varphi^{2}_1-\varphi_{0}\varphi_{2})$.
If $\varphi^{2}_1-\varphi_{0}\varphi_{2}\geq 0$, then $z_{\pm}$ are real and the principal spinors $\zeta^{\pm}_{A}$ are real. 
If $\varphi^{2}_1-\varphi_{0}\varphi_{2}< 0$, then $z_{\pm}$ are complex conjugates, and $\zeta^{+}_{A} = \overline{(\zeta^{-}_{A})}$.

\section{Further details on the covariant formalism}\label{Appendix:CovariantFormalism}

Here we give more details on the covariant formalism introduced in section \ref{sec:CovariantFormalism}.

Let $P_{\rm spin}$ denote the principal fibre bundle corresponding to the spin structure in the conformal manifold. 
Let $G$ be the associated structure group: $G={\rm Spin}(p,q)\times\mathbb{R}^{+}$, 
where ${\rm Spin}(p,q)$ is one of the three groups \eqref{spin13}--\eqref{spin22} (and $\mathbb{R}^{+}$ accounts 
for the conformal rescalings of the metric).
Let $V=(\otimes^{k}\mathbb{K}^{2})\otimes(\otimes^{k'}\bar{\mathbb{K}}^{2})\otimes(\otimes^{l}\mathbb{K}^{2*})
\otimes(\otimes^{l'}\bar{\mathbb{K}}^{2*})$, where $\mathbb{K}=\mathbb{C}$ for signature $(+---)$ and $(++++)$, 
and $\mathbb{K}=\mathbb{R}$ for $(++--)$ (in which case ``$\bar{\mathbb{R}}$'' is simply another copy of $\mathbb{R}$). 
The spinor bundles are the associated vector bundles
\begin{equation}\label{spinbundle}
 \mathbb{S}^{B...B'...}_{C...C'...} = P_{\rm spin}\times_{G} V
\end{equation}
where $\mathbb{S}^{B...B'...}_{C...C'...} $ has $k$ unprimed and $k'$ primed indices in the upper position, and 
$l$ unprimed and $l'$ primed indices in the lower position; and $\times_{G}$ denotes the ``natural'' representation of $G$ 
on $V$ \footnote{More precisely, the tensor product of the natural representations. Also, we are being a bit sloppy since $G$ 
also includes the conformal rescalings $\mathbb{R}^{+}$, so this representation is understood to have also some 
definite conformal weight. But this is not important here since it is accounted for anyway in \eqref{representationGo}.}.
For only one index, we also use the notation $\mathbb{S}=\mathbb{S}^{B}$, $\mathbb{S}'=\mathbb{S}^{B'}$, 
$\mathbb{S}^{*}=\mathbb{S}_{B}$, $\mathbb{S}'^{*}=\mathbb{S}_{B'}$.

A spinor field is a section of \eqref{spinbundle}.
But in this work we also need {\em weighted} spinor fields, as a consequence of the fact that we have the `gauge freedom' 
\eqref{gaugefreedomK} and \eqref{ConformalRescaling}. 
In section \ref{sec:CovariantFormalism} we mentioned that the gauge group associated to this freedom is $G_o$. 
We see that the transformation law \eqref{weightedfields} corresponds simply to the representation of $G_o$ on $V$ 
given by $\rho_{r,r';w}:G_o\to{\rm GL}(V)$, 
\begin{equation}\label{representationGo}
 [\rho_{r,r';w}(\lambda,\mu,\Omega)\varphi]^{B...B'...}_{C...C'...} = \lambda^{r}\mu^{r'}\Omega^{w}\varphi^{B...B'...}_{C...C'...}.
\end{equation}
The principal bundle $B$ over $M$ whose structure group is $G_o$ can be understood as 
a `reduction' of $P_{\rm spin}$ (in the sense that $G_o\subset G$), and, similarly to \eqref{spinbundle}, 
we now have the weighted spinor bundles
\begin{equation}\label{weightedbundles}
 \mathbb{S}^{B...B'...}_{C...C'...}(r,r';w) = B\times_{\rho_{r,r';w}} V.
\end{equation}
A spinor field of type $(r,r';w)$ (Definition \ref{Def:weightedsfields}) is a section of this bundle. 
For example, the spinor fields $\a_{A}$, $\beta_{A}$ associated to a given almost para-Hermitian structure $K$, 
and the scalar field $\phi=\a_{A}\beta^{A}$, are all sections of \eqref{weightedbundles}, for different values 
of $r,r',w$. We give the specific weights in table \ref{Table:weights}.
For convenience we mention that objects without indices, i.e. weighted {\em scalar} fields, are sections of the {\em line} bundles
\begin{equation}\label{weightedlinebundles}
 \mathbb{S}(r,r';w) = B\times_{\rho_{r,r';w}} V
\end{equation}
where $V=\mathbb{C}$ or $V=\mathbb{R}$.

The covariant derivative introduced in Definition \ref{Def:GHPconformalCV} is an operator on sections of \eqref{weightedbundles}, 
whose construction can be understood as follows. 
First, use the spinors $\alpha^{A},\beta^{A}$ associated to $K$ as a basis:
$\varepsilon^{A}_{\bf A}\equiv(\alpha^{A},\beta^{A})$ (where ${\bf A}=0,1$ is 
understood as a concrete index), with $\alpha_{A}\beta^{A}=\phi$. 
The dual frame is $\varepsilon^{\bf A}_{A}=(-\phi^{-1}\beta_{A},\phi^{-1}\alpha_{A})$.
The natural Weyl connection ${}^{\rm w}\nabla_{a}$ induced by $K$ defines a local 
connection 1-form by ${}^{\rm w}\nabla_{a}\varepsilon^{B}_{\bf B}= {}^{\rm w}\omega_{a{\bf B}}{}^{\bf C}\varepsilon^{B}_{\bf C}$, 
or equivalently 
\begin{equation}
  {}^{\rm w}\omega_{a{\bf B}}{}^{\bf C} = \varepsilon^{\bf C}_{B} \; {}^{\rm w}\nabla_{a}\varepsilon^{B}_{\bf B}.
\end{equation}
Under $G_{o}$, some parts of this object transform covariantly, and some other do not.
The components that transform covariantly are ${}^{\rm w}\omega_{a 0}{}^{1}$ and  ${}^{\rm w}\omega_{a 1}{}^{0}$.
For the rest, one can check that
\begin{align*}
 {}^{\rm w}\omega_{a 0}{}^{0} \xrightarrow{G_{o}}{}& {}^{\rm w}\omega_{a 0}{}^{0} +w_{0}\Omega^{-1}\partial_{a}\Omega 
 +\lambda^{-1}\partial_{a}\lambda, \\
 {}^{\rm w}\omega_{a 1}{}^{1} \xrightarrow{G_{o}}{}& {}^{\rm w}\omega_{a 1}{}^{1} +w_{1}\Omega^{-1}\partial_{a}\Omega 
 +\mu^{-1}\partial_{a}\mu.
\end{align*}
Therefore, defining the 1-forms
\begin{subequations}
\begin{align}
 L_{a}:={}& -{}^{\rm w}\omega_{a 0}{}^{0} - w_{0}f_{a}, \\
 M_{a}:={}& -{}^{\rm w}\omega_{a 1}{}^{1} - w_{1}f_{a},
\end{align}
\end{subequations}
and recalling the Lee-form $f_{a}$, we have
\begin{align*}
 f_{a} \xrightarrow{G_{o}}{}& f_{a} -\Omega^{-1}\partial_{a} \Omega, \\
 L_{a} \xrightarrow{G_{o}}{}& L_{a} -\lambda^{-1}\partial_{a}\lambda, \\
 M_{a} \xrightarrow{G_{o}}{}& M_{a} - \mu^{-1}\partial_{a}\mu,
\end{align*}
so $\psi_{a}:=(f_{a},L_{a},M_{a})$ can be thought of as the connection 1-form associated to the gauge symmetry $G_{o}$.
The explicit expression of $L_{a}$ and $M_{a}$ in terms of the spinors $\alpha_{A},\beta_{A}$ is given in 
formulas \eqref{connection1}, \eqref{connection2}. For $f_{a}$, a calculation shows the following:

\begin{proposition}
The Lee form \eqref{LeeForm} induced by an almost-para Hermitian structure $K$ is given by
\begin{equation}
 f_{AA'} = \phi^{-2}(\a_{A}\b_{C}\a^{B}\nabla_{BA'}\b^{C}+\b_{A}\a_{C}\b^{B}\nabla_{BA'}\a^{C}).
\end{equation}
where $\a_{A}$ and $\beta_{A}$ are the (projective) spinor fields representing $K$, with $\phi=\alpha_{A}\beta^{A}$.
\end{proposition}

Now, consider an element $u=(x,y,z)\in{\rm Lie}(G_{o})$ in the Lie algebra of $G_{o}$.
The representation \eqref{representationGo} of $G_{o}$ induces a representation $\rho'_{r,r';w}$
of ${\rm Lie}(G_{o})$, given by
\begin{equation}
 \rho'_{r,r';w}(u) = rx+r'y+wz.
\end{equation}
Since the connection 1-form $\psi_{a}$ is valued in ${\rm Lie}(G_{o})$, we have 
$\rho'_{r,r';w}(\psi_{a})=rL_{a}+r'M_{a}+wf_{a}$.
Combining this with the Weyl connection ${}^{\rm w}\nabla_{a}$, the induced covariant derivative is
\begin{equation}\label{CVassociatedbundle}
 {}^{\rm w}\nabla_{a} + \rho'_{r,r';w}(\psi_{a}) =: \mathcal{C}_{a}
\end{equation}
which is simply \eqref{cvC}.
For convenience we give the explicit formula for \eqref{CVassociatedbundle} in terms of an arbitrary 
Levi-Civita connection $\nabla_{AA'}$, when acting on a spinor field $\varphi^{BB'}_{CC'}$ of type $(r,r';w)$:
\begin{align}
\nonumber \mathcal{C}_{AA'}\varphi^{BB'}_{CC'}
 ={}& \nabla_{AA'}\varphi^{BB'}_{CC'} +(wf_{AA'}+rL_{AA'}+r'M_{AA'})\varphi^{BB'}_{CC'} \\
& + \epsilon_{A}{}^{B}f_{A'X}\varphi^{XB'}_{CC'}+\epsilon_{A'}{}^{B'}f_{AX'}\varphi^{BX'}_{CC'} 
- f_{A'C}\varphi^{BB'}_{AC'} - f_{AC'}\varphi^{BB'}_{CA'}.
 \label{cvCexplicit}
\end{align}

\begin{proposition}
We have the general identities:
\begin{subequations}
\begin{align}
 \mathcal{C}_{a}g_{bc} ={}& 0 = \mathcal{C}_{a}\e_{BC}=\mathcal{C}_{a}\e_{B'C'}, \\
 \mathcal{C}_{AA'}\a^{B} ={}& \phi^{-2}(\a^{C}\a^{D}\nabla_{CA'}\a_{D})\b_{A}\b^{B}, \label{Calpha} \\
 \mathcal{C}_{AA'}\b^{B} ={}& \phi^{-2}(\b^{C}\b^{D}\nabla_{CA'}\b_{D})\a_{A}\a^{B}, \label{Cbeta} \\
 \mathcal{C}_{AA'}\phi ={}& 0. \label{Cchi}
\end{align}
\end{subequations}
\end{proposition}

\begin{remark}
Recalling the definition \eqref{CProjections} of $\mathcal{C}_{A'}, \tilde{\mathcal{C}}_{A'}$, 
we deduce from \eqref{Calpha} and \eqref{Cbeta} the general identities
\begin{equation}
 \mathcal{C}_{A'}\alpha^{B} = 0, \qquad \tilde{\mathcal{C}}_{A'}\beta^{B} = 0. \label{parallelspinors}
\end{equation}
\end{remark}

\section{Curvature of $\mathcal{C}_{a}$}\label{Appendix:curvature}

In this appendix we describe some properties of the curvature of the connection \eqref{CVassociatedbundle}, 
that we need in the main text. 
The curvature of $\mathcal{C}_{a}$ is defined by the commutator $[\mathcal{C}_{a},\mathcal{C}_{b}]$. 
Since this is skew-symmetric in $ab$, the decomposition \eqref{SDdecomposition} implies that $[\mathcal{C}_{a},\mathcal{C}_{b}]$ 
splits into SD and ASD pieces; this will be applied below.

Let $v^{AA'}$ be a vector field with weights $(r,r';w)$. Then by definition
\begin{equation}
 \mathcal{C}_{b}v^{DD'} = \nabla_{b}v^{DD'} +\psi_{b}v^{DD'}+W_{bC}{}^{D}v^{CD'}+\tilde{W}_{bC'}{}^{D'}v^{C'D},
\end{equation}
where for convenience we put
\begin{align}
 \psi_{b} ={}& w f_{b} + r L_{b} + r' M_{b}, \label{internalconnection} \\
 W_{bC}{}^{D} ={}& f_{B'C}\epsilon_{B}{}^{D}, \\
 \tilde{W}_{bC'}{}^{D'} ={}& f_{BC'}\epsilon_{B'}{}^{D'}.
\end{align}
Applying another covariant derivative $\mathcal{C}_{a}$ and taking the commutator, after a short calculation we find
\begin{equation}\label{curvature1}
[\mathcal{C}_{a},\mathcal{C}_{b}]v^{DD'} = [\nabla_{a},\nabla_{b}]v^{DD'} + (2\nabla_{[a}\psi_{b]})v^{DD'} 
 + F_{abC}{}^{D}v^{CD'} + \tilde{F}_{abC'}{}^{D'}v^{C'D}
\end{equation}
where 
\begin{align}
 F_{abC}{}^{D} :={}& 2\nabla_{[a}W_{b]C}{}^{D} + 2W_{[a|E|}{}^{D}W_{b]C}{}^{E}, \\
 \tilde{F}_{abC'}{}^{D'} :={}& 2\nabla_{[a}\tilde{W}_{b]C'}{}^{D'} + 2\tilde{W}_{[a|E'|}{}^{D'}\tilde{W}_{b]C'}{}^{E'}.
\end{align}
We can see three contributions to \eqref{curvature1}: the curvature of the Levi-Civita connection, 
the curvature of the ``internal'' connection $\psi_{b}$, and the curvature of the Weyl connection, which 
is encoded in the two pieces $F_{abC}{}^{D}$ and $\tilde{F}_{abC'}{}^{D'}$. 

Now, a straightforward computation shows that
\begin{equation}\label{decompositionCurvatureC}
 [\mathcal{C}_{a},\mathcal{C}_{b}] = \epsilon_{AB}\mathcal{C}_{C(A'}\mathcal{C}_{B')}{}^{C}
 + \epsilon_{A'B'}\mathcal{C}_{C'(A}\mathcal{C}_{B)}{}^{C'},
\end{equation}
and similarly, all the pieces in the right hand side of \eqref{curvature1} decompose into SD and ASD pieces.
Thus the SD and ASD pieces in \eqref{curvature1} can be simply obtained as $\frac{1}{2}\epsilon^{AB}[\mathcal{C}_{a},\mathcal{C}_{b}]$
and $\frac{1}{2}\epsilon^{A'B'}[\mathcal{C}_{a},\mathcal{C}_{b}]$. We find:
\begin{align}
 \mathcal{C}_{A'(A}\mathcal{C}_{B)}{}^{A'}v^{DD'} ={}& \Box_{AB}v^{DD'} + (\nabla_{A'(A}\psi^{A'}_{B)})v^{DD'} 
 +F_{ABC}{}^{D}v^{CD'} + \tilde{G}_{ABC'}{}^{D'}v^{C'D}, \label{ASDcurvatureC} \\
 \mathcal{C}_{A(A'}\mathcal{C}_{B')}{}^{A}v^{DD'} ={}& \Box_{A'B'}v^{DD'} + (\nabla_{A(A'}\psi^{A}_{B')})v^{DD'} 
 +\tilde{F}_{A'B'C'}{}^{D'}v^{C'D} + G_{A'B'C}{}^{D}v^{CD'}, \label{SDcurvatureC}
\end{align}
where $\Box_{AB}$ and $\Box_{A'B'}$ are ordinary spinor curvature operators, see \cite[Section 4.9]{PR1},
and we defined
\begin{align}
 F_{ABC}{}^{D} :={}& \epsilon_{(A}{}^{D}\left[ \nabla_{B)C'}f_{C}{}^{C'}+f_{B)C'}f_{C}{}^{C'} \right], \label{FASD} \\
 G_{A'B'C}{}^{D} :={}& -\nabla_{(A'}{}^{D}f_{B')C} + f_{(A'}{}^{D}f_{B')C}, \label{FSD} \\
 \tilde{F}_{A'B'C'}{}^{D'} :={}& \epsilon_{(A'}{}^{D'}\left[ \nabla_{B')C}f_{C'}{}^{C}+f_{B')C}f_{C'}{}^{C} \right], \label{FtildeSD} \\
 \tilde{G}_{ABC'}{}^{D'} :={}& -\nabla_{(A}{}^{D'}f_{B)C'} + f_{(A}{}^{D'}f_{B)C'}. \label{FtildeASD}
\end{align}
The spinors $G,F$ are respectively the SD and ASD pieces of $F_{abC}{}^{D}$, and 
$\tilde{G},\tilde{F}$ those of $\tilde{F}_{abC'}{}^{D'}$.
Likewise, $\nabla_{A(A'}\psi^{A}_{B')}$ and $\nabla_{A'(A}\psi^{A'}_{B)}$ are the SD and ASD pieces 
of $2\nabla_{[a}\psi_{b]}$.

\begin{remark}
For weighted {\em scalars} $\Phi$, formulas \eqref{ASDcurvatureC}-\eqref{SDcurvatureC} are much simpler:
\begin{equation}\label{CCscalars}
 \mathcal{C}_{A'(A}\mathcal{C}_{B)}{}^{A'}\Phi = (\nabla_{A'(A}\psi^{A'}_{B)})\Phi, 
 \qquad \mathcal{C}_{A(A'}\mathcal{C}_{B')}{}^{A}\Phi = (\nabla_{A(A'}\psi^{A}_{B')})\Phi.
\end{equation}
\end{remark}

Using \eqref{internalconnection} and the expressions \eqref{connection1}-\eqref{connection2} for $L_{a},M_{a}$, we note that 
\begin{equation}
 \psi_{a} = (w-r w_{0}-r'w_{1})f_{a} + r'\psi^{(\beta)}_{a} + r\psi^{(\alpha)}_{a},
\end{equation}
where 
\begin{align}
 \psi^{(\beta)}_{a} ={}& -\phi^{-1}\alpha_{B}\nabla_{AA'}\beta^{B}-\phi^{-1}\alpha_{A}f_{A'B}\beta^{B}, \\
 \psi^{(\alpha)}_{a} ={}& \phi^{-1}\beta_{B}\nabla_{AA'}\alpha^{B}+\phi^{-1}\beta_{A}f_{A'B}\alpha^{B}.
\end{align}
Note also that $\psi^{(\alpha)}_{a}$ can be rewritten as $\psi^{(\alpha)}_{a}=-\psi^{(\beta)}_{a}-\nabla_{a}\log\phi-f_{a}$; 
therefore, the curvature of the ``internal'' connection can be computed in terms of the curvatures of $f_{a}$ and $\psi^{(\beta)}_{a}$:
\begin{equation}
 \nabla_{[a}\psi_{b]} = (w-r-r w_{0}-r'w_{1})\nabla_{[a}f_{b]}+(r'-r)\nabla_{[a}\psi^{(\beta)}_{b]}.
\end{equation}
For later use, we note in particular the following:
\begin{equation}
 \beta^{A}\beta^{B}\nabla_{A'(A}\psi_{B)}^{A'} = (w-r-r w_{0}-r'w_{1})\beta^{A}\beta^{B}\nabla_{A'(A}f_{B)}^{A'}
 +(r'-r)\beta^{A}\beta^{B}\nabla_{A'(A}\psi_{B)}^{(\beta)A'}. \label{bbdpsi}
\end{equation}

\subsection*{Some identities when $\beta^{A}$ is shear-free}

\begin{proposition}
Suppose that $\beta^{A}$ is shear-free, then we have the following identities:
\begin{align}
 & \Phi_{A'B'CD}\beta^{C}\beta^{D} +G_{A'B'CD} \beta^{C}\beta^{D} = 0, \label{identityPhi} \\
 & \Psi_{ABCD}\beta^{C}\beta^{D}+2\Lambda\beta_{A}\beta_{B} + F_{ABCD}\beta^{C}\beta^{D}=0. \label{identityPsi}
\end{align}
\end{proposition}

\begin{proof}
These identities follow from \eqref{SDcurvatureC} and \eqref{ASDcurvatureC} 
by applying them to $\beta^{D}$, using $\mathcal{C}_{a}\beta^{B}=0$, and contracting with $\beta_{D}$.
\end{proof}

\begin{proposition}
If $\beta^{A}$ is shear-free, we have:
\begin{align}
 & \beta^{A}\beta^{B}\Phi_{ABC'D'}+\beta^{A}\beta^{B}\tilde{G}_{ABC'D'} 
 = \tfrac{1}{2}\epsilon_{C'D'}\beta^{A}\beta^{B}\nabla_{A'(A}f_{B)}^{A'}, \label{identityPhiGtilde} \\
 & \Psi_{ABCD}\beta^{B}\beta^{C}\beta^{D} = \tfrac{1}{2}\beta_{A}\beta^{B}\beta^{C}\nabla_{B'(B}f_{C)}^{B'}
 = -\tfrac{1}{3}\beta_{A}\beta^{B}\beta^{C}\nabla_{B'(B}\psi_{C)}^{(\beta)B'}. \label{identityPsibbb}
\end{align}
\end{proposition}

\begin{proof}
To prove \eqref{identityPhiGtilde}, we use \eqref{identityPhi}:
\begin{align}
\nonumber \beta^{A}\beta^{B}\Phi_{ABC'D'}+\beta^{A}\beta^{B}\tilde{G}_{ABC'D'}  ={}& 
 -\beta^{A}\beta^{B}G_{C'D'AB} + \beta^{A}\beta^{B}\tilde{G}_{ABC'D'} \\
\nonumber ={}& \beta^{A}\beta^{B}\nabla_{A(D'}f_{C')A} - \beta^{A}\beta^{B}\nabla_{AD'}f_{C'A} \\
\nonumber ={}& -\beta^{A}\beta^{B}\nabla_{A[D'}f_{C']A} \\
 ={}& \tfrac{1}{2}\epsilon_{C'D'}\beta^{A}\beta^{B}\nabla_{A'(A}f_{B)}^{A'}.
\end{align}

Now let us prove \eqref{identityPsibbb}. Contracting \eqref{identityPsi} with $\beta^{B}$, we find
\begin{equation}
 \Psi_{ABCD}\beta^{B}\beta^{C}\beta^{D}+F_{ABCD}\beta^{B}\beta^{C}\beta^{D}=0. \label{idPsibbbaux}
\end{equation}
The second term is
\begin{equation*}
 F_{ABCD}\beta^{B}\beta^{C}\beta^{D} 
 = \epsilon_{(A|D|}\left[ \nabla_{B)C'}f_{C}{}^{C'}+f_{B)C'}f_{C}{}^{C'} \right]\beta^{B}\beta^{C}\beta^{D}
 =-\tfrac{1}{2}(\beta^{B}\beta^{C}\nabla_{B'(B}f_{C)}^{B'})\beta_{A}
\end{equation*}
which establishes the first equality in \eqref{identityPsibbb}. 
Now, applying \eqref{ASDcurvatureC} to $\beta^{D}$ and contracting with $\beta^{A}\beta^{B}$, we get
\begin{equation*}
 0=\Psi_{ABCD}\beta^{A}\beta^{B}\beta^{C}+(\beta^{B}\beta^{C}\nabla_{B'(B}\psi_{C)}^{(\beta)B'})\beta_{D}
 +F_{ABCD}\beta^{A}\beta^{B}\beta^{C}.
\end{equation*}
From \eqref{idPsibbbaux} we deduce that 
$\Psi_{ABCD}\beta^{A}\beta^{B}\beta^{C}=-F_{DABC}\beta^{A}\beta^{B}\beta^{C}$, 
where we used $\Psi_{ABCD}=\Psi_{(ABCD)}$. Thus,
\begin{equation}
 0 = (\beta^{B}\beta^{C}\nabla_{B'(B}\psi_{C)}^{(\beta)B'})\beta_{D} -F_{DABC}\beta^{A}\beta^{B}\beta^{C} 
 +F_{ABCD}\beta^{A}\beta^{B}\beta^{C}.
\end{equation}
The last term in the right hand side is
\begin{equation*}
 F_{ABCD}\beta^{A}\beta^{B}\beta^{C} 
 = \epsilon_{(A|D|}\left[ \nabla_{B)C'}f_{C}{}^{C'}+f_{B)C'}f_{C}{}^{C'} \right]\beta^{A}\beta^{B}\beta^{C} 
 = (\beta^{B}\beta^{C}\nabla_{B'(B}f_{C)}^{B'})\beta_{D},
\end{equation*}
which shows that $F_{ABCD}\beta^{A}\beta^{B}\beta^{C} = -2 F_{DABC}\beta^{A}\beta^{B}\beta^{C}$. 
Therefore
\begin{equation*}
 (\beta^{B}\beta^{C}\nabla_{B'(B}\psi_{C)}^{(\beta)B'})\beta_{D} = 3F_{DABC}\beta^{A}\beta^{B}\beta^{C}
 = -3\Psi_{ABCD}\beta^{A}\beta^{B}\beta^{C},
\end{equation*}
and the last equality in \eqref{identityPsibbb} follows.
\end{proof}

\begin{remark}
Using \eqref{identityPsibbb}, when $\beta^{A}$ is shear-free we can deduce a useful identity for \eqref{bbdpsi}, 
for arbitrary weights $w,r,r'$ in \eqref{internalconnection}:
\begin{equation}\label{bbdpsi2}
\beta^{A}\beta^{B}\nabla_{A'(A}\psi_{B)}^{A'}  =  -2\phi^{-1}(w-r-\tfrac{3}{2}(r'-r)-rw_{0}-r'w_{1})
\Psi_{ABCD}\alpha^{A}\beta^{B}\beta^{C}\beta^{D}.
\end{equation}
\end{remark}

\begin{lemma}\label{Lemma:CC0}
Suppose that $\beta^{A}$ is shear-free. Then the equation 
\begin{equation}
 \mathcal{C}_{A'}\mathcal{C}^{A'} = 0
\end{equation}
is satisfied when acting on {\em any} weighted scalar/spinor/tensor field, 
if and only if $\beta^{A}$ is a repeated principal spinor of the ASD Weyl curvature.
\end{lemma}

\begin{proof}
Consider first a primed spinor field $\pi^{A'}$, with arbitrary weights $(r,r';w)$. 
Applying \eqref{ASDcurvatureC} to $\pi^{A'}$ and contracting with $\beta^{A}\beta^{B}$, we have
\begin{align*}
 \mathcal{C}_{A'}\mathcal{C}^{A'}\pi^{D'} ={}& \beta^{A}\beta^{B}\Box_{AB}\pi^{D'} + 
 (\beta^{A}\beta^{B}\nabla_{A'(A}\psi^{A'}_{B)})\pi^{D'} + \beta^{A}\beta^{B}\tilde{G}_{ABC'}{}^{D'}\pi^{C'} \\
 ={}& (\beta^{A}\beta^{B}\Phi_{ABC'}{}^{D'} + \beta^{A}\beta^{B}\tilde{G}_{ABC'}{}^{D'})\pi^{C'} 
 +(\beta^{A}\beta^{B}\nabla_{A'(A}\psi^{A'}_{B)})\pi^{D'} \\
 ={}& (\tfrac{1}{2}\beta^{A}\beta^{B}\nabla_{A'(A}f_{B)}^{A'})\pi^{D'}+(\beta^{A}\beta^{B}\nabla_{A'(A}\psi^{A'}_{B)})\pi^{D'} \\
 ={}& (w+\tfrac{1}{2}+\tfrac{1}{2}r-rw_{0}-\tfrac{3}{2}r'-r'w_{1})(\beta^{A}\beta^{B}\nabla_{A'(A}f^{A'}_{B)})\pi^{D'},
\end{align*}
where in the third line we replaced \eqref{identityPhiGtilde}, and in the fourth we used \eqref{bbdpsi} and 
the second equality in \eqref{identityPsibbb}. 
Therefore, we see that $\mathcal{C}_{A'}\mathcal{C}^{A'}\pi^{D'}=0$ for any $\pi^{D'}$ if and only if 
$\beta^{A}\beta^{B}\nabla_{A'(A}f^{A'}_{B)}=0$, which, using the first equality in \eqref{identityPsibbb}, 
is true if and only if $\Psi_{ABCD}\beta^{B}\beta^{C}\beta^{D}=0$.

Consider now an unprimed spinor field $\varphi^{A}$ with weights $(r,r';w)$. This can be written as
$\varphi^{A} = a\alpha^{A}+b\beta^{A}$.  Using $\mathcal{C}_{a}\beta^{B}=0$ and the first equation 
in \eqref{parallelspinors}, we have
\begin{equation}
 \mathcal{C}_{A'}\mathcal{C}^{A'}\varphi^{D} 
 = \alpha^{D} \mathcal{C}_{A'}\mathcal{C}^{A'}a + \beta^{D} \mathcal{C}_{A'}\mathcal{C}^{A'}b.
\end{equation}
Thus, we see that $\mathcal{C}_{A'}\mathcal{C}^{A'}\varphi^{D} = 0$ on arbitrary weighted unprimed spinor fields $\varphi^{D}$
if and only if $\mathcal{C}_{A'}\mathcal{C}^{A'}\Phi=0$ for arbitrary weighted scalar fields. 
From \eqref{CCscalars}, we get
\begin{equation}
 \mathcal{C}_{A'}\mathcal{C}^{A'}\Phi = (\beta^{A}\beta^{B}\nabla_{A'(A}\psi^{A'}_{B)})\Phi.
\end{equation}
Using now \eqref{bbdpsi} and \eqref{identityPsibbb}, we see that $\mathcal{C}_{A'}\mathcal{C}^{A'}\Phi=0$ 
if and only if $\Psi_{ABCD}\beta^{B}\beta^{C}\beta^{D}=0$.
\end{proof}

\end{document}